\documentclass[11pt]{article}

\usepackage[margin=1in]{geometry}
\usepackage[T1]{fontenc}
\usepackage[utf8]{inputenc}
\usepackage{lmodern}
\usepackage{microtype}
\usepackage{authblk}
\usepackage{amsmath,amssymb,amsthm,mathtools}
\usepackage{bm}  
\usepackage{xcolor}
\definecolor{BlazeOrange}{HTML}{FE621D}
\definecolor{NavyBlue}{HTML}{00008B}
\definecolor{TeaGreen}{HTML}{CADBC0}
\definecolor{DarkAmaranth}{HTML}{820933}
\definecolor{DeepSpaceBlue}{HTML}{13293D}
\usepackage{tikz}
\usetikzlibrary{arrows.meta,positioning}
\usepackage{enumitem}
\usepackage{aliascnt}
\usepackage{array,booktabs,multirow,tabularx}
\newcounter{algorithm}
\newcommand{\SetKwInput}[2]{}
\newcommand{\SetKwComment}[3]{}
\newcommand{\SetKw}[2]{}
\newcommand{\KwInput}[1]{\noindent\textbf{Input: }#1\par}
\newcommand{\KwOutput}[1]{\noindent\textbf{Output: }#1\par}
\newcommand{\tcp}[1]{}
\newcommand{\Return}{\textbf{return} }
\newcommand{\While}[2]{%
  \noindent\textbf{while }#1\textbf{ do}\par
  \begingroup\leftskip=1.5em #2\par\endgroup
  \noindent\textbf{end while}\par}
\newcommand{\If}[2]{%
  \noindent\textbf{if }#1\textbf{ then }#2\par}
\newenvironment{algorithm}[1][]{%
  \par\medskip\noindent\begin{minipage}{\linewidth}\small
  \hrule\smallskip
  \renewcommand{\caption}[1]{%
    \refstepcounter{algorithm}%
    \noindent\textbf{Algorithm~\thealgorithm: ##1}\par\smallskip}
  \renewcommand{\;}{\par}%
}{%
  \smallskip\hrule\end{minipage}\par\medskip}
\usepackage{hyperref}
\usepackage[capitalize,nameinlink]{cleveref}
\crefname{algorithm}{algorithm}{algorithms}
\Crefname{algorithm}{Algorithm}{Algorithms}

\allowdisplaybreaks
\hypersetup{
  colorlinks=true,
  linkcolor=NavyBlue,
  citecolor=DarkAmaranth,
  urlcolor=NavyBlue,
  linktoc=all
}

\newtheorem{theorem}{Theorem}[section]
\newaliascnt{lemma}{theorem}
\newtheorem{lemma}[lemma]{Lemma}
\aliascntresetthe{lemma}
\newaliascnt{corollary}{theorem}
\newtheorem{corollary}[corollary]{Corollary}
\aliascntresetthe{corollary}
\newaliascnt{proposition}{theorem}

\aliascntresetthe{proposition}
\newaliascnt{definition}{theorem}
\newtheorem{definition}[definition]{Definition}
\aliascntresetthe{definition}
\newaliascnt{remark}{theorem}

\aliascntresetthe{remark}
\newaliascnt{fact}{theorem}

\aliascntresetthe{fact}

\newcommand{\bits}{\{0,1\}}
\newcommand{\cA}{\mathcal{A}}
\newcommand{\cB}{\mathcal{B}}
\newcommand{\cC}{\mathcal{C}}
\newcommand{\cD}{\mathcal{D}}
\newcommand{\cE}{\mathcal{E}}
\newcommand{\cH}{\mathcal{H}}
\newcommand{\cK}{\mathcal{K}}
\newcommand{\cM}{\mathcal{M}}
\newcommand{\cR}{\mathcal{R}}

\newcommand{\cU}{\mathcal{U}}
\newcommand{\cX}{\mathcal{X}}
\newcommand{\cY}{\mathcal{Y}}
\newcommand{\Dom}{\operatorname{Dom}}
\newcommand{\Unif}{\operatorname{Unif}}
\newcommand{\poly}{\operatorname{poly}}
\newcommand{\negl}{\operatorname{negl}}
\newcommand{\Tr}{\operatorname{Tr}}

\newcommand{\Supp}{\operatorname{Supp}}
\newcommand{\Real}{\mathsf{Real}}
\newcommand{\Learn}{\mathsf{Learn}}
\newcommand{\Heavy}{\mathsf{Heavy}}

\newcommand{\IntKR}{\mathsf{IntKR}}
\newcommand{\Adv}{\mathsf{Adv}}
\newcommand{\Gen}{\mathsf{Gen}}
\newcommand{\Enc}{\mathsf{Enc}}
\newcommand{\Dec}{\mathsf{Dec}}
\newcommand{\IND}{\mathsf{IND}}
\newcommand{\pk}{\mathsf{pk}}
\newcommand{\sk}{\mathsf{sk}}

\newcommand{\ket}[1]{\lvert #1\rangle}
\newcommand{\bra}[1]{\langle #1\rvert}
\newcommand{\proj}[1]{\ket{#1}\!\bra{#1}}

\newcommand{\ceil}[1]{\left\lceil #1\right\rceil}
\newcommand{\changed}[1]{\textcolor{DarkAmaranth}{#1}}
\newcommand{\mrchanged}[1]{\changed{#1}}

\newenvironment{hybrid}[1]{%
  \par\medskip
  \noindent\textbf{#1}
  \begin{enumerate}[leftmargin=2.5em,label=\arabic*.,itemsep=0.25em,topsep=0.35em]
}{%
  \end{enumerate}
  \medskip
}

\SetKwInput{KwInput}{Input}
\SetKwInput{KwOutput}{Output}
\SetKwComment{tcp}{\(\triangleright\) }{}
\SetKw{Return}{return}

\title{Towards the Impossibility of Imperfectly Complete Key Agreement in the QROM}

\author[1]{Fuyuki Kitagawa\thanks{\texttt{fuyuki.kitagawa@ntt.com}}}
\author[1]{Ryo Nishimaki\thanks{\texttt{ryo.nishimaki@ntt.com}}}
\author[2]{Ági Villányi\thanks{\texttt{agivilla@mit.edu} Part of this work was done while visiting NTT Social Informatics Laboratories as an internship.}}
\author[1]{Takashi Yamakawa\thanks{\texttt{takashi.yamakawa@ntt.com}}}
\affil[1]{NTT Social Informatics Laboratories, Tokyo, Japan}
\affil[2]{Massachusetts Institute of Technology, Cambridge, MA, USA}
\date{}

\begin{document}
\maketitle

\begin{abstract}
We make progress towards the impossibility of quantum-computation,
classical-communication (QCCC) key agreement by giving the first
unconditional polynomial-query attacks that tolerate imperfect
completeness in the following settings. First, we give a quantum attack on
protocols with arbitrarily many rounds in which both parties' oracle
queries and communication are classical before the final round, while local
computation may be quantum throughout, and the final round may involve
quantum oracle queries and one quantum message. Second, we give a classical attack on
constant-round QCCC protocols in which Alice has classical oracle access
and Bob may make quantum queries throughout. In both settings, the attacker is computationally unbounded and makes
\(\poly(\lambda)\) queries to recover the key with nonnegligible probability
whenever each party's total honest query bound is at most
\(\poly(\lambda)\) and the valid agreement probability is inverse-polynomial.

As consequences, we rule out query-bounded IND-CPA security for quantum
public-key encryption with classical public and secret keys and classical
messages of polynomially bounded length in the QROM in either of two
settings: (i) key generation has classical oracle access, while encryption,
decryption, and the ciphertext may be quantum; or (ii) encryption has
classical oracle access and ciphertexts are classical, while key generation
and decryption may have quantum oracle access. Both results assume
negligible correctness error and polynomial honest query complexity.
Combining the constant-round attack with the constructions of Bartusek
and Khurana (CRYPTO 2025), we also rule out constant-round oblivious state
preparation with polynomial honest query complexity and negligible
correctness error in the QROM against computationally unbounded quantum
receivers making polynomially many oracle queries.
\end{abstract}

\begingroup
\color{DeepSpaceBlue}
\hypersetup{linkcolor=DeepSpaceBlue}
\tableofcontents
\endgroup

\section{Introduction}

Merkle studied key agreement as a way for two parties with no shared secret
to communicate securely over an insecure channel.  His puzzle protocol allows
the two parties to establish a key while requiring a passive eavesdropper to
make quadratically more queries to a random oracle~\cite{Merkle78}. Subsequent work showed
that, for classical protocols, this quadratic gap is optimal under both
perfect and imperfect completeness.  Impagliazzo and Rudich ruled out a
super-polynomial gap~\cite{ImpagliazzoRudich89}, and Barak and Mahmoody proved
a tight quadratic upper bound, including when the honest parties need not
agree with certainty~\cite{BarakMahmoodyFull}.

It is natural to ask whether these results can be extended to the quantum
setting, in which parties are allowed to make queries to the oracle in
superposition. This is known as the quantum random-oracle model
(QROM)~\cite{BonehEtAl11} and quantum variants of Merkle's approach,
and key agreement in the QROM more broadly, have been studied in
several works~\cite{BrassardSalvail08,BrassardEtAl19}. The model we consider in this work was introduced by Austrin
et al.~\cite{AustrinEtAl22} as
quantum-computation, classical-communication (QCCC). In a QCCC
protocol, the parties may perform arbitrary quantum computation, retain
quantum states, and query the random oracle in superposition, but every
message they exchange is classical.
Although quantum key distribution provides secure key agreement using
quantum communication and an authenticated classical
channel~\cite{BennettBrassard84,ShorPreskill00}, it does not resolve the
question of key agreement in the QCCC model, which permits only
classical messages.

In 2018, Hosoyamada and Yamakawa asked whether QCCC key agreement in the
QROM can achieve a super-polynomial gap between
the query complexities of the honest parties and the
eavesdropper~\cite[Section~6]{HosoyamadaYamakawa18}. Assuming perfect completeness, Austrin et al. gave an unconditional quadratic-query classical attack when one party
is classical~\cite{AustrinEtAl22}. For general QCCC protocols, their Polynomial Compatibility
Conjecture yields a classical attack whose query complexity is polynomial in
the honest parties' query bounds; unconditionally, they obtained an attack using exponentially many queries. 

Li, Li, Li, and Liu subsequently obtained unconditional attacks in a sequence
of increasingly general, perfectly complete settings. They first treated
QPKE with classical-query key generation, while also proving a separate
imperfect-correctness result for oracle-free decryption~\cite{LiEtAl24}.
They then allowed logarithmically many quantum key-generation queries
unconditionally and formulated a conjecture for the general
case~\cite{LiEtAl25}. Their later work removed both restrictions for
perfectly complete QPKE with classical keys and either classical or quantum
ciphertexts~\cite{LiEtAl26PKE}. Most recently, they gave an unconditional,
round-independent attack on every perfectly complete QCCC key-agreement
protocol. The honest parties may query the oracle in superposition in every
round, while the eavesdropper makes only classical
queries~\cite{LiEtAl26ManyRound}.

What remains open is imperfectly complete QCCC key agreement when the honest
parties may make arbitrary quantum queries. A classical-query attack in
this setting would face a serious obstacle identified by Austrin et al.

Their Simulation Conjecture asserts, informally, that a classical algorithm whose
query complexity is polynomial in the oracle's input length can approximate
the acceptance probability of any quantum algorithm with the same type of
query bound, on most random oracles. It follows from
the Aaronson-Ambainis conjecture on influential variables of bounded
low-degree polynomials~\cite{AaronsonAmbainis14}. Austrin et al. showed that,
if the Simulation Conjecture is false, then there exists an imperfectly complete QCCC protocol secure against
every classical eavesdropper whose query complexity is polynomial in the
security parameter. A fully general classical
attack would therefore require substantial progress on this simulation
problem. The barrier leaves room, however, for quantum-query attacks and for
attacks that restrict the honest parties' quantum queries. We pursue both
possibilities in this work.

Before this work, no unconditional impossibility result was known for
imperfectly correct QCCC key agreement in the QROM with unrestricted quantum
oracle access, even for two-message protocols.  The two-message setting also
captures quantum public-key encryption (QPKE) through a standard
key-agreement protocol: Alice sends her public key, Bob responds with an
encryption of a uniformly random bit, and Alice recovers the shared bit by
decrypting. If the ciphertext is quantum, then the second message in this protocol is
quantum as well.  Bouaziz-Ermann et al. extended the conjectural approach to
this more general setting~\cite{BouazizEtAl24}. Our two-message results
give unconditional attacks in two complementary settings: classical-query
key generation, allowing quantum encryption, decryption, and ciphertexts,
and classical-query encryption with classical ciphertexts, allowing
quantum queries during key generation and decryption. The first follows
by assigning key generation to Alice's first stage, while the second
assigns key generation and decryption to Bob and encryption to Alice.
Both address imperfect correctness, which is not covered by the
perfectly complete impossibility results.

\subsection{Our Results}

We prove unconditional attacks that do not assume perfect completeness, where the attacker is a query-bounded, computationally unbounded party. 

For a key-agreement protocol \(\Pi\), let
\[
  \alpha_\Pi(\lambda)
  :=
  \Pr[K_{\cA}=K_{\cB}\neq\bot]
\]
denote its valid agreement probability. Our bounds are meaningful whenever
\(\alpha_\Pi\) is nonnegligible.  Thus they apply in a substantially broader
regime than negligible completeness error: the honest parties need only agree
on a valid key with inverse-polynomial probability, rather than with
probability \(1-\negl(\lambda)\).

\paragraph{Attacks with a classical-query prefix.}
We begin with two-message protocols. Before sending the first message, Alice
queries the oracle only classically, although her remaining computation and
retained state may be quantum. Bob may return a quantum message, and both Bob
and Alice's final algorithm may query the oracle in superposition. Let
\(q_{\cA,1}\), \(q_{\cB}\), and \(q_{\cA,2}\) bound the number of queries
in Alice's first stage, Bob's stage, and Alice's final stage, respectively.

\begin{theorem}[Two-message attack, informal]
\label{thm:informal-two-message}
Let \(\Pi\) be a protocol of the form above with valid agreement probability
\(\alpha>0\).  There is a computationally unbounded quantum attacker that
recovers Bob's key with probability greater than \(\alpha/2\) using
\[
  O\!\left(
    \frac{
      q_{\cA,1}\max\{1,2q_{\cB}+q_{\cA,2}\}^{4}}
      {\alpha^{4}}
    +q_{\cA,2}
  \right)
\]
random-oracle queries.  Its queries before Bob's response are classical;
its queries after intercepting the response may be quantum.
\end{theorem}

Consequently, no family of such protocols with honest query bounds at most
\(\poly(\lambda)\) and inverse-polynomial valid agreement is secure against
computationally unbounded attackers making \(\poly(\lambda)\) queries. If Alice is
entirely classical and Bob's response is classical, the attack is a passive
classical eavesdropping attack.

In the multi-round setting, we allow for a classical-communication prefix of arbitrary length. The parties may
run quantum algorithms and store quantum states, but all oracle queries made before the final round must be classical. Bob next sends one possibly quantum message, and
Alice performs her final computation. These last two stages may query the
oracle in superposition. Let
\(q_{\cA,\mathsf{pre}}\) and \(q_{\cB,\mathsf{pre}}\) bound the parties'
total prefix queries, and let \(q_{\cB,\mathsf q}\) and
\(q_{\cA,\mathsf{fin}}\) bound their final-stage queries.

\begin{theorem}[Multi-round attack, informal]
\label{thm:informal-multi-round}
Let \(\Pi^{\mathsf{mr}}\) be a protocol of the form above with valid
agreement probability \(\alpha>0\).  There is a computationally unbounded
online quantum attacker that recovers Bob's key with probability greater
than \(\alpha/2\) using
\[
  O\!\left(
    \frac{
      \max\{1,q_{\cA,\mathsf{pre}}\}
      \max\{1,q_{\cB,\mathsf{pre}}\}
      \max\{1,2q_{\cB,\mathsf q}+q_{\cA,\mathsf{fin}}\}^{6}}
      {\alpha^{6}}
    +q_{\cA,\mathsf{fin}}
  \right)
\]
random-oracle queries.  Its queries during the classical prefix are
classical, and the bound has no additional dependence on the number of
prefix rounds.
\end{theorem}

\paragraph{Quantum Bob and classical-query Alice.}
Our second attack allows Bob to make quantum oracle queries throughout
the protocol. Alice makes only classical queries, and all messages are
classical, but both parties may perform quantum local computation and
retain quantum states. We first prove an attack for two-message protocols,
where Bob sends the first message, Alice replies, and Bob performs his
final computation. We then extend it to any constant number of rounds,
with either party speaking first or last.

\begin{theorem}[Attack with one classical-query party, informal]
\label{thm:informal-quantum-bob}
Consider any constant-round protocol with classical messages in which
Alice makes only classical oracle queries and Bob may make quantum
queries throughout. A computationally unbounded passive classical
eavesdropper can recover Bob's key almost as often as the parties agree
on a valid key.
\end{theorem}

The two-message and constant-round bounds are given in
\Cref{cor:qb-ca-two,cor:qb-ca-constant}, respectively. Choosing an
additive loss of \(\alpha/3\) gives recovery probability at least
\(2\alpha/3\).
Thus, for every fixed number of rounds, polynomial honest query bounds
and inverse-polynomial valid agreement suffice for a polynomial-query
classical attack.

\paragraph{Applications.}
The two attacks give complementary impossibility results for quantum
public-key encryption of classical messages whose length is bounded by
a polynomial in \(\lambda\). In both cases, the public and secret keys
are classical.

\begin{theorem}[QPKE attack, informal]
\label{thm:informal-qpke}
No quantum public-key encryption scheme for polynomial-length classical
messages, with classical public and secret keys, negligible correctness
error, and honest query bounds at most \(\poly(\lambda)\), is
query-bounded IND-CPA secure in either of the following settings.

\begin{enumerate}
  \item Key generation has classical oracle access. Encryption, decryption,
  and the ciphertext may be quantum.
  \item Encryption has classical oracle access and the ciphertext is
  classical. Key generation and decryption may have quantum oracle access.
\end{enumerate}
For every polynomial \(p\), there is a computationally unbounded attacker
making \(\poly(\lambda)\) queries whose IND-CPA advantage is at least
\[
  \frac12-\frac{1}{p(\lambda)}-\negl(\lambda).
\]
The attacker may be quantum in the first setting and is classical in the
second.
\end{theorem}

The two settings follow from
\Cref{thm:qpke,thm:qpke-classical-encryption}, respectively. For the
second, let Bob
generate the public key and later decrypt, while Alice encrypts one of
two fixed distinct messages and takes its uniformly chosen index as her
key. This gives a two-message protocol with valid agreement at least
\(1-\negl(\lambda)\). By \Cref{cor:qb-ca-two}, a passive classical
attacker recovers Bob's key with probability at least
\(1-\xi-\negl(\lambda)\). Since Bob's key differs from Alice's
index only with negligible probability, this also yields the stated
IND-CPA advantage by choosing \(\xi\le1/p(\lambda)\).

\begin{table}[t]
  \centering
  \small
  \renewcommand{\arraystretch}{1.15}
  \begin{tabularx}{\textwidth}{@{}l*{3}{>{\centering\arraybackslash}p{1.45cm}}>{\centering\arraybackslash}p{1.75cm}>{\raggedright\arraybackslash}X@{}}
    \toprule
    \multicolumn{1}{c}{\multirow{2}{*}{\textit{Result}}}
      & \multicolumn{3}{c}{\textit{Oracle access}}
      & \multicolumn{1}{c}{\multirow{2}{*}{\textit{Correctness}}}
      & \multicolumn{1}{c}{\multirow{2}{*}{\textit{Status}}} \\
    \cmidrule(lr){2-4}
      & \(\Gen\) & \(\Enc\) & \(\Dec\) & & \\
    \midrule
    \cite{AustrinEtAl22}
      & Q & Q & Q & perfect & Conjectural \\
    \cite{BouazizEtAl24}
      & Q & Q & -- & perfect & Conjectural \\
    \cite{LiEtAl24}
      & C & Q & Q & perfect & Unconditional \\
    \cite{LiEtAl24}
      & C & Q & -- & imperfect & Unconditional \\
    \cite{LiEtAl25}
      & Q (log.) & Q & Q & perfect
      & Unconditional for \(O(\log\lambda)\) \(\Gen\) queries; conjectural
        otherwise \\
    \cite{LiEtAl26PKE}
      & Q & Q & Q & perfect & Unconditional \\
    This work
      & C & Q & Q & imperfect & Unconditional \\
    This work
      & Q & C & Q & imperfect & Unconditional \\
    \bottomrule
  \end{tabularx}
  \caption{Summary of QPKE impossibility results in the QROM.  Here C and
  Q denote classical and quantum random-oracle access, respectively, and
  -- denotes oracle-free decryption. The final row requires classical
  ciphertexts. The two rows for
  \cite{LiEtAl24} separate its perfect-correctness result from its
  imperfect-correctness result, which requires oracle-free decryption.}
  \label{tab:qpke-comparison}
\end{table}

Moreover, \Cref{thm:multi-main} gives an attack against an interactive
key-generation variant of QPKE.

Namely, suppose that ordinary key generation is
replaced by an arbitrary classical-communication setup phase between the
decrypting and encrypting parties.  The parties may perform quantum local
computation and retain quantum states, but both make only classical
random-oracle queries during this phase.  Viewing the setup as the classical
prefix and the challenge ciphertext as Bob's final message, the multi-round
attack applies without any dependence on the number of setup rounds. The reduction fixes two distinct messages and turns the scheme into key
agreement in the natural way. Alice generates and sends the public key; Bob
encrypts one of the two messages and takes its index as his key; Alice
decrypts and recovers the index. Recovering Bob's key is therefore equivalent
to guessing which message was encrypted.

The learning and reprogramming argument also gives a lifting theorem for
distinguishing games: security against classical queries implies security
against quantum queries when computation is unbounded and both challenge
distributions can be sampled using classical oracle queries.
Yamakawa and Zhandry~\cite[Theorem~4.2]{YamakawaZhandry21} previously
proved lifting for search-type games, preserving efficient quantum
computation while changing the adversary's oracle access, with polynomial
loss when the challenger makes constantly many queries. Our result
concerns a decision problem with two classical-query samplers, allows
polynomially many sampler queries, and permits unbounded classical
computation in the resulting distinguisher.

\begin{theorem}[Distinguishing-game lifting, informal]
\label{thm:informal-distinguishing-lifting}
Suppose two distributions of classical strings can each be sampled using
polynomially many classical queries to a random oracle. If a quantum
algorithm distinguishes them with nonnegligible advantage using
polynomially many queries to the same oracle, then so does a
computationally unbounded classical algorithm using polynomially many
classical queries.
\end{theorem}

More precisely, \Cref{thm:distinguishing-lifting} preserves the
distinguishing advantage up to any prescribed inverse-polynomial
additive loss. This extends the classical-generator case of Katz and
Sela's PRG lifting theorem~\cite{KatzSela24} from a generator output
versus a uniform string to two distributions with classical-query
samplers. The guarantee concerns query complexity, and does not assert
that the classical distinguisher runs in polynomial time. 

Finally, our constant-round attack rules out oblivious state preparation
(OSP), introduced by Bartusek and Khurana~\cite{BartusekKhurana24},
in the query-bounded setting below.   

Their reduction from OSP to key agreement with one classical party
lets us apply the attack even when the quantum OSP receiver queries
the oracle in superposition throughout the protocol.

\begin{theorem}[Constant-round OSP impossibility, informal]
\label{thm:informal-osp-constant}
For every fixed number of rounds, no OSP protocol in the QROM with
polynomial honest query complexity and negligible correctness error is
secure against computationally unbounded quantum receivers making
polynomially many oracle queries. The sender is classical and has
classical oracle access, and all communication is classical.
\end{theorem}

This includes two-round OSP as a special case. We prove the theorem in
\Cref{thm:osp-constant}.

\subsection{Discussion}

For the classical-prefix attack, our key-agreement model and security
notion differ from prior work.
In classical random-oracle key agreement, a passive eavesdropper can copy the
public transcript and is commonly required to be unable to recover an honest
party's key~\cite{BarakMahmoodyFull}. In our classical-prefix model, Bob's
final message may be quantum and therefore cannot in general be copied.
Moreover, unlike the Classical
Communication One Quantum Message (CC1QM) security notion of Bouaziz-Ermann
et al.~\cite[Definition~3.1]{BouazizEtAl24}, our success condition asks only
that the attacker recover Bob's key. It does not require the attacker to
provide Alice with a replacement state that makes her output the same
non-abort key. Thus our attack objective is weaker, while security against
our interception key-recovery experiment is correspondingly stronger: every
successful CC1QM attack yields a successful interception key-recovery attack,
but the converse need not hold. This does not by itself break the full CC1QM
condition for a quantum final message. It is nevertheless sufficient for our
QPKE application, because the IND-CPA challenge ciphertext is transferred to
the attacker and need not be returned to Alice.

Moreover, the resampling attacks above rely on the fact that the early oracle
queries are classical. Once the complete classical query-answer transcript is fixed, the
party's residual state depends only on that transcript, while all unqueried
oracle values remain uniform. This makes it possible to resample a conditioned
view and then apply random reprogramming. With quantum queries, a party
may instead be entangled with many oracle values in superposition,
and there is no classical transcript that both determines the
residual state and leaves the remaining oracle values independent and
uniform. Measuring the query register generally changes the execution, while
keeping it coherent does not yield the classically sampled
reprogramming set needed by \Cref{lem:reprogramming}. In
\Cref{sec:quantum-bob-classical-alice}, we handle quantum queries by Bob
throughout the protocol using classical learning to justify replacing
shared oracle answers by independent ones, while Alice's queries remain
classical. Allowing unrestricted quantum
queries by both parties throughout the prefix remains open for imperfect
agreement.

\subsection{Acknowledgements}

This research was conducted in part while AV was a summer intern at NTT Social Informatics Laboratories
in Tokyo. AV is very grateful to NTT and the organizers of the internship for the opportunity and their generous hospitality. AV also acknowledges support from the National Science Foundation Graduate Research Fellowship 2141064. 

We acknowledge the use of AI in the production of this paper. The proof
ideas for the two-round attack with a classical-query first stage were
human-generated, while GPT-5.6 Sol Extra-Hard Thinking assisted in their
formalization. When generalizing to the multi-round setting with a
classical-query prefix, GPT-5.6 Sol Extra-Hard Thinking formalized an attack
by extending Barak-Mahmoody. Through interaction with GPT-6 Astra, we found and formalized the proof for the
result with quantum Bob and classical-query Alice. The authors take full responsibility for the correctness
of this work.

\subsection{Related Work}

\paragraph{Oblivious state preparation.}
Bartusek and Khurana introduced oblivious state preparation (OSP). They showed
that two-round OSP implies one-bit CPA-secure PKE with classical key
generation, decryption, keys, and ciphertexts, but quantum
encryption~\cite[Theorem~7.9]{BartusekKhurana24}. They also reduced OSP to
key agreement with one classical party, thereby ruling out perfectly
correct OSP in the QROM~\cite[Corollary~7.7]{BartusekKhurana24}.
Their reduction preserves constant round complexity, so
\Cref{thm:osp-constant} extends this impossibility to negligible
correctness error in our query-bounded model.

\paragraph{Quantum public-key encryption.}
Under the Polynomial Compatibility Conjecture, Bouaziz-Ermann et al. handled
protocols with a final quantum message. Their result assumes perfect
completeness and no oracle queries by Alice after she receives that message,
so its QPKE application requires oracle-free
decryption~\cite{BouazizEtAl24}. Our Alice-first two-message attack
in \Cref{thm:main} instead permits imperfect completeness and quantum
oracle queries in Alice's final computation, at the cost of requiring
classical queries in her first stage. We also give a Bob-first
two-message attack in \Cref{thm:qb-ca-two}, which likewise permits
imperfect completeness while allowing Bob quantum oracle queries in
both stages and requiring Alice's queries and both messages to be
classical. In its QPKE application, Bob generates and sends the public
key, Alice returns a ciphertext, and Bob decrypts, yielding the attack
on classical-query encryption in \Cref{thm:qpke-classical-encryption}.

Li et al. proved a complementary sequence of QPKE impossibility results.
Their first paper unconditionally rules out perfectly complete QPKE with
classical keys and ciphertexts when key generation has classical oracle
access. It separately rules out imperfectly correct, classical-key QPKE with
possibly quantum ciphertexts in any oracle model, provided that decryption is
oracle-free~\cite{LiEtAl24}. Their next paper allows logarithmically many
quantum key-generation queries unconditionally and gives a conjectural
result beyond that range~\cite{LiEtAl25}. Their 2026 result removes both the
conjecture and the key-generation restriction for perfectly complete QPKE
with classical keys and either classical or quantum
ciphertexts~\cite{LiEtAl26PKE}. Our first QPKE result keeps the
classical-query restriction on key generation, but allows imperfect
correctness and unrestricted quantum oracle queries by encryption and
decryption. Our second instead requires classical queries during
encryption and classical ciphertexts, while allowing quantum queries
during key generation and decryption.

\paragraph{Many-round key agreement.}
In a recent preprint, Li et al. give an unconditional attack against
perfectly complete, arbitrary-round QCCC key agreement with unrestricted
quantum oracle queries~\cite{LiEtAl26ManyRound}. All communication in their
model is classical, and their eavesdropper makes only classical queries. Our
multi-round theorem makes a different tradeoff: it tolerates imperfect
completeness and a final quantum message, but it requires classical oracle
access during the prefix and uses a quantum attacker after the prefix.

\subsection{Organization}

\Cref{sec:technical-overview} explains the two main attacks, their
extensions to multiple rounds, and the QPKE application.
\Cref{sec:preliminaries} then defines the query model, partial oracles,
and the key-agreement and QPKE security notions, and recalls the external
lemmas used in the proofs. \Cref{sec:resampling} constructs the
heavy-query learner for the two-message setting and defines Alice's
conditioned first-round view.  \Cref{sec:main-attack} proves the two-round
and multi-round attacks with a classical-query prefix.
\Cref{sec:quantum-bob-classical-alice} gives the passive classical attack
when Bob has quantum oracle access throughout, including its
constant-round extension. Finally, \Cref{sec:corollaries} derives the classical-eavesdropping,
QPKE, distinguishing-game lifting, and OSP consequences.

\section{Technical Overview}\label{sec:technical-overview}

We first explain the attack with a classical-query prefix, proved in
\Cref{sec:main-attack}, before turning to the setting in which Bob has
quantum oracle access throughout, treated in
\Cref{sec:quantum-bob-classical-alice}.

\subsection{Attacks with a Classical-Query Prefix}
\label{sec:overview-classical-prefix}

We begin with Alice's first stage as the classical-query prefix, then
extend the argument to a prefix containing multiple rounds.

\subsubsection{Two-Round Protocols and QPKE}
\label{sec:overview-prefix-two}

We consider two-round key agreement in the QROM, where both parties
are quantum algorithms but Alice's oracle queries and her message to
Bob must be classical. Alice first makes at most $q_{\cA, 1}$ queries
to the random oracle $H$, obtaining query--answer pairs $T_\cA$,
and sends her message $m_1$ to Bob while retaining a residual state
$S_\cA$. Bob then queries $H$, outputs a key $K_B$, and sends
his possibly quantum message $\psi$ to Alice, who uses it to produce
her key $K_A$.

The difficulty is that sampling another Alice view consistent with the
public transcript does not make it consistent with the actual oracle
Bob used, while changing that oracle can be detected by Bob's
queries. Alice may also retain a quantum state and make further oracle
queries after receiving Bob's message, so learning her likely initial
queries alone is not sufficient. The argument below
must preserve the correlation between a sampled transcript and its
residual quantum state, then control the effect of reprogramming through
both Bob's computation and Alice's final stage. 

After receiving Alice's first message \(m_1\), Eve looks for an input
that is likely to occur in Alice's first-round transcript, conditioned
on \(m_1\) and all oracle answers she has learned so far, and queries
it whenever this conditional probability exceeds \(\delta\).
She stores each answer in a partial oracle \(h\), which agrees with
$H$ on every point in its domain $\Dom(h)$, and stops after at most
\(\ceil{q_{\cA,1}/\delta^2}\) queries. The probability of reaching
this cutoff is at most \(\delta\), and successful learning ensures
that every point outside \(\Dom(h)\) appears in a fresh conditioned
Alice transcript with probability at most \(\delta\).

Although Alice may retain a quantum state after her classical queries,
fixing her message and full query transcript determines the conditional
distribution of that state, while all unqueried oracle values remain
uniform. Conditioned on \((m_1,h)\), Eve can therefore prepare the
distribution over Alice's views, denoted by $\omega_{m_1, h}$,
and use it to sample a new view
\[
  (T'_{\cA},S'_{\cA})\leftarrow\omega_{m_1,h}.
\] 
Eve then intercepts Bob's response \(\Psi\) and runs Alice's final algorithm
on \(S'_{\cA}\) and \(\Psi\), using a \emph{reprogrammed oracle}
\(H[T'_{\cA}]\) that agrees with $T'_\cA$ on $\Dom(T'_{\cA})$
and with $H$ everywhere else.

The analysis first uses the heavy-query bound of Austrin et
al.~\cite[Lemma~3.4]{AustrinEtAl22}, in the form recalled by Katz and
Sela~\cite[Lemma~4]{KatzSela24}, to ensure that Eve learns the heavy
inputs before reprogramming. We then apply the random reprogramming
lemma stated by Katz and Sela~\cite[Lemma~3]{KatzSela24}, which bounds
the effect of reprogramming an oracle on a random set: if each point belongs to that set
with probability at most \(\delta\), then the acceptance probability of a
\(q\)-query quantum algorithm changes by at most \(2q\sqrt\delta\). 

The QPKE application uses the standard reduction from QPKE to key
agreement used in prior QROM impossibility work
\cite{BouazizEtAl24,LiEtAl24}. After Alice generates
\((\pk,\sk)\) and sends \(\pk\), Bob chooses a uniform bit \(b\),
encrypts \(m_b\) for two fixed distinct messages \(m_0,m_1\), and
takes \(b\) as his key. Alice decrypts and maps \(m_0,m_1\)
back to their indices, giving agreement probability at least the
scheme's minimum correct-decryption probability. In the IND-CPA game,
Eve preprocesses \(\pk\), treats the challenge ciphertext as Bob's
response, and uses her recovered key as the guess for \(b\), so
interception success \(p\) yields IND-CPA advantage at least
\(p-1/2\). Because the ciphertext is transferred rather than copied,
encryption, decryption, and the ciphertext may all be quantum, with
only key generation requiring classical oracle access.

\subsubsection{Multi-Round Protocols}
\label{sec:overview-prefix-multi}

For the multi-round result, we treat the entire classical-query prefix
as an auxiliary classical protocol, using the fact that the parties'
messages and oracle queries during the prefix are classical even though
their local computations and residual states may be quantum.
Since we do not impose a computational bound, an
exact simulator can sample these classical actions from their induced
distributions and record the corresponding conditional residual states.  We
can therefore apply the Barak--Mahmoody learner to the simulated classical
views and then map its guarantees back to the parties' actual query
transcripts and residual quantum states~\cite{BarakMahmoodyFull}.
\Cref{fig:overview-classical-prefix-view} summarizes Eve's information
during the prefix and the final round.

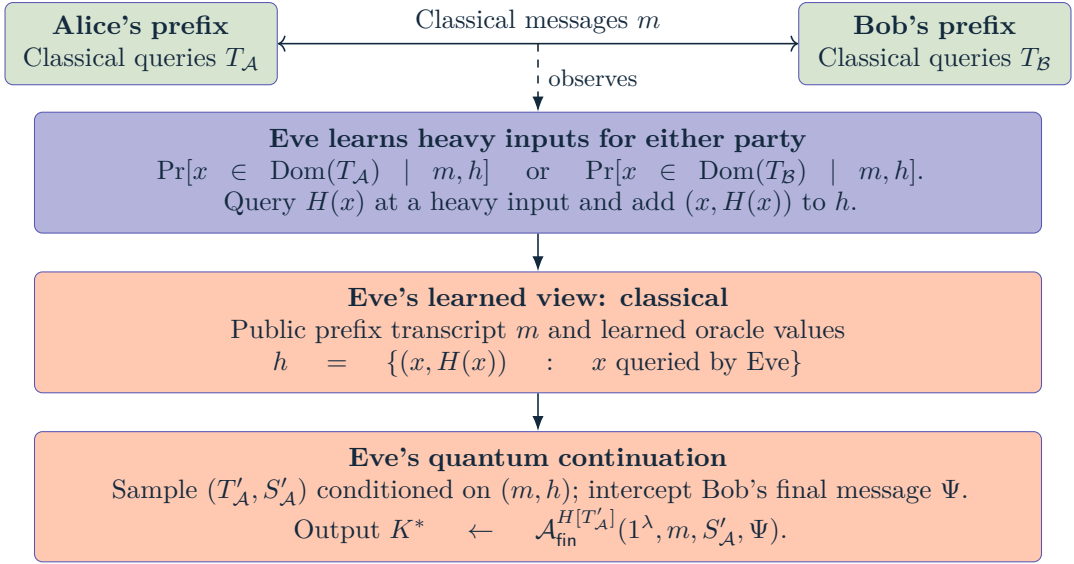
\begin{figure}[htbp]
\centering
\begin{tikzpicture}[
  font=\small,text=DeepSpaceBlue,
  party/.style={draw=NavyBlue!65,fill=TeaGreen!75,rounded corners=3pt,align=center,
    text width=3.25cm,minimum height=1cm,inner sep=5pt},
  learnbox/.style={draw=NavyBlue!65,fill=NavyBlue!30,text=DeepSpaceBlue,rounded corners=3pt,align=center,
    text width=12.9cm,inner sep=6pt},
  eveview/.style={learnbox,draw=NavyBlue!65,fill=BlazeOrange!35,text=DeepSpaceBlue},
  flow/.style={-{Latex[length=2mm]},draw=DeepSpaceBlue,semithick}
]
  \node[party] (alice) at (0,0)
    {\textbf{Alice's prefix}\\Classical queries \(T_{\cA}\)};
  \node[party] (bob) at (10.5,0)
    {\textbf{Bob's prefix}\\Classical queries \(T_{\cB}\)};
  \draw[flow,<->] (alice.east) --
    node[above,align=center] {Classical messages \(m\)}
    coordinate[midway] (messages) (bob.west);
  \node[learnbox,below=9mm of messages] (learn)
    {\textbf{Eve learns heavy inputs for either party}\\
     \(\Pr[x\in\Dom(T_{\cA})\mid m,h]\)
     \quad or \quad
     \(\Pr[x\in\Dom(T_{\cB})\mid m,h]\).\\
     Query \(H(x)\) at a heavy input and add \((x,H(x))\) to \(h\).};
  \draw[flow,dashed] (messages) -- node[right,font=\footnotesize]
    {observes} (learn.north);
  \node[eveview,below=5mm of learn] (view)
    {\textbf{Eve's learned view: classical}\\
     Public prefix transcript \(m\) and learned oracle values\\
     \(h=\{(x,H(x)):x\text{ queried by Eve}\}\)};
  \draw[flow] (learn.south) -- (view.north);
  \node[eveview,below=5mm of view] (continue)
    {\textbf{Eve's quantum continuation}\\
     Sample \((T'_{\cA},S'_{\cA})\) conditioned on \((m,h)\);
     intercept Bob's final message \(\Psi\).\\
     Output \(K^*\leftarrow\cA_{\mathsf{fin}}^{H[T'_{\cA}]}
     (1^\lambda,m,S'_{\cA},\Psi)\).};
  \draw[flow] (view.south) -- (continue.north);
\end{tikzpicture}
\caption{Eve's information in the first setting, on a successful learning
branch. Both parties make classical queries during the prefix; their local
states may be quantum. Eve learns selected oracle values, without receiving
either private query transcript. In the final round she additionally
prepares a quantum state \(S'_{\cA}\) and intercepts \(\Psi\), which may
be quantum.}
\label{fig:overview-classical-prefix-view}
\end{figure}

Conditioned on the public transcript \(m\) and the partial oracle \(h\)
learned by Eve, the parties' prefix views are close to the product of their
marginals, and no input outside \(\Dom(h)\) is likely to appear in Alice's
prefix query transcript. To extend this product guarantee to include
the random oracle used in the final round, we use the fact that every
prefix query is classical: once the two query transcripts are fixed,
the unqueried oracle values remain uniform and the residual states
have no further dependence on them. This lets us sample independent
prefix views together with a consistent oracle and show that the
resulting state is close to the real conditioned state, as detailed in
\Cref{sec:classical-prefix-views,sec:bm-learner,sec:conditioned-prefix-resampling}.

Eve can then sample a fresh Alice prefix view, intercept Bob's final message,
and run Alice's final computation from the sampled residual state.  Since
the learner's query bound depends only on the total numbers of prefix
queries made by Alice and Bob, the attack has no additional dependence on
the number of prefix rounds.

\subsection{Quantum Bob and Classical-Query Alice}
\label{sec:overview-quantum-bob}

We next consider protocols in which Bob has quantum oracle access
throughout the execution and Alice makes only classical oracle queries,
with all messages classical even though Alice's local computation and
residual state may be quantum. We begin with two rounds, where the roles
of our two learners are easiest to see, before extending the argument
to a constant number of rounds.

\subsubsection{Two-Round Protocols}
\label{sec:overview-quantum-bob-two}

Consider the two-message execution \(\cB_1\to\cA\to\cB_2\),
with respective query bounds \(q_{\cB,1},q_{\cA},q_{\cB,2}\).
We focus on the Bob-first order because Bob may already hold quantum
information about the oracle when Alice begins, while the Alice-first
case is covered by the argument in \Cref{sec:overview-classical-prefix}. 

Our attack uses two classical learners, \(L_1\) after Bob's message
and \(L_2\) after Alice's reply, and then guesses Bob's key by
classical computation from the information they have collected. 

\paragraph{Why two learners are needed.}
The goal is to make the parties' keys approximately independent once
Eve's view is fixed, because then Eve can predict Bob's key at least
as well as Alice can. In the Bob-first execution, however, Alice's
oracle answers may be correlated with information Bob already holds from
his first stage, and Bob may also learn those answers by querying the
same inputs during his final stage. The first learner accounts for
what Bob already knows about the oracle, while the second accounts
for what he may learn after Alice's reply.

The first learner, \(L_1\), lets us replace Alice's fresh oracle
answers by independent values in the proof without significantly
changing the execution. Although this separates her new answers from
Bob's earlier information, he must still receive those same answers
if he later queries inputs she used, so the first step alone does not
make their keys independent. The second learner, \(L_2\), therefore
queries inputs Alice is likely to have used, leaving only inputs that
are individually unlikely to occur in her query record. This lets us
remove her remaining private answers from Bob's oracle with little
effect on his final state, so that the information they still share
is already contained in Eve's view.

For classical parties, one can learn inputs likely to occur in either
party's private query record, as in Barak and Mahmoody's attack
\cite[Construction~3.3 and Theorem~4.3]{BarakMahmoodyFull}.
Bob's quantum queries need not have such a record, so \(L_1\)
instead uses the Fourier learner of Austrin et al.
\cite[Definition~4.9 and Construction~4.10]{AustrinEtAl22}.
In the Fourier representation of a quantum random oracle, a nonzero
label at an input records oracle information acquired through queries,
and the conditional probability of such a label is that input's
\emph{Fourier weight}. The learner queries inputs of large Fourier
weight until every unlearned input has small weight, allowing us to
change the answers Alice sees there with little effect on the joint
state. Eve computes these weights from the protocol description and
her view using unrestricted classical computation, without measuring
Bob's state or accessing the oracle's Fourier representation.

\paragraph{Eve's view and the reference experiment.}
In the real execution, it is not enough to show that each unlearned
input is unlikely to have been queried by Alice, because Bob's state
may contain information that helps him find the inputs she used.
To justify removing Alice's remaining private answers from Bob's
oracle, we therefore need to separate her private query record from
the information Bob holds, while preserving the answers already
learned by Eve.
We introduce a \emph{reference experiment} in which this separation holds
given Eve's view, so that learning Alice's likely queries is enough
to limit what Bob's later queries reveal about her remaining private
answers.
We define the learners using the reference experiment \(\mathsf D\)
of \Cref{def:qb-ca-two-learners}, in which Bob starts with the real
random oracle \(H\), while Alice uses an independently sampled
random function \(G\), modified to preserve the answers Eve has
already learned. In this experiment, Alice's new private information
is independent of Bob's earlier information once Eve's view is fixed,
giving the conditional independence needed for key recovery. The
reference experiment thus serves two purposes: it supplies the
distributions used by the learners to identify likely queries and
provides an execution with the independence property that we aim to
establish through the hybrid argument.

More concretely, \(\cB_1\) and \(L_1\) use \(H\) to
produce a learned table \(h_1\), after which Alice and \(L_2\)
use \(G[h_1]\), meaning \(G\) with its answers replaced by
\(h_1\) on the learned inputs. Eve's view \(E=(M,J,F)\)
records the processed messages, the ordered learner query--answer
pairs, and the learners' success or failure flags, so conditioning on
that view fixes every answer in the learned table. Each new learner
answer becomes public through this table, while Alice's other answers
remain private and may still be correlated with her residual quantum
state.

Let \(S\) be the set of inputs Alice queried outside
\(\Dom(h_1)\). The second learner queries an unlearned input
\(x\) whenever \(\Pr_{\mathsf D}[x\in S\mid E=e]>\gamma\),
so its purpose is to find the inputs Alice is likely to have used,
rather than to recover her entire query record.  
If \(h_2\) contains the answers it learns, then Eve knows
\(h=h_1\cup h_2\) and Alice's remaining private answers form
\(P=G|_{S\setminus\Dom(h_2)}\). Writing \(A,B\) for the
parties' retained registers, upon successful learning,
each input has conditional probability at most \(\gamma\) of
belonging to this private table, and \(AP\) is independent of
\(BH\) given Eve's view.

\paragraph{What the two learner guarantees accomplish.}
The first learner addresses Bob's information from before Alice's
stage: once every remaining Fourier weight is at most \(\beta\),
giving Alice and \(L_2\) independent answers changes the joint
state by at most \(2(q_{\cA}+N_2)\sqrt\beta\), where \(N_2\)
is the query limit for \(L_2\). For this change, Bob must still
receive the new answer if he later queries an input Alice or \(L_2\)
used, so his continuation uses \(H[h][P]\), which agrees with
the two disjoint tables and uses \(H\) elsewhere.
The first guarantee therefore separates Alice's fresh answers from
Bob's past information, but leaves him able to learn those answers
through his later queries.

The second learner addresses those later queries: because Alice's
remaining private inputs are individually unlikely and independent
of Bob's state given Eve's view, replacing \(H[h][P]\) by
\(H[h]\) changes Bob's final state by at most
\(2q_{\cB,2}\sqrt\gamma\). Bob then uses only his own state,
the original oracle, and Eve's learned table, leaving his key
independent of Alice's conditioned on Eve's view. Both arguments
retain the parties' full states, which will let us repeat the same
steps when further rounds remain.

The learners need only boundedly many queries because each query
removes a significant amount of the information being learned:
the total Fourier weight is at most \(q_{\cB,1}\), and Alice
has at most \(q_{\cA}\) private inputs. With failure allowances
\(\eta_1,\eta_2\), the corresponding query limits are
\[
  N_1=\ceil{\frac{q_{\cB,1}}{\eta_1\beta}},\qquad
  N_2=\ceil{\frac{q_{\cA}}{\eta_2\gamma}}.
\]
Each failure bound applies in the reference experiment used to define
that learner, so the proof uses it only after the relevant prefix has
been changed to follow that experiment.

\paragraph{The real attack and key recovery.}
Eve receives the complete honest transcript, runs \(L_1\) using
Bob's message and then \(L_2\) using both messages, and obtains
every answer from the real, unchanged \(H\). The reference experiment
determines her query choices by classical computation from her current
view, but does not supply the answers in the attack. If a learner fails,
Eve makes no further queries and returns \(\bot\). For the analysis,
we may place the learners in the order
\[
  \cB_1,\quad L_1,\quad \cA,\quad L_2,\quad \cB_2,
\]
because their queries only read the fixed oracle and the honest parties
receive no information from Eve, so moving the learner calls does not
change the joint distribution of the honest keys and Eve's view.
\Cref{fig:overview-quantum-bob-view} summarizes the attack.

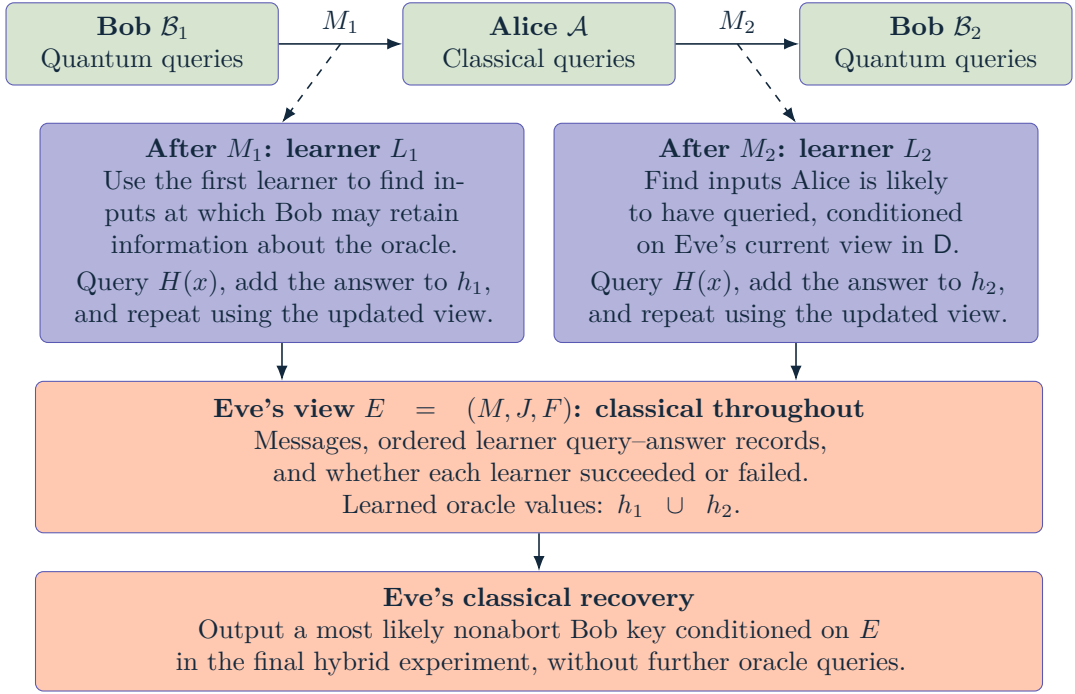
\begin{figure}[htbp]
\centering
\begin{tikzpicture}[
  font=\small,text=DeepSpaceBlue,
  party/.style={draw=NavyBlue!65,fill=TeaGreen!75,rounded corners=3pt,align=center,
    text width=3.25cm,minimum height=1cm,inner sep=5pt},
  learner/.style={draw=NavyBlue!65,fill=NavyBlue!30,text=DeepSpaceBlue,rounded corners=3pt,align=center,text width=6.05cm,
    minimum height=2.9cm,inner sep=5pt},
  eveview/.style={draw=NavyBlue!65,fill=BlazeOrange!35,rounded corners=3pt,align=center,
    text width=12.9cm,inner sep=6pt},
  flow/.style={-{Latex[length=2mm]},draw=DeepSpaceBlue,semithick}
]
  \node[party] (bobfirst) at (0,0)
    {\textbf{Bob \(\cB_1\)}\\Quantum queries};
  \node[party] (alice) at (5.25,0)
    {\textbf{Alice \(\cA\)}\\Classical queries};
  \node[party] (bobsecond) at (10.5,0)
    {\textbf{Bob \(\cB_2\)}\\Quantum queries};
  \draw[flow] (bobfirst.east) -- node[above] {\(M_1\)}
    coordinate[midway] (firstmessage) (alice.west);
  \draw[flow] (alice.east) -- node[above] {\(M_2\)}
    coordinate[midway] (secondmessage) (bobsecond.west);
  \node[learner] (beforealice) at (1.85,-2.5)
    {\textbf{After \(M_1\): learner \(L_1\)}\\
     Use the first learner to find inputs at which Bob may retain
     information about the oracle.\\[2pt]
     Query \(H(x)\), add the answer to \(h_1\),
     and repeat using the updated view.};
  \node[learner] (classical) at (8.65,-2.5)
    {\textbf{After \(M_2\): learner \(L_2\)}\\
     Find inputs Alice is likely to have queried, conditioned on
     Eve's current view in \(\mathsf D\).\\[2pt]
     Query \(H(x)\), add the answer to \(h_2\),
     and repeat using the updated view.};
  \draw[flow,dashed] (firstmessage) -- (beforealice.north);
  \draw[flow,dashed] (secondmessage) -- (classical.north);
  \node[eveview,below=5mm of beforealice,xshift=3.4cm] (view)
    {\textbf{Eve's view \(E=(M,J,F)\): classical throughout}\\
     Messages, ordered learner query--answer records,
     and whether each learner succeeded or failed.\\
     Learned oracle values: \(h_1\cup h_2\).};
  \draw[flow] (beforealice.south) -- ([xshift=-3.4cm]view.north);
  \draw[flow] (classical.south) -- ([xshift=3.4cm]view.north);
  \node[eveview,below=5mm of view] (recover)
    {\textbf{Eve's classical recovery}\\
     Output a most likely nonabort Bob key conditioned on \(E\)\\
     in the final hybrid experiment, without further oracle queries.};
  \draw[flow] (view.south) -- (recover.north);
\end{tikzpicture}
\caption{Eve's information in the second setting, shown for a successful
branch of the Bob-first two-message attack. Both learners compute their
query choices from the reference experiment in
\Cref{def:qb-ca-two-learners} and obtain their actual
answers from \(H\). Eve's view contains neither party's private state
nor Alice's remaining private query table. Unlike the attack in
\Cref{fig:overview-classical-prefix-view}, her key recovery uses only
classical computation. The constant-round attack repeats these learners.}
\label{fig:overview-quantum-bob-view}
\end{figure}

Applying the two oracle reprogramming steps in this order gives a final experiment
whose keys are conditionally independent and whose joint distribution
of keys and Eve's view is within distance
\[
  \Delta_2=\eta_1+\eta_2+
    2(q_{\cA}+N_2)\sqrt\beta+2q_{\cB,2}\sqrt\gamma
\]
of the real distribution, with both recorded keys set to \(\bot\)
on learner failure in the proof experiments. In the final experiment,
Eve guesses a most likely nonabort Bob key given her view, which
succeeds at least as often as Alice and Bob agree on a valid key.
Transferring both agreement and guessing success back to the real
execution gives success at least \(\alpha-2\Delta_2\), where
\(\alpha\) is the real valid agreement probability, with at most
\(N_1+N_2\) classical oracle queries and no further queries for
the final guess.

\subsubsection{Constant-Round Protocols}
\label{sec:overview-quantum-bob-constant}

For an alternating execution
\(\cB_1,\cA_1,\cB_2,\ldots,\cA_r,\cB_{r+1}\), we place
a learner \(L_i^-\) before each Alice stage and a learner \(L_i^+\)
after it, repeating the two changes above while retaining the parties'
states and Alice's remaining private table. The reference experiment
\(\mathsf D_r\) of \Cref{def:qb-ca-round-decoupled} uses the
same principle throughout: Alice and \(L_i^+\) reuse recorded
answers and obtain independent uniform values at fresh inputs, while
Bob and \(L_i^-\) use \(H[h]\), ignoring Alice's private table.
Since each stage reveals information from only one of the two factors
\(AP\) and \(BH\), fixing the updated Eve view preserves their
independence, even though Alice's private table now includes answers
left over from earlier rounds.

The additional difficulty is that removing Alice's private answers
from Bob's oracle also requires removing them from the oracle used
by the learner following his message. With at most \(b\) Bob
queries and \(c\) queries by that learner, the resulting error is
at most \(2b\sqrt\gamma+c\gamma\), where the second term accounts
for the learner itself encountering a private input. Consequently,
the accuracy required of one learner depends on the query limit of
the next, and we choose the thresholds in reverse order to control
the total error.

Reducing a learner's threshold makes its error smaller but increases
the number of queries it may need, forcing us to reduce the thresholds
of earlier learners and increase their query limits in turn. In turn, the query bound
grows recursively, yielding
\(\bigl(\poly(Q,r,1/\xi)\bigr)^{O(2^r)}\). This is why our polynomial-query guarantee is restricted to a
constant number of rounds: the bound is polynomial in the total
honest query bound \(Q\) and \(1/\xi\) when \(r\) is fixed
independently of the security parameter, but its growing exponent
prevents a general polynomial-query guarantee when \(r\) grows.
For each fixed \(r\),
the attack recovers Bob's key with probability at least
\(\alpha-\xi\), and by allowing empty initial or final stages,
the same argument covers protocols in which either party speaks first
or last. 

\section{Preliminaries}\label{sec:preliminaries}

Let \(\lambda\in\mathbb{N}\) be the security parameter, and let
\(n,\ell:\mathbb{N}\to\mathbb{N}_{>0}\) be functions bounded above by a
polynomial in \(\lambda\). Define
\[
  \cX_\lambda=\bits^{n(\lambda)},
  \qquad
  \cY_\lambda=\bits^{\ell(\lambda)},
  \qquad
  \cH_\lambda=\{H:\cX_\lambda\to\cY_\lambda\}.
\]
For any finite set \(\cR\), we write \(r\leftarrow\cR\) for a uniformly
random element of \(\cR\).  In particular,
\(H\leftarrow\cH_\lambda\) denotes a uniformly random function.  When
\(\lambda\) is clear, we omit the subscript.  When representing the
sampled oracle by a classical register, \(\mathsf H\) has orthonormal basis
\(\{\ket{H}_{\mathsf H}:H\in\cH\}\). For a probability mass function or subprobability mass function \(\pi\), we
write
\[
  \Supp(\pi):=\{z:\pi(z)>0\}
\]
for its support.
Let \(\cK_\lambda\) be the finite set of nonabort keys, and let
\(\bot\) denote abort. For deterministic choices among eligible
inputs or keys, fix total orders on \(\cX_\lambda\) and
\(\cK_\lambda\) and choose the first eligible element.

\subsection{Random Oracles and Query Algorithms}

\begin{definition}[Oracle quantum algorithms]\label{def:oracle-algorithms}
A quantum oracle algorithm
\(\mathcal Q=\{\mathcal Q_\lambda\}_{\lambda\in\mathbb N}\) is a family
of quantum circuits. We write \(I\) for its input register, \(X\) for its
\(n(\lambda)\)-qubit query register, \(Y\) for its
\(\ell(\lambda)\)-qubit response register, and \(Z\) for all remaining work
registers. The circuit family depends only on $\lambda$. For
\(H\in\cH_\lambda\), a query applies the unitary
\[
  U_H\ket{x}_X\ket{y}_Y\ket{z}_Z
  =
  \ket{x}_X\ket{y\oplus H(x)}_Y\ket{z}_Z.
\]
The algorithm alternates oracle-independent operations with applications of
\(U_H\). For fixed \(\lambda\), let \(q(\mathcal Q_\lambda)\) denote its
worst-case number of oracle queries; we sometimes abbreviate this to
\(q(\mathcal Q)\). 

An algorithm has \emph{classical oracle access} if each query is made on a
classical input \(x\), the answer \(H(x)\) is returned classically, and the
ordered list of query-answer pairs can be recorded without changing the
execution. The rest of the computation and the residual state may be
quantum.
\end{definition}

\subsection{States and Conditional Distributions}
\label{sec:state-notation}

We use the standard classical-quantum representation of a state that is
classical on one register. If \(C\)
is a classical register with finite alphabet \(\mathcal C\) and \(S\) is a
quantum register, then
\[
  \rho_{CS}
  =
  \sum_{c\in\mathcal C}
  p(c)\proj{c}_C\otimes\sigma_{S\mid c},
\]
where \(p\) is a probability distribution and, whenever \(p(c)>0\),
\(\sigma_{S\mid c}\) is the normalized state of \(S\) obtained by measuring
\(C\) in the computational basis and observing \(c\). Explicitly, writing
\(P_c=\proj{c}_C\),
\[
  p(c)=\Tr\!\left[(P_c\otimes I_S)\rho_{CS}\right]
\]
and, for \(p(c)>0\),
\[
  \sigma_{S\mid c}
  =
  \frac{
    \Tr_C\!\left[(P_c\otimes I_S)\rho_{CS}(P_c\otimes I_S)\right]
  }{p(c)}.
\]

We write \(D(\rho,\sigma):=\tfrac12\|\rho-\sigma\|_1\) for
trace distance. For classical distributions it equals total variation
distance. At a specified stage of an experiment \(\mathsf X\),
\(\rho^{\mathsf X}_{RE}\) denotes the joint state of retained
registers \(R\) and a classical record \(E\). For a value \(e\)
of positive probability, \(\rho^{\mathsf X,e}_R\) is the normalized
state conditioned on \(E=e\). The same convention applies to
marginals, additional registers, and conditional probabilities. We omit
\(\mathsf X\) only when the experiment and stage are clear.

\subsection{Compressed Oracles}

We use the Fourier representation of Zhandry's compressed
oracle~\cite[Appendix~C]{Zhandry19}, following the notation of
Austrin et al.~\cite[Section~2.1]{AustrinEtAl22}.

\begin{definition}[Zhandry's compressed oracle in the Fourier basis]
\label{def:compressed-oracle}
Fix \(\lambda\), and write \(\ell=\ell(\lambda)\). For
\(u\in\cY\), define the Fourier-basis state
\[
  \ket{\widehat u}:=2^{-\ell/2}
       \sum_{y\in\cY}(-1)^{u\cdot y}\ket y,
\]
where the inner product is over \(\mathbb F_2\). Represent the
random oracle by a register
\(\mathsf H=\bigotimes_{x\in\cX}\mathsf H_x\), with one
\(\ell\)-qubit register for each input, initialized in
\[
  \bigotimes_{x\in\cX}\ket{\widehat0}_{\mathsf H_x}
  =\frac{1}{\sqrt{|\cH|}}\sum_{H\in\cH}\ket H_{\mathsf H},
  \qquad \ket H_{\mathsf H}:=\bigotimes_{x\in\cX}\ket{H(x)}.
\]
A query acts as
\[
  U\ket x_X\ket z_Y\ket H_{\mathsf H}
    =\ket x_X\ket{z\oplus H(x)}_Y\ket H_{\mathsf H}.
\]
The algorithm accesses \(\mathsf H\) only through this query operation.

To compress this register, encode a Fourier-basis vector
\(\bigotimes_x\ket{\widehat{u(x)}}\), for \(u:\cX\to\cY\),
by the database
\[
  D_u:=\{(x,u(x)):u(x)\ne0\},
\]
with entries ordered lexicographically by input. An omitted input has
Fourier label zero.
The compressed oracle starts in \(\ket{\varnothing}\). On a
Fourier-basis response register, its query operation is
\[
  U\ket x_X\ket{\widehat v}_Y\ket{D_u}
   =\ket x_X\ket{\widehat v}_Y\ket{D_{u'}},\qquad
  u'(z)=u(z)\oplus\boldsymbol1_{\{z=x\}}v.
\]
Thus a query changes only the entry at its input, inserting an entry if
its label becomes nonzero and deleting it if the label becomes zero.
This operation extends linearly to superpositions. The database is a
quantum register whose entries are Fourier labels, rather than a
classical list of queried inputs and oracle answers.
\end{definition}

Measuring \(\mathsf H\) initially in the computational basis would
sample a uniform function. Since queries use it only as a control, this
measurement can be deferred. Tracing out the oracle therefore gives
exactly the same state on the algorithm's registers as querying a
uniformly sampled random oracle. The compressed encoding preserves
this equivalence. Each query increases the database size by at most
one, so after \(q\) queries its state is supported on databases of
size at most \(q\). This bound also holds after conditioning on
classical outcomes of the algorithm, including messages it sends.
We use the full register notation \(\mathsf H\) and its compressed
encoding interchangeably.

\subsection{Partial Oracles and Reprogramming}

\begin{definition}[Partial oracle]
For any $\lambda \in \mathbb{N}$, fix $\cX = \cX_\lambda, \cY = \cY_\lambda,$ and $\cH = \cH_\lambda$.  A partial oracle is a function
\(h:\Dom(h)\to\cY\), where
\(\Dom(h)\subseteq\cX\). We denote the \emph{size} of $h$ by \(|h|=|\Dom(h)|\). We write
\(h\preceq H\) when a full oracle \(H\in\cH\) extends \(h\), meaning that
\(H(x)=h(x)\) for every \(x\in\Dom(h)\).  Define
\[
  \cH(h)=\{H\in\cH:h\preceq H\}.
\]
Two partial oracles \(h\) and \(\tau\) are \emph{consistent} if they agree on the
intersection of their domains.  For consistent \(h\) and \(\tau\), define
\[
  \cH(h,\tau)
  =
  \{H\in\cH:h\preceq H,\ \tau\preceq H\}.
\]
The notation \(H\leftarrow\Unif(\cH(h))\), and similarly
\(H\leftarrow\Unif(\cH(h,\tau))\), means that \(H\) is sampled uniformly
from the indicated set.
\end{definition}

For a partial function \(h\) and a set \(S\subseteq\cX\),
\(h|_S\) denotes the restriction to \(\Dom(h)\cap S\).
Thus \(h|_{\cX\setminus\{x\}}\) deletes any entry at \(x\).

An ordered transcript
\(\tau=((x_1,y_1),\ldots,(x_t,y_t))\) is \emph{consistent} when
\(x_i=x_j\) implies \(y_i=y_j\).  A consistent transcript defines the
partial oracle \(\tau(x_i)=y_i\).  Repeated queries may be retained in the
ordered list, while \(\Dom(\tau)\) denotes the set of distinct queried
points.  We write \(|\tau|=t\) for the length of the ordered list, including
repetitions.

\begin{definition}[Reprogrammed oracle]
\label{def:reprogrammed-oracle}
For an oracle \(H\) and a consistent transcript \(\tau\) (equivalently, a
partial oracle), define
\[
  H[\tau](x)
  =
  \begin{cases}
    \tau(x),&x\in\Dom(\tau),\\
    H(x),&x\notin\Dom(\tau).
  \end{cases}
\]
For another internally consistent transcript \(\tau'\), which need not be
consistent with \(\tau\), successive reprogramming is
defined so that the second transcript overwrites the first on their common
domain:
\[
  H[\tau][\tau'](x)
  =
  \begin{cases}
    \tau'(x),&x\in\Dom(\tau'),\\
    \tau(x),&x\in\Dom(\tau)\setminus\Dom(\tau'),\\
    H(x),&x\notin\Dom(\tau)\cup\Dom(\tau').
  \end{cases}
\]
\end{definition}

\subsection{Classical Learners}
\label{sec:learner-records}

We use the following notation for the data maintained by the learners in
\Cref{sec:quantum-bob-classical-alice}. Eve runs the learners as
subroutines and maintains a classical record \(E=(M,J,F)\), called
her \emph{current view}. Write \(M_{\mathrm{full}}\) for the complete
honest transcript. Eve receives this transcript after the honest execution
but supplies each learner only the prefix specified by its rule. The
registers \(M,J,F\) store lists that start empty:
\begin{itemize}[leftmargin=2em,itemsep=0.2em,label=\(\bullet\)]
  \item \(M\) records the prefix of \(M_{\mathrm{full}}\) processed
  so far, possibly the full transcript.
  \item \(J\) records \((\mathsf L,x,y)\) after each learner query,
  where \(\mathsf L\) identifies the learner, \(x\) is its input,
  and \(y\) is the answer.
  \item \(F\) records \((\mathsf L,f)\) when learner \(\mathsf L\)
  terminates with \(f\in\{\mathsf{success},\mathsf{failure}\}\).
\end{itemize}
Eve extends all three lists. She appends a message to \(M\) when she
processes it, a query triple to \(J\) when she receives the answer,
and a termination pair to \(F\) when the corresponding learner stops.
Messages beyond \(M\) are not used to choose the current learner's
queries or test its stopping condition. In the real attack, Eve processes
\(M_{\mathrm{full}}\) in order. In interleaved proof executions, she
appends each message as it is sent. At the end of the real attack,
\(M=M_{\mathrm{full}}\), even if a learner fails.

Eve's \emph{learned table} \(h\) is the partial function formed from
the query--answer pairs in \(J\). Each learner queries only outside
\(\Dom(h)\), so every answer adds one entry. Neither honest party
is given \(h\) or notified of Eve's queries.
The registers \(A,B\) contain the parties' retained states, including
their keys once produced. Alice's register \(A\) includes her complete
ordered query--answer record \(T_{\cA}\), with repetitions, and her
residual state. In the proof we also maintain a partial function \(P\)
containing her answers at inputs not yet queried by Eve. It specifies
the private answers still used by the current oracle.

Initially \(h=P=\varnothing\), and their domains remain disjoint.
A new private Alice answer is added to \(P\). If a learner queries
\(x\) and obtains \(y\), the experiment updates
\[
  h\leftarrow h\cup\{(x,y)\},\qquad
  P\leftarrow P|_{\cX\setminus\{x\}}.
\]
It records the query in Eve's view and leaves Alice's historical record
and residual state unchanged. Thus \(P\) records current private
oracle values, whereas \(T_{\cA}\) records all past queries and answers.
The oracle-answer rules are specified by each experiment.
Using \Cref{sec:state-notation}, for example,
\(\rho^{\mathsf X,e}_{AB}\) denotes the parties' state conditioned
on Eve's current view. The state \(\rho^{\mathsf X,e}_{BH}\) instead
contains Bob's register and a classical register storing the originally
sampled oracle, conditioned on that same view.

\subsection{External Lemmas}

Below we recall the two prior results used in our reprogramming attack.

The first is the random reprogramming
lemma in the form used by Katz and Sela~\cite[Lemma~3]{KatzSela24}, itself a
restatement of a result by Alagic et al.~\cite[Lemma~3]{AlagicEtAl22}.

\begin{lemma}[Random reprogramming lemma]\label{lem:reprogramming}
Let \(\cD\) be a distinguisher in the following game.
\begin{enumerate}[leftmargin=2em]
  \item \(\cD\) outputs a full function
  \(F_0:\cX\to\cY\) and a possibly inefficient randomized classical
  algorithm \(C\) whose output is a
  partial function \(R\).  Let
  \[
    \varepsilon
    =
    \max_{x\in\cX}
    \Pr_{R\leftarrow C}[x\in\Dom(R)].
  \]
  \item The challenger samples \(R\leftarrow C\), sets
  \(F_1=F_0[R]\), chooses \(b\leftarrow\bits\), and gives \(\cD\) quantum
  oracle access to \(F_b\).
  \item The oracle access is removed.  The challenger reveals the random
  coins used by \(C\), and \(\cD\) outputs a bit.
\end{enumerate}
If \(\cD\) makes at most \(q\) challenge-oracle queries, then
\[
  \left|
    \Pr[\cD\text{ outputs }1\mid b=1]
    -
    \Pr[\cD\text{ outputs }1\mid b=0]
  \right|
  \le 2q\sqrt\varepsilon.
\]
\end{lemma}

The restriction that \(C\) is classical is deliberate.  Its only role is to
sample the classical partial function \(R\), and the game later reveals the
random coins used for that sample.  In our applications, \(C\) samples a
classical query transcript, while \(\cD\) separately prepares any residual
quantum state conditioned on the revealed transcript.  Thus \(C\) never
simulates Alice's quantum state.  Allowing \(C\) itself to retain quantum
state would also require a different formulation of the game, since revealing
classical random coins would no longer specify its complete internal state.

The second is the following heavy-query counting lemma of Austrin et
al.~\cite[Lemma~3.4]{AustrinEtAl22}, in the form recalled by Katz and
Sela~\cite[Lemma~4]{KatzSela24}.

\begin{lemma}[Heavy-query bound]\label{lem:heavy-bound}
Let \(L,z_1,x_1,\ldots,z_q,x_q\) be correlated random variables, where
\(L\) ranges over subsets of a universe \(\cU\), each \(x_i\in\cU\), and
the \(x_i\)'s are pairwise distinct.  Conditioned on a history
\(z_1,x_1,\ldots,z_i\), call \(x_i\) \(\delta\)-heavy if
\[
  \Pr[x_i\in L\mid z_1,x_1,\ldots,z_i]\ge\delta.
\]
Let \(S\) be the set of \(\delta\)-heavy points among the \(x_i\)'s.  Then
\[
  \mathbb{E}[|S|]\le\frac{\mathbb{E}[|L|]}{\delta}.
\]
\end{lemma}

\subsection{Public Key Cryptography with Quantum Parties}\label{sec:model}

We now define the protocols used throughout
the paper. We note that the two-message model in \Cref{def:protocol} is close to the CC1QM
model of Bouaziz-Ermann et al.~\cite{BouazizEtAl24}, with one additional
restriction: before sending her first message, Alice has only classical
oracle access.

\subsubsection{Two-Message Key Agreement}

\begin{definition}[Two-message key agreement with a classical first message]
\label{def:protocol}
A two-message key-agreement protocol with a classical first message and a
possibly quantum response is a tuple of oracle quantum algorithms \(\Pi=(\cA_1,\cB,\cA_2)\) with the
following syntax relative to \(\cH_\lambda\). 
\begin{itemize}[leftmargin=2em]
  \item \(\cA_1^H(1^\lambda)\) has classical oracle access and outputs
  \((T_{\cA},M_1,S_{\cA})\).  Here \(M_1\) is a classical first message,
  \(S_{\cA}\) is Alice's residual register, and \(T_{\cA}\) is the complete
  ordered classical query-answer transcript of \(\cA_1\).
  \item \(\cB^H(1^\lambda,M_1)\) has quantum oracle access and outputs a
  possibly quantum message register \(\Psi\) and a classical key
  \(K_{\cB}\in\cK_\lambda\cup\{\bot\}\).
  \item \(\cA_2^H(1^\lambda,S_{\cA},\Psi)\) has quantum oracle access and
  outputs a classical key \(K_{\cA}\in\cK_\lambda\cup\{\bot\}\).
\end{itemize}
Other than access to the common random oracle, Alice and Bob initially share
no state or correlated randomness.  We assume worst-case bounds
\[
  |T_{\cA}|\le q_{\cA,1},
  \qquad
  q(\cB)\le q_{\cB},
  \qquad
  q(\cA_2)\le q_{\cA,2}.
\] 
The protocol is the following experiment, denoted by \(\Real_\Pi(1^\lambda)\):
\begin{enumerate}[leftmargin=2em]
  \item Sample \(H\leftarrow\cH_\lambda\).
  \item Run
  \((T_{\cA},M_1,S_{\cA})\leftarrow\cA_1^H(1^\lambda)\), and send the
  classical message \(M_1\) to Bob.
  \item Run
  \((\Psi,K_{\cB})\leftarrow\cB^H(1^\lambda,M_1)\), and send the
  message register \(\Psi\) to Alice.
  \item Run
  \(K_{\cA}\leftarrow\cA_2^H(1^\lambda,S_{\cA},\Psi)\).
\end{enumerate}
The valid agreement probability is
\[
  \alpha_\Pi(\lambda)
  :=
  \Pr_{\Real_\Pi(1^\lambda)}
  [K_{\cA}=K_{\cB}\neq\bot].
\]
We use the following notions of completeness and security for such a protocol:

\begin{enumerate}[leftmargin=2em,label=\textup{(\roman*)}]
  \item \textbf{Completeness.}  The protocol has completeness error
  \(\epsilon_c(\lambda)\) if
  \[
    \alpha_\Pi(\lambda)\ge 1-\epsilon_c(\lambda).
  \]
  We say the protocol is \emph{correct} if \(\epsilon_c\) is negligible.
  \item \textbf{Security.}  The protocol is interception key-recovery secure,
  as defined in \Cref{def:intkr}, if, for every polynomially query-bounded attacker
  \(\cE\),
  \begin{equation}
    \Adv^{\IntKR}_{\Pi,\cE}(\lambda) \leq \negl(\lambda).
  \end{equation}

\end{enumerate}
\end{definition}

\begin{definition}[Intercepting attacker]
An \emph{intercepting-attacker} is a pair
\(\cE=(\cE_{1},\cE_{2})\). Its first stage receives
Alice's classical message, may query the oracle before Bob generates his
response, and outputs a possibly quantum state \(Z_{\cE}\). Its second stage
receives \(Z_{\cE}\) together with Bob's message register and outputs a guess for the
classical key.
\end{definition}

\begin{definition}[Interception key-recovery experiment]\label{def:intkr}
Let \(\cE=(\cE_{1},\cE_{2})\) be any polynomially
query-bounded intercepting-attacker: its computation may be
unbounded, but its two stages together make at most \(\poly(\lambda)\)
random-oracle queries.  The experiment
\(\IntKR_{\Pi,\cE}(1^\lambda)\) is:
\begin{enumerate}[leftmargin=2em]
  \item Sample \(H\leftarrow\cH_\lambda\).
  \item Run
  \((T_{\cA},M_1,S_{\cA})\leftarrow\cA_1^H(1^\lambda)\) and reveal
  \(M_1\) to \(\cE_{1}\).
  \item Run
  \(Z_{\cE}\leftarrow\cE_{1}^H(1^\lambda,M_1)\).
  \item Run
  \((\Psi,K_{\cB})\leftarrow\cB^H(1^\lambda,M_1)\) and transfer
  \(\Psi\) to \(\cE_{2}\).
  \item Run
  \(K^*\leftarrow\cE_{2}^H(1^\lambda,Z_{\cE},\Psi)\).
  The experiment outputs \(1\) iff \(K^*=K_{\cB}\neq\bot\).
\end{enumerate}
Define
\[
  \Adv^{\IntKR}_{\Pi,\cE}(\lambda)
  :=
  \Pr[\IntKR_{\Pi,\cE}(1^\lambda)=1].
\]
\end{definition}

\subsubsection{Multi-Round Key Agreement with a Classical-Query Prefix}

We also consider protocols with an arbitrary classical-communication prefix.
During the prefix, both parties query the oracle only
classically, although their other computations and residual states may be
quantum. In the final round, Bob sends a possibly quantum message, after
which Alice performs her final computation. This communication pattern is similar to
the CC1QM pattern of Bouaziz-Ermann et
al.~\cite[Definition~11]{BouazizEtAl24}, except that we additionally require classical
oracle access during the prefix.

\begin{definition}[Multi-round key agreement with a classical-query prefix]\label{def:multi-protocol}

A key-agreement protocol with a classical prefix and one final possibly
quantum message is a tuple of oracle quantum algorithms
\[
  \Pi^{\mathsf{mr}}
  =
  \bigl((\cA_i)_{i=1}^{r},(\cB_i)_{i=1}^{r-1},
    \cB_{\mathsf q},\cA_{\mathsf{fin}}\bigr).
\]
Here \(r\ge1\) is the number of messages that Alice sends during the
classical prefix. To simplify the indexing, we assume that Alice sends both
the first and last prefix messages.
The honest experiment, denoted
\(\Real_{\Pi^{\mathsf{mr}}}(1^\lambda)\), is:
\begin{enumerate}[leftmargin=2em]
  \item Sample \(H\leftarrow\cH_\lambda\). Alice and Bob begin in an
  oracle-independent product state
  \(\rho_{\cA,0}\otimes\rho_{\cB,0}\).
  \item For \(i=1,\ldots,r\), Alice runs \(\cA_i\) on her current residual
  register and the public transcript so far, makes only classical oracle
  queries, sends a classical message \(M_{2i-1}\), and retains her new
  residual register. If \(i<r\), Bob proceeds analogously: he runs
  \(\cB_i\), makes only classical oracle queries, sends a classical message
  \(M_{2i}\), and retains his new residual register.
  \item Let
  \[
    M=(M_1,\ldots,M_{2r-1})
  \]
  be the complete classical prefix transcript. Let \(T_{\cA}\) and
  \(T_{\cB}\) be the complete ordered query-answer transcripts accumulated by
  Alice and Bob during the prefix, and let \(S_{\cA}\) and \(S_{\cB}\) be
  their residual quantum registers at its end.
  \item Bob runs
  \[
    (\Psi,K_{\cB})
    \leftarrow
    \cB_{\mathsf q}^{H}
    (1^\lambda,M,S_{\cB}),
  \]
  where \(\Psi\) is a possibly quantum message register and
  \(K_{\cB}\in\cK_\lambda\cup\{\bot\}\) is classical.
  \item After receiving \(\Psi\), Alice runs
  \[
    K_{\cA}
    \leftarrow
    \cA_{\mathsf{fin}}^{H}
    (1^\lambda,M,S_{\cA},\Psi).
  \]
  Here \(K_{\cA}\in\cK_\lambda\cup\{\bot\}\) is classical.
\end{enumerate}
Write
\[
  q_{\cA,\mathsf{pre}}
  :=\max |T_{\cA}|,
  \qquad
  q_{\cB,\mathsf{pre}}
  :=\max |T_{\cB}|,
\]
and let \(q_{\cB,\mathsf q}\) and
\(q_{\cA,\mathsf{fin}}\) bound the quantum-oracle queries made by
\(\cB_{\mathsf q}\) and \(\cA_{\mathsf{fin}}\), respectively. The valid
agreement probability is
\[
  \alpha_{\Pi^{\mathsf{mr}}}(\lambda)
  :=
  \Pr_{\Real_{\Pi^{\mathsf{mr}}}(1^\lambda)}
  [K_{\cA}=K_{\cB}\neq\bot].
\]
We use the following notions of completeness and security for such a protocol:
\begin{enumerate}[leftmargin=2em,label=\textup{(\roman*)}]
  \item \textbf{Completeness.}  The protocol has completeness error
  \(\epsilon_c(\lambda)\) if
  \[
    \alpha_{\Pi^{\mathsf{mr}}}(\lambda)
    \ge 1-\epsilon_c(\lambda).
  \]
  We say the protocol is \emph{correct} if \(\epsilon_c\) is negligible.
  \item \textbf{Security.}  The protocol is online interception key-recovery
  secure, as defined in \Cref{def:multi-intkr}, if, for every
  computationally unbounded attacker \(\cE\) making at most \(\poly(\lambda)\)
  random-oracle queries, restricted to classical queries during the prefix
  but allowed quantum queries after receiving \(\Psi\),
  \begin{equation}
    \Adv^{\IntKR,\mathsf{mr}}_{\Pi^{\mathsf{mr}},\cE}(\lambda)
    \leq \negl(\lambda).
  \end{equation}
\end{enumerate}
\end{definition}

\begin{definition}[Online interception experiment]
\label{def:multi-intkr}
The experiment
\(\IntKR^{\mathsf{mr}}_{\Pi^{\mathsf{mr}},\cE}(1^\lambda)\) proceeds as
follows.
\begin{enumerate}[leftmargin=2em]
  \item Sample \(H\leftarrow\cH_\lambda\) and initialize the parties and Eve.
  \item Run the honest classical prefix. After each public message, activate
  Eve on the transcript so far and its retained state; during these
  activations Eve may make classical oracle queries and update its state.
  \item Run \(\cB_{\mathsf q}^H\) on Bob's residual register. Transfer its
  message register \(\Psi\) to Eve and retain Bob's key \(K_{\cB}\).
  \item Eve may make quantum oracle queries and outputs a classical key guess
  \(K^*\). The experiment outputs \(1\) iff
  \(K^*=K_{\cB}\neq\bot\).
\end{enumerate}
Define
\[
  \Adv^{\IntKR,\mathsf{mr}}_{\Pi^{\mathsf{mr}},\cE}(\lambda)
  :=
  \Pr[\IntKR^{\mathsf{mr}}_{\Pi^{\mathsf{mr}},\cE}(1^\lambda)=1].
\]
\end{definition}

\subsubsection{Quantum Public-Key Encryption}

We consider classical public and secret keys and classical messages whose
length is bounded by a polynomial in \(\lambda\). We allow imperfect
correctness and two patterns of oracle access. In the first, key generation
has classical oracle access, while encryption, decryption, and ciphertexts
may be quantum. In the second, encryption has classical oracle access and
ciphertexts are classical, while key generation and decryption may make
quantum queries.
Our security experiment is the corresponding query-bounded IND-CPA
experiment: the attacker may perform unbounded computation but must respect
a bound on its oracle queries. When \(\Gen\), \(\Dec\), and the ciphertext are
classical but \(\Enc\) may be quantum, the \(\mu(\lambda)=1\) case
is exactly the PKE syntax used by Bartusek and
Khurana~\cite[Definition~7.8]{BartusekKhurana24}.

\begin{definition}[QPKE with classical-query key generation]
\label{def:qpke}
Let \(\mu:\mathbb{N}\to\mathbb{N}\) be any polynomial-time computable
function satisfying \(1\le\mu(\lambda)\le p(\lambda)\) for some polynomial
\(p\), and let
\(\cM_\lambda=\bits^{\mu(\lambda)}\). A quantum public-key encryption scheme
\(\Sigma=(\Gen,\Enc,\Dec)\) for the message spaces
\(\{\cM_\lambda\}_{\lambda\in\mathbb N}\), in the QROM with oracle family
\(\cH_\lambda\), has the following syntax.
\begin{itemize}[leftmargin=2em]
  \item \((\pk,\sk)\leftarrow\Gen^H(1^\lambda)\) outputs a classical
  public key and secret key.  The algorithm has classical random-oracle
  access, and its complete ordered query-answer transcript can be recorded as
  \(T_{\Gen}\).
  \item \(C\leftarrow\Enc^H(1^\lambda,\pk,m)\), for
  \(m\in\cM_\lambda\), uses quantum random-oracle access and outputs a
  possibly quantum ciphertext register \(C\).
  \item \(\widehat m\leftarrow\Dec^H(1^\lambda,\sk,C)\) uses quantum
  random-oracle access and outputs
  \(\widehat m\in\cM_\lambda\cup\{\bot\}\).
\end{itemize}
Write \(q_{\Gen}\), \(q_{\Enc}\), and \(q_{\Dec}\) for the respective
worst-case query bounds. Define the minimum correct-decryption probability
\begin{align*}
  \alpha_\Sigma(\lambda)
  :=\min_{m\in\cM_\lambda}
  \Pr\bigl[&\Dec^H(1^\lambda,\sk,C)=m:\ \\
    &H\leftarrow\cH_\lambda,\
    (\pk,\sk)\leftarrow\Gen^H(1^\lambda),\
    C\leftarrow\Enc^H(1^\lambda,\pk,m)\bigr].
\end{align*}
The scheme has correctness error \(\epsilon_c(\lambda)\) if
\(\alpha_\Sigma(\lambda)\ge1-\epsilon_c(\lambda)\).
\end{definition}

\begin{definition}[QPKE with classical-query encryption]
\label{def:qpke-classical-encryption}
A QPKE scheme with classical-query encryption has the syntax, classical
keys, message spaces, and correctness requirement of \Cref{def:qpke},
with the following changes to oracle access and ciphertexts. Key
generation and decryption may make quantum oracle queries. Encryption
makes only classical oracle queries and outputs a classical ciphertext.
Its local computation may otherwise be quantum. We use the same notation
\(q_{\Gen},q_{\Enc},q_{\Dec}\) and \(\alpha_\Sigma\) for its
query bounds and minimum correct-decryption probability.
\end{definition}

Both settings use the following security experiment.

\begin{definition}[Query-bounded IND-CPA security]\label{def:qpke-ind}
For a two-stage attacker \(\cD=(\cD_0,\cD_1)\), the experiment
\(\IND_{\Sigma,\cD}(1^\lambda)\) is:
\begin{enumerate}[leftmargin=2em]
  \item Sample \(H\leftarrow\cH_\lambda\) and
  \((\pk,\sk)\leftarrow\Gen^H(1^\lambda)\).
  \item Run
  \((Z_{\cD},m_0,m_1)\leftarrow\cD_0^H(1^\lambda,\pk)\), where
  \(Z_{\cD}\) is a possibly quantum state and
  \(m_0,m_1\in\cM_\lambda\).
  \item Sample \(b\leftarrow\bits\), run
  \(C\leftarrow\Enc^H(1^\lambda,\pk,m_b)\), and transfer \(C\) to
  \(\cD_1\).
  \item Run \(b'\leftarrow\cD_1^H(1^\lambda,Z_{\cD},C)\).  The experiment
  outputs \(1\) iff \(b'=b\).
\end{enumerate}
Define
\[
  \Adv^{\IND\text{-}\mathrm{CPA}}_{\Sigma,\cD}(\lambda)
  :=
  \left|\Pr[\IND_{\Sigma,\cD}(1^\lambda)=1]-\frac12\right|.
\]
The scheme is query-bounded IND-CPA secure if this advantage is negligible
for every computationally unbounded attacker making at most
\(\poly(\lambda)\) random-oracle queries. This differs from standard
computational IND-CPA only in the attacker's resource bound. Because
encryption is public, the experiment does not need a separate encryption
oracle.
\end{definition}

\section{Conditional Resampling and Heavy-Query Learning}\label{sec:resampling}

\subsection{The Heavy-Query Attacker}

The following procedure is the \textsc{findTranscript} heavy-query attacker
of Katz and Sela~\cite[Section~3 and Lemma~5]{KatzSela24} adapted to our setting. Its cutoff analysis is precisely
their application of the heavy-query counting bound of Austrin et
al.~\cite{AustrinEtAl22}.

Throughout this subsection, probabilities are taken over
\[
  H\leftarrow\cH,
  \qquad
  (T_{\cA},M_1,S_{\cA})\leftarrow\cA_1^H(1^\lambda).
\]
For a first-message value \(m_1\) and a partial oracle \(h\), let
\[
  \mathsf E_{m_1,h}:=\{M_1=m_1,\ h\preceq H\}.
\]
Whenever \(\Pr[\mathsf E_{m_1,h}]>0\), define
\begin{equation}\label{eq:p-heavy}
  p_{m_1,h}(x)
  :=
  \Pr[x\in\Dom(T_{\cA})\mid\mathsf E_{m_1,h}].
\end{equation}
Define
\[
  \Heavy_\delta(m_1,h)
  :=
  \{x\in\cX\setminus\Dom(h):p_{m_1,h}(x)>\delta\}.
\]
Fix a total order \(\prec\) on \(\cX\).

\begin{algorithm}[H]
\caption{\(\Learn_\delta^H(1^\lambda,m_1)\)}
\label{alg:learn}
\KwInput{Security parameter \(1^\lambda\), first message \(m_1\), oracle
access to \(H\), threshold \(\delta\), and bound \(q_{\cA,1}\).}
\KwOutput{A partial oracle \(h\) or \(\bot\).}
\(N_{\max}\gets\ceil{q_{\cA,1}/\delta^2}\); \(h\gets\emptyset\)\;
\While{\(\Heavy_\delta(m_1,h)\neq\emptyset\)}{
  \If{\(|\Dom(h)|\ge N_{\max}\)}{\Return \(\bot\)\;}
  Let \(x\) be the \(\prec\)-least point in \(\Heavy_\delta(m_1,h)\)\;
  Query \(H(x)\) and set \(h\gets h\cup\{x\mapsto H(x)\}\)\;
}
\Return \(h\)\;
\end{algorithm}

A pair \((m_1,h)\) is \emph{reachable} if
\[
  \Pr[M_1=m_1,\ \Learn_\delta^H(1^\lambda,m_1)=h]>0.
\]

As in the Katz-Sela attacker, fixing the order \(\prec\) makes every next
query a deterministic function of \((m_1,h)\). Hence, for every reachable
pair \((m_1,h)\) and every oracle \(H\),
\[
  \Learn_\delta^H(1^\lambda,m_1)=h
  \quad\Longleftrightarrow\quad
  h\preceq H.
\]
We refer to this equivalence as the \emph{deterministic-output property} of
\(\Learn\).

\subsection{The Conditioned View of the First Round}

For every reachable pair \((m_1,h)\), let \(\omega_{m_1,h}\)
denote the conditional state of Alice's view after the first round
\((T_{\cA},S_{\cA})\) given \(\mathsf E_{m_1,h}\).  Write
\[
  \mu_{m_1,h}(\tau)
  =
  \Pr[T_{\cA}=\tau\mid\mathsf E_{m_1,h}]
\]
for its transcript marginal.  Explicitly,
for every \(\tau\) with \(\mu_{m_1,h}(\tau)>0\), let
\(\sigma_{m_1,\tau}\) be the corresponding conditional state of
\(S_{\cA}\) given \((M_1,T_{\cA})=(m_1,\tau)\). Then
\[
  \omega_{m_1,h}
  =
  \sum_{\substack{\tau:\\
        \mu_{m_1,h}(\tau)>0}}
  \mu_{m_1,h}(\tau)
  \proj{\tau}_{T_{\cA}}\otimes\sigma_{m_1,\tau}.
\]
The notation
\((T'_{\cA},S'_{\cA})\leftarrow\omega_{m_1,h}\) means that we prepare a new
instance of this state: sample \(\tau\leftarrow\mu_{m_1,h}\), place \(\tau\)
in \(T'_{\cA}\), and prepare \(S'_{\cA}\) in the state
\(\sigma_{m_1,\tau}\). By the deterministic-output property, the events
\(\mathsf E_{m_1,h}\) and
\(\{M_1=m_1,\Learn_\delta^H(1^\lambda,m_1)=h\}\) coincide for every
reachable pair.

\begin{lemma}[Katz-Sela heavy-query attacker]\label{lem:learn-success}
For every \(0<\delta<1\), with probability at least \(1-\delta\) over the
real execution of the first round, \(\Learn_\delta^H(1^\lambda,M_1)\) outputs a
partial oracle \(h\neq\bot\).  On every nonabort execution,
\[
  |\Dom(h)|\le\ceil{q_{\cA,1}/\delta^2}
\]
and, for every \(x\notin\Dom(h)\),
\[
  \Pr_{(T'_{\cA},S'_{\cA})\leftarrow\omega_{M_1,h}}
  [x\in\Dom(T'_{\cA})]
  \le\delta.
\]
\end{lemma}

\begin{proof}
The argument is a restatement, in our notation, of the proof of Katz and
Sela~\cite[Lemma~5]{KatzSela24}; we include it explicitly for completeness.
The residual register \(S_{\cA}\) plays no role in the counting argument and
is simply prepared together with its sampled transcript.

The second claim is immediate upon successful termination:
\(\Heavy_\delta(M_1,h)=\emptyset\), and the displayed probability is exactly
\(p_{M_1,h}(x)\).

We now bound the abort probability. If \(q_{\cA,1}=0\), the attacker
returns the empty partial oracle. Assume \(q_{\cA,1}>0\), and let
\(N_{\max}=\ceil{q_{\cA,1}/\delta^2}\). Let
\(L=\Dom(T_{\cA})\). Let \(N\le N_{\max}\) be the number of real queries
issued by the attacker before it returns or aborts. On abort,
\(N=N_{\max}\); successful termination with \(N=N_{\max}\) is also allowed.

To apply \Cref{lem:heavy-bound} with a fixed-length sequence, set
\(z_1=M_1\). For each real attacker step \(i\le N\), let \(x_i\) be its
query and let \(z_{i+1}=H(x_i)\). For \(N<i\le N_{\max}\), pad with pairwise
distinct symbols \(x_i=\bot_i\notin\cX\) and fixed values
\(z_{i+1}=0\). Work over the enlarged universe
\(\cX\sqcup\{\bot_1,\ldots,\bot_{N_{\max}}\}\). The points are pairwise
distinct because the attacker never repeats a query.

Immediately before a real query \(x_i\), the history determines the current
partial oracle \(h_{i-1}\). Hence
\[
  \Pr[x_i\in L\mid z_1,x_1,\ldots,z_i]
  =p_{M_1,h_{i-1}}(x_i)>\delta.
\]
Every real attacker query is therefore \(\delta\)-heavy, while every padding
point has conditional membership probability zero. The heavy set in
\Cref{lem:heavy-bound} has size exactly \(N\), and so
\[
  \mathbb{E}[N]
  \le
  \frac{\mathbb{E}[|L|]}{\delta}
  \le
  \frac{q_{\cA,1}}{\delta}.
\]
By Markov's inequality,
\[
  \Pr[N\ge N_{\max}]
  \le
  \frac{q_{\cA,1}}{\delta N_{\max}}
  \le\delta.
\]
Since abort implies \(N=N_{\max}\),
\[
  \Pr[\Learn_\delta^H(1^\lambda,M_1)=\bot]
  \le \Pr[N\ge N_{\max}]
  \le \delta.
\]
\end{proof}

\section{Attacks on Key Agreement with a Classical Prefix}\label{sec:main-attack}

\subsection{Two-Message Protocols}

\begin{definition}[Eve's two-stage reprogramming attack]\label{def:eve}
Fix \(0<\delta<1\).  The attacker
\(\cE_\delta=(\cE_{\delta,1},
\cE_{\delta,2})\) is defined as follows.

On input the first message \(m_1\), the preprocessing stage runs
\[
  h\leftarrow\Learn_\delta^H(1^\lambda,m_1).
\]
If \(h=\bot\), it outputs \(Z_{\cE}=\bot\).  Otherwise it prepares
\[
  (T'_{\cA},S'_{\cA})\leftarrow\omega_{m_1,h}
\]
and outputs \(Z_{\cE}=(m_1,h,T'_{\cA},S'_{\cA})\).

After receiving Bob's message register \(\Psi\), the postprocessing stage
outputs \(\bot\) if \(Z_{\cE}=\bot\).  Otherwise it defines the fixed
reprogrammed oracle \(H'=H[T'_{\cA}]\), runs
\[
  K^*\leftarrow
  \cA_2^{H'}(1^\lambda,S'_{\cA},\Psi),
\]
and outputs \(K^*\).
\end{definition}

The preprocessing stage makes at most
\(\ceil{q_{\cA,1}/\delta^2}\) classical queries. Each query to
\(H[T'_{\cA}]\) can be implemented using one query to \(H\), a reversible
lookup in the known transcript, and a \(\ket{+}^{\otimes\ell}\) ancilla; see
the proof of \Cref{lem:g3-g4}. The postprocessing stage therefore makes at
most \(q_{\cA,2}\) quantum queries.

\begin{theorem}\label{thm:main}
Let \(\Pi=(\cA_1,\cB,\cA_2)\) be a protocol as in
\Cref{def:protocol}, and let
\(\alpha=\alpha_\Pi(\lambda)\).  For every \(0<\delta<1\), the attacker
\(\cE_\delta\) of \Cref{def:eve} satisfies
\[
  \Adv^{\IntKR}_{\Pi,\cE_\delta}(\lambda)
  \ge
  \alpha
  -\delta
  -2(2q_{\cB}+q_{\cA,2})\sqrt\delta.
\]
It makes at most
\[
  \ceil{q_{\cA,1}/\delta^2}+q_{\cA,2}
\]
random-oracle queries.
\end{theorem}

\begin{proof}[Proof of \Cref{thm:main}]
We proceed via a hybrid argument. Consider the hybrids \(G_0,G_1,\ldots,G_5\), where $G_0$ is the honest
execution and $G_5$ is Eve's attack. We say that \(G_i\) \emph{succeeds} when
its two recorded keys agree and are not \(\bot\) (an abort counts as failure).

\begin{hybrid}{Hybrid \(G_0\) (real execution)}
  \item Sample \(H_0\leftarrow\cH\).
  \item Run
  \((T_{\cA,0},M_{1,0},S_{\cA,0})
  \leftarrow\cA_1^{H_0}(1^\lambda)\).
  \item Run
  \((\Psi_0,K_{\cB,0})
  \leftarrow\cB^{H_0}(1^\lambda,M_{1,0})\).
  \item Run
  \(K_{\cA,0}
  \leftarrow\cA_2^{H_0}(1^\lambda,S_{\cA,0},\Psi_0)\).
\end{hybrid}
Therefore,
\[
  \Pr[G_0\text{ succeeds}]=\alpha.
\]

\begin{hybrid}{Hybrid \(G_1\)}
  \item Sample \(H_1\leftarrow\cH\).
  \item Run
  \((T_{\cA,1},M_{1,1},S_{\cA,1})
  \leftarrow\cA_1^{H_1}(1^\lambda)\).
  \item \changed{Run
  \(h_1\leftarrow\Learn_\delta^{H_1}(1^\lambda,M_{1,1})\); abort if
  \(h_1=\bot\).}
  \item \changed{Sample a fresh oracle
  \(\widehat H_1\leftarrow\Unif(\cH(h_1))\) and set
  \(\widetilde H_1=\widehat H_1[T_{\cA,1}]\).}
  \item Run
  \((\Psi_1,K_{\cB,1})
  \leftarrow\cB^{\changed{\widetilde H_1}}(1^\lambda,M_{1,1})\).
  \item Run
  \(K^*_1
  \leftarrow\cA_2^{\changed{\widetilde H_1}}
  (1^\lambda,S_{\cA,1},\Psi_1)\). 
\end{hybrid}

\begin{lemma}[\(G_0\) to \(G_1\)]\label{lem:g0-g1}
\[
  \left|
  \Pr[G_0\text{ succeeds}]
  -
  \Pr[G_1\text{ succeeds}]
  \right|
  \le\delta.
\]
\end{lemma}

\begin{proof}
We begin by adding an intermediate hybrid \(G'_0\) that inserts into \(G_0\) an additional call
\(h\leftarrow\Learn_\delta^{H_0}(1^\lambda,M_{1,0})\) whose output is otherwise
ignored. This computation does
not change the real execution. By
\Cref{lem:learn-success}, the call aborts with probability at most
\(\delta\).

Condition on the attacker not aborting, and fix values
\((m_1,\tau,h)\) with positive probability. By the deterministic-output
property of \(\Learn\),
\[
  \{\Learn_\delta^H(1^\lambda,m_1)=h\}
  =
  \{h\preceq H\}.
\]
Because \(\cA_1\) makes only classical oracle queries, once
\((M_1,T_{\cA})=(m_1,\tau)\) is fixed, its conditional residual state
depends only on the recorded query-answer history, not on any value of
\(H\) outside \(\Dom(\tau)\). Since the original oracle is uniform,
conditioning on
\[
  M_1=m_1,\qquad T_{\cA}=\tau,\qquad h\preceq H
\]
leaves \(H\) uniform over \(\cH(h,\tau)\) and independent of
\(S_{\cA}\). This is the classical-quantum case of the
conditional-oracle identity used by Katz and
Sela~\cite[Lemma~6]{KatzSela24}.

In \(G_1\),
\(\widehat H_1\leftarrow\Unif(\cH(h))\) is independent of Alice's conditioned
residual state, and \(\widehat H_1[\tau]\) is uniform over \(\cH(h,\tau)\). The
conditioned experiments are therefore identical. The only loss comes from
the abort event, which has probability at most \(\delta\).
\end{proof}

\begin{hybrid}{Hybrid \(G_2\)}
  \item Sample \(H_2\leftarrow\cH\).
  \item Run
  \((T_{\cA,2},M_{1,2},S_{\cA,2})
  \leftarrow\cA_1^{H_2}(1^\lambda)\).
  \item Run
  \(h_2\leftarrow\Learn_\delta^{H_2}(1^\lambda,M_{1,2})\); abort if
  \(h_2=\bot\).
  \item Sample \(\widehat H_2\leftarrow\Unif(\cH(h_2))\).
  \item \changed{Prepare
  \((T'_{\cA,2},S'_{\cA,2})\leftarrow\omega_{M_{1,2},h_2}\) and set
  \(\widetilde H_2=\widehat H_2[T'_{\cA,2}]\).}
  \item Run
  \((\Psi_2,K_{\cB,2})
  \leftarrow\cB^{\widetilde H_2}(1^\lambda,M_{1,2})\).
  \item Run
  \(K^*_2
  \leftarrow\cA_2^{\widetilde H_2}
  (1^\lambda,S'_{\cA,2},\Psi_2)\).
\end{hybrid}

\begin{lemma}[\(G_1\) to \(G_2\)]\label{lem:g1-g2}
\[
  \Pr[G_1\text{ succeeds}]
  =
  \Pr[G_2\text{ succeeds}].
\]
\end{lemma}

\begin{proof}
Both hybrids abort on the same event. Condition on a reachable
nonabort pair \((m_1,h)\). By the deterministic-output property, the
conditioned state of Alice's actual view in \(G_1\) is exactly
\(\omega_{m_1,h}\); \(G_2\) simply prepares a new instance of that same
classical-quantum state. In both hybrids, the fresh oracle
\(\widehat H_i\leftarrow\Unif(\cH(h))\) is independent of the view, and is then
reprogrammed according to its transcript. The rest of the experiments is
identical and therefore their conditional success probabilities agree for every
reachable nonabort \((m_1,h)\). Since the two hybrids have the same
distribution on \((M_1,h)\) and both fail on abort, their success probabilities
are equal.
\end{proof}

\begin{hybrid}{Hybrid \(G_3\)}
  \item Sample \(H_3\leftarrow\cH\).
  \item Run
  \((T_{\cA,3},M_{1,3},S_{\cA,3})
  \leftarrow\cA_1^{H_3}(1^\lambda)\).
  \item Run
  \(h_3\leftarrow\Learn_\delta^{H_3}(1^\lambda,M_{1,3})\); abort if
  \(h_3=\bot\).
  \item Sample \(\widehat H_3\leftarrow\Unif(\cH(h_3))\).
  \item Prepare
  \((T'_{\cA,3},S'_{\cA,3})\leftarrow\omega_{M_{1,3},h_3}\) and set
  \(\widetilde H_3=\widehat H_3[T'_{\cA,3}]\).
  \item \changed{Run
  \((\Psi_3,K_{\cB,3})
  \leftarrow\cB^{\widehat H_3}(1^\lambda,M_{1,3})\).}
  \item Run
  \(K^*_3
  \leftarrow\cA_2^{\widetilde H_3}
  (1^\lambda,S'_{\cA,3},\Psi_3)\).
\end{hybrid}

\begin{lemma}[\(G_2\) to \(G_3\)]\label{lem:g2-g3}
\[
  \left|
  \Pr[G_2\text{ succeeds}]
  -
  \Pr[G_3\text{ succeeds}]
  \right|
  \le 2q_{\cB}\sqrt\delta.
\]
\end{lemma}

\begin{proof}
Fix a reachable nonabort pair \((m_1,h)\). We construct a distinguisher
\(\cD_{m_1,h}\) for the game in \Cref{lem:reprogramming} as follows.
\begin{enumerate}[leftmargin=2em]
  \item The distinguisher samples
  \(\widehat H\leftarrow\Unif(\cH(h))\), sends
  \(F_0:=\widehat H\) to the challenger, and supplies the randomized
  algorithm \(C_{m_1,h}\). On fresh coins, \(C_{m_1,h}\) samples
  \(\tau'\leftarrow\mu_{m_1,h}\) and returns
  \[
    R:=\tau'|_{\Dom(\tau')\setminus\Dom(h)}.
  \]
  \item The challenger samples the coins of \(C_{m_1,h}\), obtains \(R\),
  sets \(F_1:=F_0[R]\), samples \(b\leftarrow\bits\), and gives the
  distinguisher quantum oracle access to \(F_b\). The distinguisher runs
  \[
    (\Psi,K_{\cB})
    \leftarrow\cB^{F_b}(1^\lambda,m_1)
  \]
  and retains \((\Psi,K_{\cB})\).
  \item The challenger removes access to \(F_b\) and reveals the coins of
  \(C_{m_1,h}\), thereby revealing \(\tau'\). The distinguisher prepares
  \(S'_{\cA}\) in the conditional state \(\sigma_{m_1,\tau'}\), runs
  \[
    K^*\leftarrow
    \cA_2^{\widehat H[\tau']}
    (1^\lambda,S'_{\cA},\Psi),
  \]
  and outputs \(1\) iff \(K^*=K_{\cB}\neq\bot\).
\end{enumerate}
Only Bob's algorithm receives the challenge oracle, so the distinguisher
makes at most \(q_{\cB}\) challenge queries.

Every transcript in \(\Supp(\mu_{m_1,h})\) is consistent with \(h\),
so \(F_1=\widehat H[R]=\widehat H[\tau']\). The cases \(b=1\) and \(b=0\)
reproduce \(G_2\) and \(G_3\), respectively, conditioned on \((m_1,h)\).
Moreover, for every \(x\),
\[
  \Pr[x\in\Dom(R)]
  \le
  \begin{cases}
    0,&x\in\Dom(h),\\
    \Pr_{\tau'\leftarrow\mu_{m_1,h}}
      [x\in\Dom(\tau')]\le\delta,&x\notin\Dom(h),
  \end{cases}
\]
where the last inequality is \Cref{lem:learn-success}. The reprogramming
lemma bounds the conditional difference by \(2q_{\cB}\sqrt\delta\).
To remove the conditioning, let \(\pi(m_1,h)\) be the common subprobability
mass of nonabort executions of the shared prefix that output
\((M_1,h)=(m_1,h)\), and write
\[
  s_i(m_1,h)
  :=
  \Pr[G_i\text{ succeeds}\mid M_{1,i}=m_1,\ h_i=h]
  \qquad (i\in\{2,3\}).
\]
The preceding argument gives
\(\lvert s_2(m_1,h)-s_3(m_1,h)\rvert
\le 2q_{\cB}\sqrt\delta\) for every \((m_1,h)\in\Supp(\pi)\).
Because success is defined to be false on abort, the two success
probabilities are the corresponding \(\pi\)-weighted sums. Hence the triangle
inequality gives
\[
  \begin{aligned}
  \left|\Pr[G_2\text{ succeeds}]-\Pr[G_3\text{ succeeds}]\right|
  &\le
  \sum_{(m_1,h)\in\Supp(\pi)}\pi(m_1,h)
    \left|s_2(m_1,h)-s_3(m_1,h)\right|\\
  &\le 2q_{\cB}\sqrt\delta,
  \end{aligned}
\]
where \(\sum_{(m_1,h)\in\Supp(\pi)}\pi(m_1,h)\le1\).
\end{proof}

\begin{hybrid}{Hybrid \(G_4\)}
  \item Sample \(H_4\leftarrow\cH\).
  \item Run
  \((T_{\cA,4},M_{1,4},S_{\cA,4})
  \leftarrow\cA_1^{H_4}(1^\lambda)\).
  \item Run
  \(h_4\leftarrow\Learn_\delta^{H_4}(1^\lambda,M_{1,4})\); abort if
  \(h_4=\bot\).
  \item \changed{Sample
  \(\widehat H_4\leftarrow\Unif(\cH(h_4))\) and set Bob's oracle
  \(H^{\mathsf B}_4=\widehat H_4[T_{\cA,4}]\).}
  \item Prepare
  \((T'_{\cA,4},S'_{\cA,4})\leftarrow\omega_{M_{1,4},h_4}\) and set
  \(H^{\mathsf E}_4=H^{\mathsf B}_4[T'_{\cA,4}]\).
  \item Run
  \((\Psi_4,K_{\cB,4})
  \leftarrow\cB^{H^{\mathsf B}_4}(1^\lambda,M_{1,4})\).
  \item Run
  \(K^*_4
  \leftarrow\cA_2^{H^{\mathsf E}_4}
  (1^\lambda,S'_{\cA,4},\Psi_4)\).
\end{hybrid}

\begin{lemma}[\(G_3\) to \(G_4\)]\label{lem:g3-g4}
\[
  \left|
  \Pr[G_3\text{ succeeds}]
  -
  \Pr[G_4\text{ succeeds}]
  \right|
  \le 2(q_{\cB}+q_{\cA,2})\sqrt\delta.
\]
\end{lemma}

\begin{proof}
Fix a reachable nonabort pair \((m_1,h)\). In \(G_4\), the original
residual register \(S_{\cA,4}\) is unused and the conditional marginal of
\(T_{\cA,4}\) is \(\mu_{m_1,h}\). Thus the conditioned hybrid is unchanged
if we independently sample
\[
  \widehat H\leftarrow\Unif(\cH(h)),
  \qquad
  \tau''\leftarrow\mu_{m_1,h},
  \qquad
  (T'_{\cA},S'_{\cA})\leftarrow\omega_{m_1,h},
\]
and then run Bob with \(\widehat H[\tau'']\) and Alice's final algorithm with
\(\widehat H[\tau''][T'_{\cA}]\). As usual, the second reprogramming
overwrites the first wherever their domains overlap.

We construct a distinguisher \(\cD_{m_1,h}\) for the game in
\Cref{lem:reprogramming}.
\begin{enumerate}[leftmargin=2em]
  \item The distinguisher prepares
  \((T'_{\cA},S'_{\cA})\leftarrow\omega_{m_1,h}\), samples
  \(\widehat H\leftarrow\Unif(\cH(h))\), and sends
  \(F_0:=\widehat H\) to the challenger. It also supplies a randomized
  algorithm \(C_{m_1,h}\) that samples
  \(\tau''\leftarrow\mu_{m_1,h}\) and returns
  \[
    R:=\tau''|_{\Dom(\tau'')\setminus\Dom(h)}.
  \]
  \item The challenger samples the coins of \(C_{m_1,h}\), obtains \(R\),
  sets \(F_1:=F_0[R]\), samples \(b\leftarrow\bits\), and gives the
  distinguisher quantum oracle access to \(F_b\). The distinguisher runs
  \[
    (\Psi,K_{\cB})\leftarrow\cB^{F_b}(1^\lambda,m_1)
  \]
  followed by
  \[
    K^*\leftarrow
    \cA_2^{F_b[T'_{\cA}]}
    (1^\lambda,S'_{\cA},\Psi),
  \]
  and records the bit \(d:=\mathbf 1[K^*=K_{\cB}\neq\bot]\).
  \item The challenger removes access to \(F_b\) and reveals the coins of
  \(C_{m_1,h}\), thereby revealing \(\tau''\). The distinguisher outputs
  the recorded bit \(d\).
\end{enumerate}

Each query to \(F_b[T'_{\cA}]\) uses one query to \(F_b\). First compute
reversibly whether the query point lies in \(\Dom(T'_{\cA})\). On that
branch, swap the answer register with an ancilla in
\(\ket{+}^{\otimes\ell}\), apply the XOR oracle \(U_{F_b}\), and swap back.
Because every XOR translation fixes \(\ket{+}^{\otimes\ell}\), the challenge
oracle has no effect on the reprogrammed branch.  XOR the known value
\(T'_{\cA}(x)\) into the answer register there and uncompute the lookup.  This
implements \(U_{F_b[T'_{\cA}]}\) exactly. The total number of challenge queries
is at most \(q_{\cB}+q_{\cA,2}\).

If \(b=0\), this is \(G_3\). If \(b=1\), consistency on
\(\Dom(h)\) gives \(F_1=\widehat H[R]=\widehat H[\tau'']\), yielding the
equivalent form of \(G_4\). By \Cref{lem:learn-success},
\(\max_x\Pr[x\in\Dom(R)]\le\delta\). The reprogramming lemma gives the
pointwise bound
\[
  \left|
    \Pr[G_3\text{ succeeds}\mid M_1=m_1,\ h_3=h]
    -
    \Pr[G_4\text{ succeeds}\mid M_1=m_1,\ h_4=h]
  \right|
  \le 2(q_{\cB}+q_{\cA,2})\sqrt\delta
\]
for every reachable nonabort pair \((m_1,h)\). Let \(\pi(m_1,h)\) denote the
common subprobability mass of nonabort executions of the shared prefix that
produce that pair, and define
\[
  s_i(m_1,h)
  :=
  \Pr[G_i\text{ succeeds}\mid M_{1,i}=m_1,\ h_i=h]
  \qquad (i\in\{3,4\}).
\]
Because success is defined to be false on abort, the two success
probabilities are the corresponding \(\pi\)-weighted sums. The triangle
inequality yields
\[
  \begin{aligned}
  \left|\Pr[G_3\text{ succeeds}]-\Pr[G_4\text{ succeeds}]\right|
  &\le
  \sum_{(m_1,h)\in\Supp(\pi)}\pi(m_1,h)
    \left|s_3(m_1,h)-s_4(m_1,h)\right|\\
  &\le 2(q_{\cB}+q_{\cA,2})\sqrt\delta.
  \end{aligned}
\]
\end{proof}

\begin{hybrid}{Hybrid \(G_5\)}
  \item Sample \(H_5\leftarrow\cH\).
  \item Run
  \((T_{\cA,5},M_{1,5},S_{\cA,5})
  \leftarrow\cA_1^{H_5}(1^\lambda)\).
  \item Run
  \(h_5\leftarrow\Learn_\delta^{H_5}(1^\lambda,M_{1,5})\); abort if
  \(h_5=\bot\).
  \item Prepare
  \((T'_{\cA,5},S'_{\cA,5})\leftarrow\omega_{M_{1,5},h_5}\) and set
  \(\widetilde H_5=H_5[T'_{\cA,5}]\).
  \item Run
  \((\Psi_5,K_{\cB,5})
  \leftarrow\cB^{H_5}(1^\lambda,M_{1,5})\).
  \item Run
  \(K^*_5
  \leftarrow\cA_2^{\widetilde H_5}
  (1^\lambda,S'_{\cA,5},\Psi_5)\).
\end{hybrid}

\begin{lemma}[\(G_4\) to \(G_5\)]\label{lem:g4-g5}
\[
  \Pr[G_4\text{ succeeds}]
  =
  \Pr[G_5\text{ succeeds}].
\]
\end{lemma}

\begin{proof}
Both hybrids abort on the same event. On success, fix
\((m_1,\tau,h)\). By the deterministic-output property and the
conditional-oracle identity of Katz and Sela~\cite[Lemma~6]{KatzSela24},
conditioned on this triple the
original oracle \(H_5\) in \(G_5\) is uniform over \(\cH(h,\tau)\) and
independent of the unused original residual state. In \(G_4\), sampling
\(\widehat H_4\leftarrow\Unif(\cH(h))\) and setting
\(H^{\mathsf B}_4=\widehat H_4[\tau]\) gives the
same uniform distribution over \(\cH(h,\tau)\). The independently prepared
view \(\omega_{m_1,h}\) and every subsequent operation are the same.
Thus the conditional success probabilities of \(G_4\) and \(G_5\) are equal
for every reachable nonabort triple \((m_1,\tau,h)\). The two hybrids induce
the same distribution \(\pi(m_1,\tau,h)\) on these triples. Define
\[
  s_i(m_1,\tau,h)
  :=
  \Pr[G_i\text{ succeeds}\mid
    M_{1,i}=m_1,\ T_{\cA,i}=\tau,\ h_i=h]
  \qquad (i\in\{4,5\}).
\]
Then \(s_4(m_1,\tau,h)=s_5(m_1,\tau,h)\) throughout \(\Supp(\pi)\).
Both hybrids have success probability zero on abort, so the law of total
probability gives
\[
  \Pr[G_4\text{ succeeds}]
  =
  \sum_{m_1,\tau,h}\pi(m_1,\tau,h)
    s_4(m_1,\tau,h)
  =
  \sum_{m_1,\tau,h}\pi(m_1,\tau,h)
    s_5(m_1,\tau,h)
  =
  \Pr[G_5\text{ succeeds}].
\]
\end{proof}

Hybrid \(G_5\) is exactly \(\IntKR_{\Pi,\cE_\delta}\): the preprocessing
stage learns \(h\) and prepares the conditioned view; Bob runs with the
original oracle; Eve intercepts \(\Psi_5\); and Eve runs Alice's final
algorithm with \(H_5[T'_{\cA,5}]\).  Therefore
\begin{align*}
  \Adv^{\IntKR}_{\Pi,\cE_\delta}(\lambda)
  &=
  \Pr[G_5\text{ succeeds}]\\
  &\ge
  \Pr[G_0\text{ succeeds}]
  -\delta
  -2q_{\cB}\sqrt\delta
  -2(q_{\cB}+q_{\cA,2})\sqrt\delta\\
  &=
  \alpha-\delta-2(2q_{\cB}+q_{\cA,2})\sqrt\delta.
\end{align*}
The query bound follows from \Cref{def:eve}.
\end{proof}

\begin{corollary}[Concrete parameter choice]\label{cor:parameters}
Let \(\Pi\) satisfy \Cref{def:protocol}, let
\(\alpha=\alpha_\Pi(\lambda)>0\), and define
\[
  Q:=\max\{1,2q_{\cB}+q_{\cA,2}\},
  \qquad
  \delta:=\left(\frac{\alpha}{8Q}\right)^2.
\]
Then the adversary of \Cref{thm:main} succeeds with probability strictly
greater than \(\alpha/2\) and makes at most
\[
  \ceil{4096\,q_{\cA,1}Q^4/\alpha^4}+q_{\cA,2}
\]
random-oracle queries.
\end{corollary}

\begin{proof}
The reprogramming loss is
\[
  2Q\sqrt\delta=\alpha/4.
\]
Since \(0<\alpha\le1\) and \(Q\ge1\),
\(\delta=\alpha^2/(64Q^2)\le\alpha/64\).  Hence the success lower bound is
at least
\[
  \alpha-\alpha/64-\alpha/4>\alpha/2.
\]
Moreover,
\[
  \frac{q_{\cA,1}}{\delta^2}
  =
  \frac{4096\,q_{\cA,1}Q^4}{\alpha^4}.
\]
\end{proof}

\begin{corollary}[Polynomial-query two-message impossibility]\label{cor:quantum-impossibility}
Let \(\Pi=\{\Pi_\lambda\}\) be a family of protocols satisfying
\Cref{def:protocol}. Suppose that \(q_{\cA,1}\), \(q_{\cB}\), and
\(q_{\cA,2}\) are each bounded by \(\poly(\lambda)\), and that
\(\alpha_\Pi(\lambda)\ge1/\poly(\lambda)\). Then the family is not
interception key-recovery secure against computationally unbounded quantum
adversaries making at most \(\poly(\lambda)\) random-oracle queries. In
particular, the conclusion applies when the completeness error is negligible.
\end{corollary}

\subsection{Multi-Round Protocols}
\label{sec:multi-round}

We give an attack on a multi-round key agreement protocol with a classical prefix and one final possibly quantum message, as defined in 
\Cref{def:multi-protocol}. Recall that in this model, the only property the prefix must satisfy is that all communication and all queries to the random oracle are classical;
the parties' local computation and residual states may otherwise be
quantum.  Our attack and corresponding analysis extend the techniques from \cite{BarakMahmoodyFull} to our setting.

\subsubsection{Classical Views of the Prefix}
\label{sec:classical-prefix-views}

The main difference between our setting and the purely classical setting in \cite{BarakMahmoodyFull} is that we allow the parties' local computation to be quantum. However, as we discuss in this section, we can assume without loss of generality that Alice and Bob are executing classical algorithms to generate the prefix. In turn, we will be able to apply the techniques from \cite{BarakMahmoodyFull} directly. 

We first explain how to generate classical views of the prefix. Let $\Pi$ be a protocol as defined in Definition~\ref{def:multi-protocol}. Recall that while the parties must be query bounded, they may be computationally unbounded. We therefore can define an exact classical
simulation of the prefix of $\Pi$ as follows. The simulator maintains a classical description of the parties' residual quantum states. Whenever an honest party sends a
message or queries the oracle during the prefix, the simulator produces a sample according to the distribution induced by the party's current
quantum state and updates its classical description of the corresponding
conditional residual state. Thus the simulation makes the same classical
oracle queries and produces exactly the same joint state of the public
transcript, the parties' query transcripts, and their residual registers as
the original prefix.

For $P\in\{\cA,\cB\}$, define
\[
  U_P := (R_P,M,T_P)
\]
as the classical view of party $P$ in this simulation, where $R_P$ records all
private random choices of the simulator, $M$ is the public prefix transcript,
and $T_P$ is the complete ordered oracle query-answer transcript of $P$.
This is equivalent to the notion of a view used by Barak and
Mahmoody~\cite[Section~3.1]{BarakMahmoodyFull}.  We note that the classical description of the residual states is not part of $U_P$.

We now explain what is needed to use this simulated classical view in combination with the techniques of \cite{BarakMahmoodyFull}. The eventual attack will retain the public prefix transcript $m$ and a
partial oracle $h$ learned during the prefix, and then prepare a fresh Alice
prefix view from its conditional marginal given $(m,h)$.  To justify replacing
Alice's real view by this fresh sample without significantly changing Bob's
subsequent behavior, we need the real joint state of the two prefix views,
conditioned on $(m,h)$, to be close to the product of its marginals.
\Cref{sec:bm-learner} obtains the corresponding approximate-product
guarantee for the simulated classical views $U_{\cA}$ and $U_{\cB}$ from
the Barak--Mahmoody learner.

That classical guarantee is not yet sufficient for our application.  The
final round of the key agreement protocol acts on the parties' residual quantum registers, rather than on the
simulator's private randomness, and its algorithms may query the same random
oracle $H$ in superposition.  We must therefore both transfer the
approximate-product guarantee from $U_{\cA},U_{\cB}$ to the parties' actual
query transcripts and residual quantum states, and preserve the correct
correlation of those states with $H$.  The local channels defined below
perform the first task.  The conditional-oracle identity in
\Cref{eq:prefix-conditional-oracle-main} performs the second.  In
\Cref{sec:conditioned-prefix-resampling}, we combine the two facts to replace
the real conditioned prefix state by independently sampled Alice and Bob
views together with a suitably sampled and reprogrammed oracle.

For every $u_P$ in the support of $U_P$, let $T_P(u_P)$ denote the query
transcript contained in $u_P$, and let $\sigma_P(u_P)$ be the residual density
operator tracked by the simulator after the prefix on that view.  Define the
local channel
\[
  \Phi_P(\proj{u_P})
  :=
  \proj{T_P(u_P)}\otimes\sigma_P(u_P).
\]
Thus $U_P$ is only an auxiliary classical view used to invoke the
Barak--Mahmoody result; applying $\Phi_P$ discards the simulator's private
randomness and keeps only the actual query transcript and residual quantum
state.

For a complete public prefix transcript $m$ and a partial oracle $h$, write
\[
  \mathsf F_{m,h}:=\{M=m,\ h\preceq H\}.
\]
Whenever $\Pr[\mathsf F_{m,h}]>0$, let $\nu^{AB}_{m,h}$ be the diagonal density
operator corresponding to the conditional distribution of
$(U_{\cA},U_{\cB})$ given $\mathsf F_{m,h}$, and let
$\nu^A_{m,h},\nu^B_{m,h}$ be its marginals.  Define
\[
  \omega^{AB}_{m,h}
  :=
  (\Phi_{\cA}\otimes\Phi_{\cB})(\nu^{AB}_{m,h}),
\]
and define $\omega^A_{m,h},\omega^B_{m,h}$ analogously.  Hence
$\omega^{AB}_{m,h}$ is a state on
$T_{\cA}S_{\cA}T_{\cB}S_{\cB}$ and contains no simulator randomness.
We note that, since we condition on $M=m$, the transcript register is fixed to
$\proj{m}_M$.  Let
$\Omega_{m,h}$ be the real
conditional state on
$\mathsf H T_{\cA}S_{\cA}T_{\cB}S_{\cB}$ given $\mathsf F_{m,h}$, where
$\mathsf H$ stores the random oracle.  Since the simulation is exact,
$\omega^{AB}_{m,h}$ is precisely the
$T_{\cA}S_{\cA}T_{\cB}S_{\cB}$-marginal of $\Omega_{m,h}$.

Since all oracle queries in the prefix are classical, once
$M=m$, $T_{\cA}=t_{\cA}$, and $T_{\cB}=t_{\cB}$ are fixed, the
residual state of Alice and Bob is the same for every oracle consistent
with $t_{\cA}$ and $t_{\cB}$.  Therefore, since $H$ is initially
uniform, whenever the conditioning event has positive probability,
\begin{equation}
  H\ \big|\
  \bigl(
    \mathsf F_{m,h},
    T_{\cA}=t_{\cA},
    T_{\cB}=t_{\cB}
  \bigr)
  \ \sim\
  \Unif\!\left(
    \cH(h,t_{\cA},t_{\cB})
  \right).
  \label{eq:prefix-conditional-oracle-main}
\end{equation}
Moreover, under this conditioning, the residual state is independent
of $H$.

\subsubsection{The Barak--Mahmoody Learner}
\label{sec:bm-learner}

In this section, we recall the relevant results from \cite{BarakMahmoodyFull}, lifted to our setting.  The first lemma below restates the
general result of Barak and Mahmoody for an arbitrary classical protocol.  In
particular, their Theorem~4.3 gives an expected query bound, guarantees
approximate independence, and bounds the conditional probability of each
unlearned input occurring in each party's query record after each round.
We retain only the
guarantees after the last protocol message, when \(M\) is the complete public
transcript and \(h\) is the learner's final partial oracle.  The second lemma
then applies this result to the exact classical simulation from
\Cref{sec:classical-prefix-views}.  It fixes the parameters in terms of
\(\eta\), truncates the learner to obtain a worst-case query bound, aborts
if the truncation threshold is reached or either the independence guarantee
or the bound on query probabilities fails, and transfers the
approximate-independence guarantee to the parties' residual quantum states.
It also characterizes the learner's output by \(h\preceq H\).  Thus the first
lemma isolates the result imported from Barak and Mahmoody, while the second
puts it in the form used for conditioned prefix resampling in
\Cref{sec:conditioned-prefix-resampling}.  The Barak--Mahmoody statement uses
Items~2, 4, and~5 of their Theorem~4.3 and their definition of
self-dependency~\cite[Definition~4.2 and Theorem~4.3,
Items~2, 4, and~5]{BarakMahmoodyFull}.

\begin{lemma}[Barak--Mahmoody]
\label{lem:bm-final-transcript}
Consider a classical two-party random-oracle protocol.  Let $M$ be its complete
public transcript and let
\[
  U_P=(R_P,M,T_P),\qquad P\in\{\cA,\cB\},
\]
be the parties' classical views, where $T_P$ is the oracle query-answer
transcript.  Suppose Alice and Bob make at most $n_A$ and $n_B$ oracle
queries, respectively, where $n_A,n_B\ge1$.

Let
\[
  0<\varepsilon_1\le\varepsilon_2<1/10,
  \qquad
  0<\alpha,\beta<1,
  \qquad
  \alpha\beta\ge\varepsilon_2.
\]
There is a deterministic classical-query algorithm
\[
  \mathsf L^H_{\varepsilon_1,\varepsilon_2}(m)
\]
that takes a complete public transcript $m$ as input and outputs a partial
oracle $h$.  For every realized pair $(m,h)$, let $\mu^{AB}_{m,h}$ be the
conditional distribution of $(U_{\cA},U_{\cB})$ given
\[
  M=m,
  \qquad
  \mathsf L^H_{\varepsilon_1,\varepsilon_2}(m)=h,
\]
and let $\mu^A_{m,h},\mu^B_{m,h}$ be its marginals.  Then:
\begin{enumerate}[leftmargin=2em,label=\textup{(\roman*)}]
  \item The expected number of oracle queries made by
  $\mathsf L_{\varepsilon_1,\varepsilon_2}$ is at most
  \[
    \frac{4n_A n_B}{\varepsilon_1}.
  \]

  \item Except with probability at most $9\alpha$ over the realized pair
  $(M,h)$,
  \[
    \Delta_{\mathrm{TV}}\!\left(
      \mu^{AB}_{M,h},
      \mu^A_{M,h}\times\mu^B_{M,h}
    \right)
    \le 9\beta.
  \]

  \item Except with probability at most $9\alpha$ over the realized pair
  $(M,h)$, every $x\notin\Dom(h)$ satisfies
  \[
    \Pr[x\in\Dom(T_{\cA})\mid M,h]
      < \frac{\varepsilon_2}{n_B}+\beta,
    \qquad
    \Pr[x\in\Dom(T_{\cB})\mid M,h]
      < \frac{\varepsilon_2}{n_A}+\beta.
  \]
\end{enumerate}
\end{lemma}

\begin{proof}
We use the deterministic learner from Barak and Mahmoody's
Construction~3.3 with parameter $\varepsilon_1$, together with the stopping
modification from their Theorem~4.3.  Barak and Mahmoody describe this learner
as running alongside the protocol.  Its oracle queries do not affect the
honest execution, and Construction~3.3 uses deterministic tie-breaking.
Consequently, given the complete transcript $m$, we may replay the learner on
the successive prefixes of $m$ and make the same oracle queries.  This defines
the deterministic algorithm $\mathsf L^H_{\varepsilon_1,\varepsilon_2}(m)$
with the same final partial oracle and query complexity.

With the stopping modification, the learner is an
$(\varepsilon_1,\varepsilon_2)$-attacker in the sense of Definition~4.1 of
Barak and Mahmoody.  Item~2 of their Theorem~4.3 gives the expected query
bound in Item~\textup{(i)}.  Their Definition~4.2 identifies self-dependency
with statistical distance from the product of the marginals.  Evaluating
Item~4 of Theorem~4.3 after the final protocol message therefore gives
Item~\textup{(ii)}.  Evaluating Item~5 at the same point gives the two
bounds on query probabilities in Item~\textup{(iii)}
\cite[Construction~3.3, Definitions~4.1--4.2, and Theorem~4.3,
Items~2, 4, and~5]{BarakMahmoodyFull}.
\end{proof}

\begin{lemma}[Barak--Mahmoody prefix attacker]
\label{lem:multi-independence-attacker}
Let $0<\eta<1/100$, and set
\[
  \bar q_{\cA}:=\max\{1,q_{\cA,\mathsf{pre}}\},
  \qquad
  \bar q_{\cB}:=\max\{1,q_{\cB,\mathsf{pre}}\},
\]
\[
  Q_{\mathsf{learn}}
  :=
  \left\lceil
    \frac{4\cdot 10^5\,\bar q_{\cA}\bar q_{\cB}}{\eta^3}
  \right\rceil.
\]
There is a deterministic algorithm
$\Learn^{\mathsf{BM},H}_\eta$ with classical oracle access that, on input a
complete public prefix transcript $m$, outputs a partial oracle $h$ or aborts.
A pair $(m,h)$ is \emph{reachable} if
\[
  \Pr\!\left[
    M=m,\
    \Learn^{\mathsf{BM},H}_\eta(1^\lambda,m)=h
  \right]>0.
\]
It satisfies the following properties.
\begin{enumerate}[leftmargin=2em,label=\textup{(\roman*)}]
  \item It makes at most $Q_{\mathsf{learn}}$ oracle queries and aborts with
  probability at most $\eta$.

  \item For every reachable pair $(m,h)$,
  \[
    \frac12
    \left\|
      \omega^{AB}_{m,h}
      -
      \omega^A_{m,h}\otimes\omega^B_{m,h}
    \right\|_1
    \le\eta.
  \]

  \item For every reachable pair $(m,h)$ and every $x\notin\Dom(h)$,
  \[
    \Pr[x\in\Dom(T_{\cA})\mid\mathsf F_{m,h}]
    \le\eta.
  \]
\end{enumerate}
Moreover, for every reachable pair $(m,h)$ and every oracle $H$,
\[
  \Learn^{\mathsf{BM},H}_\eta(1^\lambda,m)=h
  \quad\Longleftrightarrow\quad
  h\preceq H.
\]
\end{lemma}

\begin{proof}
We apply \Cref{lem:bm-final-transcript} to the exact classical simulation of
the prefix with parameters
\[
  n_A=\bar q_{\cA},
  \qquad
  n_B=\bar q_{\cB},
  \qquad
  \varepsilon_1=\varepsilon_2=\frac{\eta^2}{10^4},
  \qquad
  \alpha_0=\beta_0=\frac{\eta}{100}.
\]
These parameters satisfy $\alpha_0\beta_0=\varepsilon_2$.  The resulting
learner makes at most
\[
  \frac{4\bar q_{\cA}\bar q_{\cB}}{\varepsilon_1}
  =
  \frac{4\cdot10^4\bar q_{\cA}\bar q_{\cB}}{\eta^2}
\]
queries in expectation.  We truncate this learner after
$Q_{\mathsf{learn}}$ queries and abort if the truncation threshold is
reached.  By Markov's inequality, this occurs with probability at most
$\eta/10$.

The exceptional events in Items~\textup{(ii)} and~\textup{(iii)} of
\Cref{lem:bm-final-transcript} have total probability at most
$18\alpha_0$.  Since computation is unbounded in our query model, once the
learner terminates we may test whether both conditional guarantees hold and
abort otherwise, without making any further oracle queries.  By a union
bound, the total abort probability is at most
\[
  18\alpha_0+\frac{\eta}{10}
  =\frac{28\eta}{100}<\eta.
\]
This proves Item~\textup{(i)}.

We next establish the claimed characterization of the learner's output.  Fix
a reachable nonabort pair $(m,h)$.  For fixed $m$, every query and stopping
decision of the learner is a deterministic function of the partial oracle
learned so far.  Since some execution produces $h$, every oracle extending
$h$ gives the same answers along that execution and therefore produces the
same output.  Conversely, if the learner outputs $h$, then every pair in $h$
was obtained from an oracle query, and hence $h\preceq H$.  It follows that,
for every oracle $H$,
\[
  \Learn^{\mathsf{BM},H}_\eta(1^\lambda,m)=h
  \quad\Longleftrightarrow\quad
  h\preceq H.
\]
Thus, for every reachable nonabort pair $(m,h)$, conditioning on the learner's
output is equivalent to conditioning on $\mathsf F_{m,h}$.

It remains to transfer the guarantees for the simulated classical views to
the actual prefix views.  By Item~\textup{(ii)} of
\Cref{lem:bm-final-transcript},
\[
  \Delta_{\mathrm{TV}}\!\left(
    \nu^{AB}_{m,h},
    \nu^A_{m,h}\times\nu^B_{m,h}
  \right)
  \le 9\beta_0<\eta.
\]
Applying the local channel $\Phi_{\cA}\otimes\Phi_{\cB}$ and using
trace-distance data processing~\cite[Chapter~3]{Watrous18} proves
Item~\textup{(ii)}.  Finally, Item~\textup{(iii)} of
\Cref{lem:bm-final-transcript} gives, for every $x\notin\Dom(h)$,
\[
  \Pr[x\in\Dom(T_{\cA})\mid\mathsf F_{m,h}]
  <
  \frac{\varepsilon_2}{\bar q_{\cB}}+\beta_0
  \le
  \frac{\eta^2}{10^4}+\frac{\eta}{100}
  <\eta,
\]
which proves Item~\textup{(iii)}.
\end{proof}

\subsubsection{Conditioned Prefix Resampling}
\label{sec:conditioned-prefix-resampling}

The preceding subsection shows that, conditioned on a reachable pair
\((m,h)\), the parties' actual prefix views are close to the product state
\(\omega^A_{m,h}\otimes\omega^B_{m,h}\).  This does not by itself justify
sampling the two views independently in the attack, because that statement
does not include the random oracle.  Bob's final algorithm and Eve's
simulation of Alice both query the same oracle in superposition, so we must
also specify how to sample that oracle together with the independent views.

The conditional-oracle identity in
\Cref{eq:prefix-conditional-oracle-main} tells us how to do so.  After
sampling the two marginal views independently, we first sample a base oracle
\(H_0\) consistent with \(h\) and Bob's transcript \(T'_{\cB}\), and then
reprogram it on Alice's transcript \(T'_{\cA}\).  The order is deliberate:
the base oracle is compatible with Bob's sampled view, while the final
oracle \(\widetilde H=H_0[T'_{\cA}]\) contains the answers recorded in
Alice's sampled view. 

To bound the trace distance between the real and resampled states, we apply
the same channel to two input states.  On the real joint view
\(\omega^{AB}_{m,h}\), the channel reconstructs the real conditioned state
\(\Omega_{m,h}\) exactly.  On the product of the marginals, it produces the
resampled state \(\widetilde\Omega_{m,h}\) defined below.  The trace-distance
bound then follows directly from the
approximate-independence guarantee of the preceding subsection.  In
particular, any later joint quantum computation using the oracle and the
prefix views changes the success probability by at most \(\eta\) when the
real conditioned prefix is replaced by this resampled one.

\begin{lemma}[Conditioned prefix resampling]
\label{lem:prefix-resampling}
Fix a reachable pair $(m,h)$, and consider the following experiment:
\begin{enumerate}[leftmargin=2em]
  \item independently prepare
  \[
    (T'_{\cA},S'_{\cA})\leftarrow\omega^A_{m,h},
    \qquad
    (T'_{\cB},S'_{\cB})\leftarrow\omega^B_{m,h};
  \]
  \item sample $H_0\leftarrow\Unif(\cH(h,T'_{\cB}))$, and define
  \[
    \widetilde H=H_0[T'_{\cA}].
  \]
\end{enumerate}
Let $\widetilde\Omega_{m,h}$ be the resulting state on
$\mathsf H T'_{\cA}S'_{\cA}T'_{\cB}S'_{\cB}$, with $\mathsf H$ holding the
sampled function $\widetilde H$.  Then
\[
  \frac12
  \left\|
    \Omega_{m,h}
    -
    \widetilde\Omega_{m,h}
  \right\|_1
  \le \eta.
\]
\end{lemma}

\begin{proof}
To bound the trace distance between the two states, we apply the same
channel to the real joint prefix view and to the product of its marginals.
Define a channel
$\mathcal R_{m,h}$ on $T_{\cA}S_{\cA}T_{\cB}S_{\cB}$ as follows.  On
classical transcripts $t_{\cA},t_{\cB}$, it samples
\[
  H_0\leftarrow\Unif(\cH(h,t_{\cB}))
\]
and attaches the oracle $H_0[t_{\cA}]$, while leaving the four input
registers unchanged.  Every $t_{\cB}$ in the support of
$\omega^B_{m,h}$ is consistent with $h$, so this channel is well defined on
both input states considered below.

We first apply the channel to the real joint prefix view.  In this case,
$t_{\cA}$ and $t_{\cB}$ are mutually consistent, and
$H_0[t_{\cA}]$ is uniform over $\cH(h,t_{\cA},t_{\cB})$.  By
\eqref{eq:prefix-conditional-oracle-main}, this is precisely the conditional
distribution of the real oracle given the two query transcripts.  Moreover,
the residual state is independent of the unobserved oracle values under this
conditioning.  Therefore,
\[
  \mathcal R_{m,h}(\omega^{AB}_{m,h})=\Omega_{m,h}.
\]
Applying the same channel to the product of the two marginal views gives,
by the definition of the resampling experiment,
\[
  \mathcal R_{m,h}
  (\omega^A_{m,h}\otimes\omega^B_{m,h})
  =\widetilde\Omega_{m,h}.
\]
Combining these two identities with trace-distance data processing and
\Cref{lem:multi-independence-attacker}, we obtain
\[
  \frac12\left\|
    \Omega_{m,h}-\widetilde\Omega_{m,h}
  \right\|_1
  \le
  \frac12\left\|
    \omega^{AB}_{m,h}
    -\omega^A_{m,h}\otimes\omega^B_{m,h}
  \right\|_1
  \le\eta.
\]
\end{proof}

\subsubsection{The Multi-Round Attacker}

In this section we conclude the proof by defining the attacker.

\begin{definition}[Multi-round reprogramming attack]
\label{def:multi-eve}
Fix \(0<\eta<1/100\). During the classical prefix, Eve records
its messages, and immediately after its last message she runs
\(\Learn^{\mathsf{BM},H}_\eta(1^\lambda,m)\) from
\Cref{lem:multi-independence-attacker} on the complete prefix transcript
\(m\), before Bob's final stage begins. The online interception
experiment in \Cref{def:multi-intkr} allows Eve to query after each
prefix message, but this attacker only records the earlier messages
and makes no oracle queries until the prefix is complete. The learner
therefore receives the complete transcript, as required by
\Cref{lem:multi-independence-attacker}.

If the learner aborts, Eve records an abort flag. Otherwise, it returns
a partial oracle \(h\), and she prepares an independent Alice prefix view
\[
  V'_{\cA}=(W'_{\cA},S'_{\cA})
  \leftarrow\omega^A_{m,h},
  \qquad
  T'_{\cA}:=T_{\cA}(W'_{\cA}).
\]
After intercepting Bob's final message register \(\Psi\), Eve outputs
\(\bot\) if the attacker aborted. Otherwise she runs
\[
  K^*
  \leftarrow
  \cA_{\mathsf{fin}}^{H[T'_{\cA}]}
  (1^\lambda,m,S'_{\cA},\Psi)
\]
and outputs \(K^*\).
\end{definition}

\begin{theorem}[Multi-round quantitative attack]
\label{thm:multi-main}
Let \(\Pi^{\mathsf{mr}}\) satisfy
\Cref{def:multi-protocol}, and let
\(\alpha=\alpha_{\Pi^{\mathsf{mr}}}(\lambda)\).  For every
\(0<\eta<1/100\), the attacker of \Cref{def:multi-eve} satisfies
\[
  \Adv^{\IntKR,\mathsf{mr}}_{
    \Pi^{\mathsf{mr}},\cE_\eta}(\lambda)
  \ge
  \alpha
  -3\eta
  -2(2q_{\cB,\mathsf q}+q_{\cA,\mathsf{fin}})\sqrt\eta.
\]
It makes at most
\[
  \left\lceil
    \frac{4\cdot 10^5\,
      \max\{1,q_{\cA,\mathsf{pre}}\}
      \max\{1,q_{\cB,\mathsf{pre}}\}}
      {\eta^3}
  \right\rceil
  +q_{\cA,\mathsf{fin}}
\]
random-oracle queries.  The first term consists only of classical queries.
There is no explicit dependence on the number of prefix rounds beyond the
two parties' total prefix-query bounds.
\end{theorem}

\begin{proof}
We proceed via a hybrid argument. Hybrid \(J_0\) is the honest
execution, while in \(J_1,\ldots,J_5\) the learner runs on the
complete classical prefix transcript before Bob's final stage, and a
learner abort counts as failure. Unless stated otherwise, these hybrids induce the same distribution
on the public transcript \(M\) and the attacker's output \(h\). We say that
\(J_i\) succeeds when its two recorded keys agree and are not \(\bot\).
Throughout, \(V_P=(W_P,S_P)\) denotes the extended classical-quantum
prefix view defined above. Changes from one hybrid to the next are shown in
\mrchanged{color}.

\begin{hybrid}{Hybrid \(J_0\) (real execution)}
  \item Sample \(H_0\leftarrow\cH\) and run the honest classical prefix,
  obtaining
  \((M^{(0)},V_{\cA,0},V_{\cB,0})\).
  \item Run
  \((\Psi_0,K_{\cB,0})\leftarrow
    \cB_{\mathsf q}^{H_0}(1^\lambda,M^{(0)},S_{\cB,0})\).
  \item Run
  \(K_{\cA,0}\leftarrow
    \cA_{\mathsf{fin}}^{H_0}
    (1^\lambda,M^{(0)},S_{\cA,0},\Psi_0)\).
\end{hybrid}
Thus
\nopagebreak[4]
\[
  \Pr[J_0\text{ succeeds}]=\alpha.
\]

\begin{hybrid}{Hybrid \(J_1\)}
  \item \mrchanged{Sample \(H_1\leftarrow\cH\), and run the honest
  classical prefix using \(H_1\), then run
  \(\Learn^{\mathsf{BM},H_1}_\eta(1^\lambda,M^{(1)})\) on its
  complete transcript. Obtain \((M^{(1)},V_{\cA,1},V_{\cB,1},h_1)\),
  and abort if the learner aborts.}
  \item Run
  \((\Psi_1,K_{\cB,1})\leftarrow
    \cB_{\mathsf q}^{H_1}(1^\lambda,M^{(1)},S_{\cB,1})\), followed by
  \(K_{\cA,1}\leftarrow
    \cA_{\mathsf{fin}}^{H_1}
    (1^\lambda,M^{(1)},S_{\cA,1},\Psi_1)\).
\end{hybrid}
The learner only reads the fixed oracle and does not act on the honest
registers, so inserting it after the prefix leaves the honest execution
unchanged. The only difference is that \(J_1\) counts a learner abort
as failure. Item~(i) of
\Cref{lem:multi-independence-attacker} gives
\begin{equation}
\label{eq:multi-j0-j1}
  \left|
    \Pr[J_0\text{ succeeds}]-\Pr[J_1\text{ succeeds}]
  \right|
  \le\eta.
\end{equation}

\begin{hybrid}{Hybrid \(J_2\)}
  \item Run the real prefix and attacker as in \(J_1\), retaining
  \((M^{(2)},h_2)\) and discarding the original oracle and prefix views.
  \item \mrchanged{Independently prepare
  \[
    V'_{\cA,2}=(W'_{\cA,2},S'_{\cA,2})
      \leftarrow\omega^A_{M^{(2)},h_2},
  \]
  \[
    V'_{\cB,2}=(W'_{\cB,2},S'_{\cB,2})
      \leftarrow\omega^B_{M^{(2)},h_2}.
  \]
  Set \(T'_{\cA,2}=T_{\cA}(W'_{\cA,2})\) and
  \(T'_{\cB,2}=T_{\cB}(W'_{\cB,2})\).}
  \item \mrchanged{Sample
  \(\widehat H_2\leftarrow\Unif(\cH(h_2,T'_{\cB,2}))\) and set
  \(\widetilde H_2=\widehat H_2[T'_{\cA,2}]\).}
  \item Run
  \((\Psi_2,K_{\cB,2})\leftarrow
    \cB_{\mathsf q}^{\widetilde H_2}
    (1^\lambda,M^{(2)},S'_{\cB,2})\).
  \item Run
  \(K^*_2\leftarrow
    \cA_{\mathsf{fin}}^{\widetilde H_2}
    (1^\lambda,M^{(2)},S'_{\cA,2},\Psi_2)\).
\end{hybrid}

Condition on a reachable pair \((m,h)\). By the output characterization in
\Cref{lem:multi-independence-attacker}, the oracle and prefix views in \(J_1\)
have joint state \(\Omega_{m,h}\), while the resampled objects in \(J_2\)
have state \(\widetilde\Omega_{m,h}\). Applying
\Cref{lem:prefix-resampling} and trace-distance data processing through the
final algorithms~\cite[Chapter~3]{Watrous18} shows that, for every reachable
nonabort pair,
\[
  \left|
    \Pr[J_1\text{ succeeds}\mid M=m,\ h_1=h]
    -
    \Pr[J_2\text{ succeeds}\mid M=m,\ h_2=h]
  \right|
  \le\eta.
\]  

Let \(\pi(m,h)\) be the common subprobability mass of nonabort prefix
executions in the two hybrids that output \((M,h)=(m,h)\). Since success is
defined to be false on abort, only these executions contribute. Define
\[
  s_i(m,h)
  :=
  \Pr[J_i\text{ succeeds}\mid M^{(i)}=m,\ h_i=h]
  \qquad (i\in\{1,2\}).
\]
The triangle inequality gives
\[
  \begin{aligned}
  &\left|
    \Pr[J_1\text{ succeeds}]-\Pr[J_2\text{ succeeds}]
  \right|\\
  &\qquad\le
  \sum_{(m,h)\in\Supp(\pi)}\pi(m,h)
  \left|s_1(m,h)-s_2(m,h)\right|
  \le\eta.
  \end{aligned}
\]
In particular,
\begin{equation}
\label{eq:multi-j1-j2}
  \left|
    \Pr[J_1\text{ succeeds}]-\Pr[J_2\text{ succeeds}]
  \right|
  \le\eta.
\end{equation}

\begin{hybrid}{Hybrid \(J_3\)}
  \item Run the real prefix and attacker as in \(J_1\), retain
  \((M^{(3)},h_3)\), discard the original oracle and prefix views, and
  independently prepare
  \[
    V'_{\cA,3}=(W'_{\cA,3},S'_{\cA,3})
      \leftarrow\omega^A_{M^{(3)},h_3},
    \qquad
    V'_{\cB,3}=(W'_{\cB,3},S'_{\cB,3})
      \leftarrow\omega^B_{M^{(3)},h_3},
  \]
  set \(T'_{\cA,3}=T_{\cA}(W'_{\cA,3})\) and
  \(T'_{\cB,3}=T_{\cB}(W'_{\cB,3})\), and sample
  \(\widehat H_3\leftarrow\Unif(\cH(h_3,T'_{\cB,3}))\).
  \item Set \(\widetilde H_3=\widehat H_3[T'_{\cA,3}]\).
  \item \mrchanged{Run
  \((\Psi_3,K_{\cB,3})\leftarrow
    \cB_{\mathsf q}^{\widehat H_3}
    (1^\lambda,M^{(3)},S'_{\cB,3})\).}
  \item Run
  \(K^*_3\leftarrow
    \cA_{\mathsf{fin}}^{\widetilde H_3}
    (1^\lambda,M^{(3)},S'_{\cA,3},\Psi_3)\).
\end{hybrid}

\begin{lemma}[\(J_2\) to \(J_3\)]
\label{lem:multi-j2-j3}
\[
  \left|
    \Pr[J_2\text{ succeeds}]-\Pr[J_3\text{ succeeds}]
  \right|
  \le2q_{\cB,\mathsf q}\sqrt\eta.
\]
\end{lemma}

\begin{proof}
Condition on a reachable pair \((m,h)\) and on a value
\(W'_{\cB}=w'_{\cB}\) in the support of the classical-history marginal of
\(\omega^B_{m,h}\). Prepare \(S'_{\cB}\) in the corresponding conditional
state and set \(T'_{\cB}=T_{\cB}(w'_{\cB})\).
For the random reprogramming game of Katz and Sela
\cite[Lemma~3]{KatzSela24}, sample
\(\widehat H\leftarrow\Unif(\cH(h,T'_{\cB}))\) and set
\(F_0=\widehat H\). The randomized classical procedure \(C_{m,h}\) samples
the history \(W'_{\cA}\) from the classical marginal of
\(\omega^A_{m,h}\), sets \(T'_{\cA}=T_{\cA}(W'_{\cA})\), and returns
\[
  R
  =
  T'_{\cA}
  \big|_{\Dom(T'_{\cA})\setminus\Dom(h)}.
\]
Given challenge access to \(F_b\), run
\[
  (\Psi,K_{\cB})
  \leftarrow
  \cB_{\mathsf q}^{F_b}(1^\lambda,m,S'_{\cB}).
\]
After challenge access is removed, the coins of \(C_{m,h}\) reveal
\(W'_{\cA}\). Prepare \(S'_{\cA}\) in the corresponding conditional
state specified by \(\omega^A_{m,h}\), and run
\[
  K^*
  \leftarrow
  \cA_{\mathsf{fin}}^{\widehat H[T'_{\cA}]}
  (1^\lambda,m,S'_{\cA},\Psi).
\]
The distinguisher outputs \(1\) iff \(K^*=K_{\cB}\neq\bot\).

Every sampled Alice transcript is consistent with \(h\), so
\(F_1=\widehat H[R]=\widehat H[T'_{\cA}]\). Thus \(b=1\) reproduces \(J_2\)
and \(b=0\) reproduces \(J_3\), conditioned on
\((M,h,W'_{\cB})=(m,h,w'_{\cB})\). Only Bob queries the challenge oracle.
Moreover, by
\Cref{lem:multi-independence-attacker}, for every \(x\),
\[
  \Pr[x\in\Dom(R)]
  \le
  \begin{cases}
    0,&x\in\Dom(h),\\
    \Pr[x\in\Dom(T'_{\cA})]\le\eta,&x\notin\Dom(h).
  \end{cases}
\]
Thus \Cref{lem:reprogramming} bounds the absolute difference between the two
conditional success probabilities by \(2q_{\cB,\mathsf q}\sqrt\eta\) for
every \((m,h,w'_{\cB})\) in the support of the shared conditioning data. If
\(\pi(m,h,w'_{\cB})\) denotes the common subprobability mass of that data on
nonabort executions of \(J_2\) and \(J_3\), define
\[
  s_i(m,h,w'_{\cB})
  :=
  \Pr[J_i\text{ succeeds}\mid
    M^{(i)}=m,\ h_i=h,\ W'_{\cB,i}=w'_{\cB}]
  \qquad (i\in\{2,3\}).
\]
Because success is defined to be false on abort,
\[
  \Pr[J_i\text{ succeeds}]
  =
  \sum_{(m,h,w'_{\cB})\in\Supp(\pi)}
    \pi(m,h,w'_{\cB})s_i(m,h,w'_{\cB})
  \qquad (i\in\{2,3\}).
\]
The triangle inequality gives
\[
  \begin{aligned}
  \left|\Pr[J_2\text{ succeeds}]-\Pr[J_3\text{ succeeds}]\right|
  &\le
  \sum_{(m,h,w'_{\cB})\in\Supp(\pi)}\pi(m,h,w'_{\cB})
  \left|s_2(m,h,w'_{\cB})-s_3(m,h,w'_{\cB})\right|\\
  &\le 2q_{\cB,\mathsf q}\sqrt\eta.
  \end{aligned}
\]
\end{proof}

\begin{hybrid}{Hybrid \(J_4\)}
  \item Run the real prefix and attacker as in \(J_1\), retain
  \((M^{(4)},h_4)\), discard the original oracle and prefix views, and
  independently prepare
  \[
    V^{\mathsf E}_{\cA,4}
      =(W^{\mathsf E}_{\cA,4},S^{\mathsf E}_{\cA,4})
      \leftarrow\omega^A_{M^{(4)},h_4},
    \qquad
    V'_{\cB,4}=(W'_{\cB,4},S'_{\cB,4})
      \leftarrow\omega^B_{M^{(4)},h_4},
  \]
  set \(T^{\mathsf E}_{\cA,4}=T_{\cA}(W^{\mathsf E}_{\cA,4})\) and
  \(T'_{\cB,4}=T_{\cB}(W'_{\cB,4})\), and sample
  \(\widehat H_4\leftarrow\Unif(\cH(h_4,T'_{\cB,4}))\).
  \item \mrchanged{Independently prepare the Alice view that determines
  Bob's oracle,
  \[
    V^{\mathsf B}_{\cA,4}
    =
    (W^{\mathsf B}_{\cA,4},S^{\mathsf B}_{\cA,4})
    \leftarrow\omega^A_{M^{(4)},h_4},
  \]
  and set
  \(T^{\mathsf B}_{\cA,4}=T_{\cA}(W^{\mathsf B}_{\cA,4})\).
  Set
  \[
    H^{\mathsf B}_4=\widehat H_4[T^{\mathsf B}_{\cA,4}],
    \qquad
    H^{\mathsf E}_4=H^{\mathsf B}_4[T^{\mathsf E}_{\cA,4}].
  \]}
  \item Run
  \((\Psi_4,K_{\cB,4})\leftarrow
    \cB_{\mathsf q}^{H^{\mathsf B}_4}
    (1^\lambda,M^{(4)},S'_{\cB,4})\).
  \item Run
  \(K^*_4\leftarrow
    \cA_{\mathsf{fin}}^{H^{\mathsf E}_4}
    (1^\lambda,M^{(4)},S^{\mathsf E}_{\cA,4},\Psi_4)\).
\end{hybrid}

\begin{lemma}[\(J_3\) to \(J_4\)]
\label{lem:multi-j3-j4}
\[
  \left|
    \Pr[J_3\text{ succeeds}]-\Pr[J_4\text{ succeeds}]
  \right|
  \le
  2(q_{\cB,\mathsf q}+q_{\cA,\mathsf{fin}})\sqrt\eta.
\]
\end{lemma}

\begin{proof}
Condition on a reachable pair \((m,h)\) and on a value
\(W'_{\cB}=w'_{\cB}\) in the support of the classical-history marginal of
\(\omega^B_{m,h}\). Prepare \(S'_{\cB}\) in the corresponding conditional
state, and write \(T'_{\cB}=T_{\cB}(w'_{\cB})\). Independently prepare Eve's Alice view
\(V^{\mathsf E}_{\cA}=(W^{\mathsf E}_{\cA},S^{\mathsf E}_{\cA})
\leftarrow\omega^A_{m,h}\), and set
\(T^{\mathsf E}_{\cA}=T_{\cA}(W^{\mathsf E}_{\cA})\).  Sample
\(\widehat H\leftarrow\Unif(\cH(h,T'_{\cB}))\), output
\(F_0=\widehat H\), and let the randomized classical procedure
\(C_{m,h}\) independently sample \(W^{\mathsf B}_{\cA}\) from the classical
marginal of \(\omega^A_{m,h}\), set
\(T^{\mathsf B}_{\cA}=T_{\cA}(W^{\mathsf B}_{\cA})\), and return
\[
  R
  =
  T^{\mathsf B}_{\cA}
  \big|_{\Dom(T^{\mathsf B}_{\cA})\setminus\Dom(h)}.
\]
Given challenge access to \(F_b\), run Bob's final algorithm with \(F_b\)
and Alice's final algorithm with \(F_b[T^{\mathsf E}_{\cA}]\).  A query to
\(F_b[T^{\mathsf E}_{\cA}]\) is implemented with one query to \(F_b\) by
reversibly testing membership in the known transcript and supplying the
recorded answer on member inputs. The total number of challenge queries is
therefore at most
\(q_{\cB,\mathsf q}+q_{\cA,\mathsf{fin}}\).

After challenge access is removed, the revealed coins of \(C_{m,h}\)
determine \(W^{\mathsf B}_{\cA}\). The reduction can then prepare the
corresponding register \(S^{\mathsf B}_{\cA}\). This register is unused in
the \(b=0\) case and may be discarded; in the \(b=1\) case it completes the
auxiliary classical-quantum view sampled in \(J_4\). Every sampled Alice
transcript is consistent with \(h\). Therefore, for \(b=0\) the experiment
is \(J_3\) after discarding that unused view, while for \(b=1\),
\(F_1=\widehat H[T^{\mathsf B}_{\cA}]\) and the experiment is \(J_4\).
As in the preceding proof, each point belongs to the reprogramming set with
probability at most \(\eta\): points in \(\Dom(h)\) are excluded, and
Item~(iii) of \Cref{lem:multi-independence-attacker} applies everywhere else.
The random reprogramming lemma of Katz and Sela
\cite[Lemma~3]{KatzSela24} therefore bounds the absolute difference
between the conditional success probabilities of \(J_3\) and \(J_4\), given
\((M,h,W'_{\cB})=(m,h,w'_{\cB})\), by
\(2(q_{\cB,\mathsf q}+q_{\cA,\mathsf{fin}})\sqrt\eta\). Let
\(\pi(m,h,w'_{\cB})\) denote the common subprobability of these conditioning
data in the two hybrids, and define
\[
  s_i(m,h,w'_{\cB})
  :=
  \Pr[J_i\text{ succeeds}\mid
    M^{(i)}=m,\ h_i=h,\ W'_{\cB,i}=w'_{\cB}]
  \qquad (i\in\{3,4\}).
\]
The law of total probability and the triangle inequality yield
\[
  \begin{aligned}
  \left|\Pr[J_3\text{ succeeds}]-\Pr[J_4\text{ succeeds}]\right|
  &\le
  \sum_{m,h,w'_{\cB}}\pi(m,h,w'_{\cB})
  \left|s_3(m,h,w'_{\cB})-s_4(m,h,w'_{\cB})\right|\\
  &\le
  2(q_{\cB,\mathsf q}+q_{\cA,\mathsf{fin}})\sqrt\eta.
  \end{aligned}
\]
Again, the omitted abort branch has success probability zero in both
hybrids.
\end{proof}

\begin{hybrid}{Hybrid \(J_5\) (the real interception attack)}
  \item \mrchanged{Sample \(H_5\leftarrow\cH\), run the honest classical
  prefix, and then run \(\Learn^{\mathsf{BM},H_5}_\eta\) on its
  complete transcript. Obtain
  \[
    (M^{(5)},h_5,V_{\cA,5},V_{\cB,5}),
  \]
  and abort if the attacker aborts.}
  \item Independently prepare Eve's Alice prefix view
  \[
    V'_{\cA,5}
    =(W'_{\cA,5},S'_{\cA,5})
    \leftarrow\omega^A_{M^{(5)},h_5},
  \]
  and set \(T'_{\cA,5}=T_{\cA}(W'_{\cA,5})\).
  \item \mrchanged{Run
  \((\Psi_5,K_{\cB,5})\leftarrow
    \cB_{\mathsf q}^{H_5}
    (1^\lambda,M^{(5)},S_{\cB,5})\),
  and transfer \(\Psi_5\) to Eve.}
  \item Run
  \[
    K^*_5
    \leftarrow
    \cA_{\mathsf{fin}}^{H_5[T'_{\cA,5}]}
    (1^\lambda,M^{(5)},S'_{\cA,5},\Psi_5).
  \]
\end{hybrid}

Fix a reachable pair \((m,h)\). The tuple
\[
  \bigl(
    H^{\mathsf B}_4,
    V^{\mathsf B}_{\cA,4},
    V'_{\cB,4}
  \bigr)
\]
in \(J_4\) has exactly the distribution of the conditioned prefix-resampling
experiment from \Cref{lem:prefix-resampling}. Eve's second Alice view
\(V^{\mathsf E}_{\cA,4}\) is independent in both hybrids. Hence, by
\Cref{lem:prefix-resampling} and trace-distance data processing
\cite[Chapter~3]{Watrous18}, for every reachable nonabort pair \((m,h)\),
\[
  \left|
    \Pr[J_4\text{ succeeds}\mid M^{(4)}=m,\ h_4=h]
    -
    \Pr[J_5\text{ succeeds}\mid M^{(5)}=m,\ h_5=h]
  \right|
  \le\eta.
\] 

Let \(\pi(m,h)\) be the common probability that the prefix and attacker in
the two hybrids output \((m,h)\) without aborting, and define
\[
  s_i(m,h)
  :=
  \Pr[J_i\text{ succeeds}\mid M^{(i)}=m,\ h_i=h]
  \qquad (i\in\{4,5\}).
\]
The abort branch contributes zero, so the law of total probability and the
triangle inequality give
\[
  \left|
    \Pr[J_4\text{ succeeds}]-\Pr[J_5\text{ succeeds}]
  \right|
  \le
  \sum_{m,h}\pi(m,h)\left|s_4(m,h)-s_5(m,h)\right|
  \le\eta.
\]
Thus
\begin{equation}
\label{eq:multi-j4-j5}
  \left|
    \Pr[J_4\text{ succeeds}]-\Pr[J_5\text{ succeeds}]
  \right|
  \le\eta.
\end{equation}

Hybrid \(J_5\) is precisely
\(\IntKR^{\mathsf{mr}}_{\Pi^{\mathsf{mr}},\cE_\eta}\). Combining
\eqref{eq:multi-j0-j1}, \eqref{eq:multi-j1-j2},
\Cref{lem:multi-j2-j3,lem:multi-j3-j4}, and
\eqref{eq:multi-j4-j5} gives
\begin{align*}
  \Adv^{\IntKR,\mathsf{mr}}_{
    \Pi^{\mathsf{mr}},\cE_\eta}(\lambda)
  &\ge
  \alpha
  -\eta-\eta
  -2q_{\cB,\mathsf q}\sqrt\eta
  \\
  &\quad
  -2(q_{\cB,\mathsf q}+q_{\cA,\mathsf{fin}})\sqrt\eta
  -\eta
  \\
  &=
  \alpha
  -3\eta
  -2(2q_{\cB,\mathsf q}+q_{\cA,\mathsf{fin}})\sqrt\eta.
\end{align*}
The query bound follows directly from
\Cref{lem:multi-independence-attacker,def:multi-eve}.
\end{proof}

\begin{corollary}[Polynomial-query multi-round impossibility]
\label{cor:multi-impossibility}
Let \(\Pi^{\mathsf{mr}}\) satisfy \Cref{def:multi-protocol}, let
\(\alpha=\alpha_{\Pi^{\mathsf{mr}}}(\lambda)>0\), and define
\[
  Q_{\mathsf{fin}}
  :=
  \max\{1,2q_{\cB,\mathsf q}+q_{\cA,\mathsf{fin}}\},
  \qquad
  \eta
  :=
  \left(\frac{\alpha}{16Q_{\mathsf{fin}}}\right)^2.
\]
Then the attacker of \Cref{thm:multi-main} succeeds with probability greater
than \(\alpha/2\) and makes
\[
  O\!\left(
    \frac{
      \max\{1,q_{\cA,\mathsf{pre}}\}
      \max\{1,q_{\cB,\mathsf{pre}}\}
      Q_{\mathsf{fin}}^6}
      {\alpha^6}
    +q_{\cA,\mathsf{fin}}
  \right)
\]
queries. It follows that any family of such protocols with
inverse-polynomial valid agreement and honest query complexity bounded by
\(\poly(\lambda)\) is insecure against a computationally unbounded quantum
attacker making at most \(\poly(\lambda)\) oracle queries.
The query bound depends on the number of prefix rounds only through the total
honest prefix-query bounds.
\end{corollary}

\begin{proof}
The choice of \(\eta\) satisfies \(\eta\le1/256<1/100\), and the
reprogramming loss is
\(2Q_{\mathsf{fin}}\sqrt\eta=\alpha/8\).  Since
\(Q_{\mathsf{fin}}\ge1\) and \(0<\alpha\le1\),
\(3\eta\le3\alpha/256\).  Thus the success lower bound is strictly greater
than \(\alpha/2\).  Substituting
\(1/\eta^3=(16Q_{\mathsf{fin}}/\alpha)^6\) gives the query bound.
\end{proof}

\section{Attacks on Constant-Round Key Agreement with One Classical-Query Party}
\label{sec:quantum-bob-classical-alice}

We now prove the attack outlined in \Cref{sec:overview-quantum-bob} for
protocols in which all messages are classical and Alice makes only
classical oracle queries, while Bob may query quantumly throughout and
both parties may retain quantum states. Allowing imperfect correctness,
we begin with the two-message order \(\cB_1\to\cA\to\cB_2\),
since the Alice-first case is covered by \Cref{sec:main-attack}, and
then extend the argument to a constant number of rounds.

The aim is to learn enough oracle answers that the parties have little
common information left beyond what Eve knows, which we make precise by
constructing a nearby execution in which their keys are independent
conditioned on her view. Once that view is fixed, Alice's key is just
an independent guess for Bob's key, so Eve can do at least as well by
choosing his most likely nonabort key. The main task is therefore to
reach this independent execution while controlling how much the joint
distribution of the keys and Eve's view changes.

\paragraph{The attacker.}
After receiving the complete transcript of the honest execution, Eve
runs a learner \(L_1\) using Bob's first message and then a learner
\(L_2\) using both messages, as specified in
\Cref{def:qb-ca-two-learners}. She computes their query choices
classically from her current view and the reference experiment of
\Cref{def:qb-ca-reference-experiment}, obtaining every actual answer
from the honest oracle \(H\). When both learners succeed, she outputs
\(K^*=R(E)\) using the classical recovery rule of
\Cref{lem:qb-ca-recovery}, and otherwise outputs \(\bot\), with
success defined by \(K^*=K_{\cB}\ne\bot\) as in
\Cref{def:qb-ca-two}.

\paragraph{Proof outline.}
For the analysis, we place each learner immediately after the honest
message it processes, which preserves the final distribution because
Eve only reads the fixed oracle and sends nothing to the parties.
The two oracle reprogramming steps address different sources of dependence:
the first learner handles information
Bob already holds about the oracle, while the second handles information
he could obtain from Alice's private answers when he resumes.

For the first change, we need to account for the oracle information Bob
already holds, since giving Alice an independent answer at an input
\(x\) whose value \(H(x)\) he has recorded could destroy a correlation
used for agreement. Since Bob may have queried in superposition,
recording his queries as a classical list could disturb his state, so
we use the learner \(L_1\) supplied by
\Cref{lem:standalone-learning-switching}. After successful learning,
the lemma lets us give Alice and \(L_2\) independent uniform answers
at fresh inputs with only a small change to the joint output state,
keeping Eve's learned answers in place and retaining every new answer
in the oracle used for Bob's remaining computation. Retaining these
answers ensures that Bob still receives the same value if he later
queries an input Alice used, so the first change separates her fresh
answers from his earlier oracle information while leaving him able to
learn those answers when he resumes.

The second learner, \(L_2\), reduces this remaining dependence by
querying inputs likely to occur in Alice's classical query record,
with these probabilities conditioned on Eve's current view in the
reference experiment of \Cref{def:qb-ca-reference-experiment}.
After it succeeds, every input outside Eve's learned table has only a
small conditional probability of having been queried by Alice, while
her retained state and private answers are independent of Bob's
retained state and the original oracle, conditioned on that same view.
Together, these facts limit how much Bob's quantum queries can detect
a change at Alice's remaining private inputs, so
\Cref{lem:qb-ca-removal} lets us replace her private answers in his oracle
by the original oracle values at small cost while keeping Eve's learned
answers in place. Alice retains her state and query history, but Bob's
subsequent computation now uses only his own state, the original oracle,
and information fixed by Eve's view, making the resulting keys
independent conditioned on that view and giving the key-guessing
argument above.

If the real and final hybrid distributions are within distance \(\Delta\),
\Cref{lem:qb-ca-recovery} transfers this guarantee to success at least
\(\alpha-2\Delta\), where \(\alpha\) is the real valid agreement
probability. The bound on \(\Delta\) in \Cref{thm:qb-ca-two} accounts
for both oracle reprogramming steps and learner failures, and the loss of
\(2\Delta\) reflects two transfers: one for the agreement probability
and one for Eve's guessing probability when returning to the real execution.

For a constant number of rounds, we repeat these two reprogramming steps while retaining both parties' states, so later stages can
use all the information, including quantum information, kept from earlier
stages. The error bounds for these two reprogramming steps, together with
the learner-failure probabilities, let us build the hybrids recursively
and sum the resulting errors across rounds, yielding \Cref{thm:qb-ca-rounds} and the polynomial
query bound for a fixed number of rounds in \Cref{cor:qb-ca-constant}.

Throughout this section, we fix the security parameter \(\lambda\)
and use the oracle and key spaces from \Cref{sec:preliminaries},
taking every query bound to be a worst-case bound in every experiment
considered below. Both honest parties may use unrestricted local
quantum computation.

Using the notation of \Cref{sec:learner-records}, we write Eve's
current view as \(E=(M,J,F)\), her learned table as \(h\), and
the parties' retained registers as \(A,B\), with Alice's register
including her complete ordered query--answer record \(T_{\cA}\).
The separate table \(P\) contains her answers still private from Eve
that are used when the proof changes the oracle, and
\Cref{tab:qb-ca-records} summarizes these data and their updates.

\begin{table}[htbp]
\centering
\caption{Registers and oracle data used in the proof.}
\label{tab:qb-ca-records}
\begin{tabularx}{\linewidth}{@{}lX@{}}
\toprule
Data & Contents and updates \\
\midrule
\(E=(M,J,F)\) & \(M\) is the processed message prefix, \(J\) is the
ordered list of learner query--answer triples \((L_j,x,y)\), and \(F\)
lists termination pairs \((L_j,f)\) recording success or failure.
Eve appends each new entry to its corresponding list. \\
\(h\) & Eve's learned table, formed from \(J\). Every learner answer
adds a new input and its value. \\
\(T_{\cA}\) & Alice's complete ordered query--answer record.
Each Alice query is appended, including repeated queries. \\
\(P\) & Alice's query--answer pairs still private from Eve and used by
the current oracle.
A new private input
adds an entry, and a learner query at that input deletes the entry. \\
\(A,B\) & The parties' retained registers, updated by their respective
stages and including their keys once these are produced. \\
\bottomrule
\end{tabularx}
\end{table}

\subsection{Two-Message Protocols}
\label{sec:qb-ca-two}

\begin{definition}[Bob-first two-message protocol]
\label{def:qb-ca-two}
A Bob-first two-message key-agreement protocol is a tuple of oracle
quantum algorithms \(\Pi^{\mathsf{qb}}=(\cB_1,\cA,\cB_2)\).
Its honest experiment, denoted
\(\Real_{\Pi^{\mathsf{qb}}}(1^\lambda)\), samples
\(H\leftarrow\cH\) and proceeds as follows:
\begin{enumerate}[leftmargin=2em]
  \item Bob runs
  \((M_1,S_{\cB})\leftarrow\cB_1^H(1^\lambda)\), sends the
  classical message \(M_1\), and retains the possibly quantum
  register \(S_{\cB}\).
  \item Alice runs
  \((M_2,K_{\cA})\leftarrow\cA^H(1^\lambda,M_1)\), sends
  the classical message \(M_2\), and outputs her classical key.
  \item Bob runs
  \(K_{\cB}\leftarrow\cB_2^H(1^\lambda,S_{\cB},M_2)\)
  and outputs his classical key.
\end{enumerate}
Both keys lie in \(\cK_\lambda\cup\{\bot\}\), and the oracle access
is as in \Cref{def:oracle-algorithms}: Alice queries classically, while
Bob may query quantumly in both stages. The parties' private work
registers start in the all-zero state with no shared state or correlated
randomness, and any local state preparation or generation of private
randomness is part of their algorithms.
Let \(q_{\cB,1},q_{\cA},q_{\cB,2}\) bound the oracle queries
in the three stages, respectively. Define the valid agreement probability
\[
  \alpha_{\Pi^{\mathsf{qb}}}(\lambda)
  :=\Pr_{\Real_{\Pi^{\mathsf{qb}}}(1^\lambda)}
       [K_{\cA}=K_{\cB}\ne\bot].
\]
We write \(\alpha:=\alpha_{\Pi^{\mathsf{qb}}}(\lambda)\)
when the protocol and security parameter are fixed.

For a passive attacker \(\cE\), the key-recovery experiment
\(\mathsf{KR}_{\Pi^{\mathsf{qb}},\cE}(1^\lambda)\) first runs
the honest experiment to completion and then gives the complete
classical transcript \(M_{\mathrm{full}}=(M_1,M_2)\) to \(\cE\), which has access
to the same oracle \(H\) and outputs a classical guess
\(K^*\in\cK_\lambda\cup\{\bot\}\) by running
\[
  K^*\leftarrow\cE^H(1^\lambda,M_1,M_2).
\]
The attacker knows the protocol description but receives no private
registers or oracle query records from either party. The experiment
outputs \(1\) if and only if \(K^*=K_{\cB}\ne\bot\). Define
\[
  \Adv^{\mathsf{KR}}_{\Pi^{\mathsf{qb}},\cE}(\lambda)
  :=\Pr[\mathsf{KR}_{\Pi^{\mathsf{qb}},\cE}(1^\lambda)=1],
\]
where the probability is over the sampled oracle and all randomness
and measurement outcomes in the experiment.

We use the following notions of completeness and security:
\begin{enumerate}[leftmargin=2em,label=\textup{(\roman*)}]
  \item \textbf{Completeness.} The protocol has completeness error
  \(\epsilon_c(\lambda)\) if
  \[
    \alpha_{\Pi^{\mathsf{qb}}}(\lambda)\ge1-\epsilon_c(\lambda).
  \]
  We say the protocol is \emph{correct} if \(\epsilon_c\) is negligible.
  \item \textbf{Security.} The protocol is \emph{passive key-recovery
  secure against quantum attackers} if, for every computationally
  unbounded passive quantum attacker \(\cE\) making at most
  \(\poly(\lambda)\) quantum oracle queries,
  \[
    \Adv^{\mathsf{KR}}_{\Pi^{\mathsf{qb}},\cE}(\lambda)
       \le\negl(\lambda).
  \]
  Security against classical attackers is defined by restricting
  \(\cE\) to classical computation and classical oracle queries.
\end{enumerate}
\end{definition}

In the theorem below, \(\beta\) controls the error when replacing
shared oracle answers by independent answers after the first learner
succeeds, while \(\gamma\) is the threshold for the probability that
Alice queried an input, used by the learner after her reply.
The parameters \(\eta_1,\eta_2\) bound the
respective failure probabilities in the reference experiment of
\Cref{def:qb-ca-two-learners} and
determine the query limits \(N_1,N_2\).

\begin{theorem}[Two-message attack with quantum Bob]
\label{thm:qb-ca-two}
For \(0<\beta,\gamma,\eta_1,\eta_2<1\), set
\[
  N_1:=\ceil{\frac{q_{\cB,1}}{\eta_1\beta}},\qquad
  N_2:=\ceil{\frac{q_{\cA}}{\eta_2\gamma}},
\]
and
\[
  \Delta_2:=\eta_1+\eta_2+
       2(q_{\cA}+N_2)\sqrt\beta+2q_{\cB,2}\sqrt\gamma.
\]
Every protocol \(\Pi^{\mathsf{qb}}\) in \Cref{def:qb-ca-two} has a
computationally unbounded passive classical attacker \(\cE\) that
makes at most \(N_1+N_2\) classical oracle queries and satisfies
\[
  \Adv^{\mathsf{KR}}_{\Pi^{\mathsf{qb}},\cE}(\lambda)
    \ge\alpha-2\Delta_2.
\]
The attacker receives \((M_1,M_2)\) after the honest execution and
obtains all its oracle information through two learners, after which
it computes \(K^*\) deterministically from their recorded view
without further oracle queries. If the protocol has completeness
error \(\epsilon_c(\lambda)\), the advantage is at least
\(1-\epsilon_c(\lambda)-2\Delta_2\).
\end{theorem}

\paragraph{Structure of the proof.}
Let \(L_1\) denote Eve's learner using \(M_1\), and \(L_2\)
her learner using both messages. We analyze the order
\[
  \cB_1,\ L_1,\ \cA,\ L_2,\ \cB_2.
\]
In the proof of \Cref{thm:qb-ca-two}, we justify that this gives
the same joint distribution of the honest keys and Eve's final view as
running the learners after Eve receives the full transcript.
We write \(\cA+L_2\) for Alice's stage followed by the learner
\(L_2\).

The proof has three steps: first, we construct \(L_1\) and show
that, after it succeeds, we can replace the oracle answers during
\(\cA+L_2\) by independent answers at fresh inputs, retaining
every answer in a table and using those values throughout Bob's remaining
computation. Second, we construct \(L_2\) and show that we can
remove the remaining private answers from Bob's oracle, while retaining
Alice's state and private table. Finally, the keys in the resulting
experiment are conditionally independent, allowing Eve to guess Bob's
key by classical computation.

The learners compute their conditional probabilities using the reference
experiment of \Cref{def:qb-ca-reference-experiment}, in which Bob queries
\(H[h]\) and Alice uses recorded answers and independent uniform
values at fresh inputs. We first construct the learner after Bob's
initial message, and in \Cref{sec:qb-ca-constant} extend its construction
to prefixes containing earlier stages and learners.

\subsubsection{Switching to Independent Answers after Learning}

The following lemma constructs the learner \(L_1\) that runs after
Bob's first message, while Alice has not yet run and both tables are
empty. Once \(L_1\) succeeds, we can give independent answers to fresh
queries made by Alice and \(L_2\), with a bound on the change in the
final state that also covers Bob's last stage when its oracle retains
every answer recorded during \(\cA+L_2\).

We first define the two experiments used to express this guarantee.

\begin{definition}[Shared and independent oracle answers]
\label{def:oracle-switching-experiments}
Let \(\cX\) be a finite ordered set and
\(\cY=\{0,1\}^{\ell}\). Fix a classical record \(E=e\), its
learned table \(h=h(e)\), and a joint state of quantum registers
\(A,B\), a classical table \(P\) whose domain is disjoint from that of \(h\),
and a classical register storing a function \(H:\cX\to\cY\).
The register containing \(H\) is retained for the analysis and is
not given to the procedures below.

Fix a procedure \(\mathcal C\) that acts on \(A\), reads \(E\),
and may append classical observations to it, making at most \(n\)
classical queries to the oracle specified separately in each experiment
below, while \(B\) is idle and otherwise allowing quantum computation. 

Both experiments below reuse answers already in \(h\)
or \(P\), and designate each query of \(\mathcal C\) as private
or public: a private answer is stored in \(P\) unless its input
is already in \(h\), while a public query--answer pair is appended
to \(E\) and stored in \(h\), deleting any entry at that input
from \(P\). In our application, \(\mathcal C\) consists of an
Alice stage followed by a learner, with Alice's queries designated
private and the learner's queries designated public.

After \(\mathcal C\), both experiments run the same fixed finite
continuation \(\mathcal T\) on the resulting registers and tables
\(h',P'\), querying \(H[h'][P']\) with no other access to \(H\)
and allowing any finite number of classical or quantum queries.
Starting from the same fixed state, define:
\begin{enumerate}[leftmargin=2em,itemsep=0.3em]
  \item \(\mathsf{SharedAnswers}\): run \(\mathcal C\), answering
  each query \(x\) outside both tables with \(H(x)\), and then
  run \(\mathcal T\).
  \item \(\mathsf{IndependentAnswers}\): run \(\mathcal C\),
  answering each query outside both tables with a fresh uniform value
  from \(\cY\), independently of \(H\) and the preceding
  computation, and then run \(\mathcal T\).
\end{enumerate}
Thus \(\mathsf{SharedAnswers}\) leaves the oracle function unchanged
from \(H[h][P]\) at the start of \(\mathcal C\), while
\(\mathsf{IndependentAnswers}\) retains every newly sampled answer
for all subsequent queries at its input, including those of \(\mathcal T\).
Write \(\rho_{\mathrm{out}}^{\mathsf{SharedAnswers}}\) and
\(\rho_{\mathrm{out}}^{\mathsf{IndependentAnswers}}\) for the final
states of the same retained registers, which may include all algorithm
registers, \(E\), and the final tables, but exclude a separate copy
of the original oracle \(H\).
\end{definition}

We now state the learner's guarantees for the two-message protocol,
using \(\mathcal C=\cA+L_2\) and \(\mathcal T=\cB_2\) in these
experiments. 
The learner \(L_1\) is chosen before \(L_2\), and its
guarantees hold for every subsequent choice of \(L_2\) with the stated
query bound. The version that permits earlier rounds is stated
separately in \Cref{lem:general-learning-switching}.

\begin{lemma}[Classical learning in the two-message setting]
\label{lem:standalone-learning-switching}
Fix a two-message protocol \(\cB_1\to\cA\to\cB_2\) as in
\Cref{def:qb-ca-two}, with at most \(q_{\cB,1}\) oracle queries
by \(\cB_1\) and at most \(q_{\cA}\) classical oracle queries
by Alice. Sample the uniform oracle \(H\) and run \(\cB_1^H\),
retaining Bob's state \(B\) and recording his first message \(M_1\)
in Eve's view \(E\), while Alice's register remains in its initial state
and \(h=P=\varnothing\).

For every \(\beta,\eta\in(0,1)\), there is a deterministic
classical learner \(L_1\), depending only on
\(\cB_1,q_{\cB,1},\beta,\eta\),
that, in the execution used for the analysis, runs immediately after
Bob's first message and before Alice's stage, with the following guarantees.
\begin{enumerate}[leftmargin=2em,itemsep=0.3em]
  \item \textbf{Queries and failure probability.}
  The learner makes at most
  \[
    N_1=\ceil{\frac{q_{\cB,1}}{\eta\beta}}
  \]
  classical queries to \(H[h]\), and fails with probability at most
  \(\eta\), conditioned on any first message \(M_1\) of positive
  probability.\footnote{Here \(H[h]=H\) throughout \(L_1\), because
  \(h\) starts empty and contains only answers obtained from \(H\).}
  \item \textbf{Information used.}
  The learner computes each query and its stopping decision from
  \(\cB_1,\allowbreak q_{\cB,1},\allowbreak\beta,\eta,E\) by unrestricted classical computation,
  without accessing either party's private state. It queries only
  inputs outside \(\Dom(h)\), appends each query--answer pair to
  \(E\), stores it in \(h\), and records whether it succeeds or
  fails. At a view of zero probability in this experiment, it fails.
  \item \textbf{Reprogramming after successful learning.}
  Fix any successful terminal view \(E=e\) of \(L_1\) with
  positive probability. For any subsequently fixed classical learner
  \(L_2\) whose actions depend only on its current view and that
  makes at most \(N_2\) queries, run
  \(\mathsf{SharedAnswers}\) and \(\mathsf{IndependentAnswers}\)
  from \Cref{def:oracle-switching-experiments} with
  \(\mathcal C=\cA+L_2\) and \(\mathcal T=\cB_2\). 
  Both experiments start from the conditional state after
  \(\cB_1^H,L_1^H\), with learned table \(h(e)\), empty private
  table \(P\), and Alice's initial register. Bob's final stage uses
  \(H[h'][P']\), retaining every answer recorded during
  \(\cA+L_2\). Their final states satisfy
  \[
    D\bigl(\rho_{\mathrm{out}}^{\mathsf{SharedAnswers}},
           \rho_{\mathrm{out}}^{\mathsf{IndependentAnswers}}\bigr)
      \le 2(q_{\cA}+N_2)\sqrt\beta.
  \]
  \item \textbf{Independence before Bob resumes.}
  In \(\mathsf{IndependentAnswers}\), Alice's retained state and
  private table \(AP\) are independent of \(BH\), conditioned on
  every resulting view of positive probability after \(\cA+L_2\)
  and before \(\cB_2\). Here \(H\) denotes the original oracle,
  retained as a classical register only for the analysis.
\end{enumerate}
\end{lemma}

\begin{proof}[Proof of \Cref{lem:standalone-learning-switching}]
Write \(\mathcal C=\cA+L_2\), \(\mathcal T=\cB_2\),
\(n=q_{\cA}+N_2\), and \(N=N_1\) for the query limit of
\(L_1\). The learner and its query-count analysis adapt the Fourier heavy-query
learner of Austrin et al.~\cite[Definition~4.9, Construction~4.10,
and Lemma~4.12]{AustrinEtAl22}. We explain how their learner and
query bound apply to the conditional states and recorded tables used
here, then bound the trace
distance between \(\mathsf{SharedAnswers}\) and
\(\mathsf{IndependentAnswers}\), including the subsequent computation.

\paragraph{The Fourier oracle.}
We use Zhandry's compressed oracle in the Fourier basis from
\Cref{def:compressed-oracle}, with oracle register \(\mathsf H\)
and database \(D_u\), to identify the inputs at which the learner
should query and then translate the resulting bounds back to the
oracle executions of \Cref{def:oracle-switching-experiments}. Recall that
\(D_u\) records Fourier labels, whereas the tables \(h,P\) store actual
oracle answers. For each \(x\in\cX\), define
\[
  \Pi_x:=I-\proj{\widehat0}_{\mathsf H_x}.
\]
The quantity \(\Tr(\Pi_x\rho)\) is the probability that a
Fourier-basis measurement finds a nonzero label at \(x\).

\paragraph{Representing Bob's first stage and the tables.}
To represent the oracle \(H[h][P]\), queries use the table entries
on their domains and \(\mathsf H_x\) elsewhere, coherently when
the query is quantum. Once a value is stored in a table, subsequent
shared queries at that input use the table, so the unused underlying
register may be reset by the channel
\[
  \mathcal R_x(\rho)
    :=\proj{\widehat0}_{\mathsf H_x}\otimes
      \Tr_{\mathsf H_x}\rho,
\]
with the tensor factors restored to their original order. This channel
discards \(\mathsf H_x\) and prepares a fresh register in state
\(\ket{\widehat0}\), leaving the reduced state of the other registers
unchanged. These resets
are changes of the proof representation, not queries or operations
performed by the learner. An entry in \(h\) remains present permanently, even
when a private entry is deleted, so future queries never need its
previous underlying value.

Initially, every oracle coordinate is in state
\(\ket{\widehat0}\), independently of the work registers, and
each query by \(\cB_1\) adds at most one nonzero Fourier label.
The state after \(\cB_1\) is therefore supported on Fourier basis
states with at most \(q_{\cB,1}\) nonzero labels, even after
conditioning on any first message of positive probability.

Fix a first message \(M_1=m_1\) of positive probability and let
\(e_0\) be Eve's view before \(L_1\) begins, containing only that
message and no learned query--answer pairs. Running \(\cB_1\) with
the Fourier oracle and conditioning on this message gives the joint
state \(\sigma^{e_0}_{B\mathsf H}\) of Bob's retained register and
the quantum oracle register. We include Alice's initial register and
her empty private table by defining
\[
  \tau^{e_0}_{APB\mathsf H}
    :=\proj{0}_A\otimes\proj{\varnothing}_P
        \otimes\sigma^{e_0}_{B\mathsf H}.
\]
Thus \(\tau^{e_0}_{APB\mathsf H}\) represents the joint state
immediately after Bob's first message and before the learner's first
query, with Alice still in her initial state, \(P\) empty, and all
correlations between Bob and the oracle retained. We use \(\tau\) for
these auxiliary learner states with quantum oracle register \(\mathsf H\),
distinguishing them from the corresponding states \(\rho\) in the
ordinary oracle experiments. The state \(\tau\) is used to analyze the learner's
queries and is not a quantum state given to Eve.
Measuring the oracle coordinates in the computational basis would give
the ordinary oracle and Bob's corresponding conditional state, but these
measurements can be deferred because queries read the coordinates
without changing their computational-basis values. Once an answer is
recorded in a table, later queries use that answer, so resetting its
underlying coordinate does not affect the subsequent computation.
Thus this auxiliary state gives the same learner observations and
retained output states as the ordinary execution, while \(AP\) is
initially a fixed state independent of \(B\mathsf H\).

At a given step of the simulated learner, let
\(\tau_{APB\mathsf H E}\) be the joint state obtained from
the initial auxiliary state, with \(E\) recording the learner's
current view. For each possible view \(e\), set
\(Q_e:=I_{APB\mathsf H}\otimes\proj{e}_E\) and define
\begin{equation}
  \begin{gathered}
    p(e):=\Tr\!\left(Q_e\tau_{APB\mathsf H E}\right),\qquad
    \tau^e_{APB\mathsf H}
      :=\frac{\Tr_E\!\left(
        Q_e\tau_{APB\mathsf H E}Q_e\right)}{p(e)},\\
    \tau^e_{AP}:=\Tr_{B\mathsf H}\tau^e_{APB\mathsf H},
    \qquad
    \tau^e_{B\mathsf H}:=\Tr_{AP}\tau^e_{APB\mathsf H},
  \end{gathered}
  \label{eq:qb-ca-auxiliary-conditioning}
\end{equation}
where the conditional states are defined whenever \(p(e)>0\).

Every learner observation comes from \(B\mathsf H\), while
\(A\) remains in its initial state and \(P\) stays empty, so
\[
  \tau^e_{APB\mathsf H}
     =\tau^e_{AP}\otimes\tau^e_{B\mathsf H}.
\]

\paragraph{Constructing the learner.}
We now define the learner's query rule.  

For its current view \(e\), compute the normalized conditional state
\(\tau^e_{B\mathsf H}\) by simulating \(\cB_1\) and the
learner's queries so far in this representation, and set
\[
  w_e(x):=\Tr(\Pi_x\tau^e_{B\mathsf H})
       \qquad(x\notin\Dom(h)).
\]
We call \(w_e(x)\) the \emph{Fourier weight} at \(x\): it is
the conditional probability of a nonzero Fourier label at that input.
If all weights are at most \(\beta\), return success. Otherwise
query the first input whose Fourier weight exceeds \(\beta\), append
its input and answer to \(E\), add the answer to \(h\), and repeat.   
In the simulated computation using the auxiliary quantum oracle register
\(\mathsf H\), also apply \(\mathcal R_x\) to reset the coordinate
whose value has just been recorded. This reset is part of the analysis
and does not act on the actual oracle or Bob's state.  
Return failure if a further query is required after
\(N\) queries, or at a view of zero probability in the defining
experiment. The learner computes the state from the experiment's
description, without receiving \(B\) or a Fourier measurement outcome.

\paragraph{Bounding the failure probability.}
Fix a starting view \(e_0\) of positive probability and let
\(N_{\mathrm{full}}\) be the number of queries made by the same
learner without its query limit. Since the learner tests for success
before enforcing the limit \(N\), it fails precisely when
\(N_{\mathrm{full}}>N\), apart from views of zero probability,
which cannot occur in this experiment. We therefore need to prove
\[
  \Pr[N_{\mathrm{full}}>N\mid E_{\mathrm{start}}=e_0]\le\eta,
\]
where \(E_{\mathrm{start}}\) denotes the view before the learner begins.

Apply the expected-query bound of Austrin et al.~\cite[Lemma~4.12]{AustrinEtAl22}
with their threshold \(\varepsilon=\beta\) and their Fourier
support bound \(d=q_{\cB,1}\). Their list \(L\) of learned
query--answer pairs corresponds to our table \(h\), and their total
weight \(W\) on unqueried inputs is
\[
  \Phi(e):=\sum_{x\notin\Dom(h)}w_e(x).
\]
The bound in the proof of their lemma applies from any starting
node of the learner's execution, so we apply it after conditioning on
\(e_0\), where the Fourier support bound already proved gives
\(\Phi(e_0)\le q_{\cB,1}\). Our use of mixed states is covered
by purifying them with an untouched auxiliary register, and resetting
a recorded oracle coordinate leaves all unrecorded coordinates
unchanged, so it does not affect the weights used in their bound.
Finally, our rule queries only inputs with weight strictly greater
than \(\beta\), which satisfies the lower bound of \(\beta\)
required by their argument. Consequently,
\[
  \mathbb E[N_{\mathrm{full}}\mid E_{\mathrm{start}}=e_0]
    \le\frac{\Phi(e_0)}{\beta}
    \le\frac{q_{\cB,1}}{\beta}.
\]

For \(q_{\cB,1}>0\), Markov's inequality and
\(N=\ceil{q_{\cB,1}/(\eta\beta)}\) now give
\[
  \Pr[N_{\mathrm{full}}>N\mid E_{\mathrm{start}}=e_0]
    \le\frac{q_{\cB,1}}{\beta N}\le\eta.
\]
If \(q_{\cB,1}=0\), every Fourier weight is zero and the learner
succeeds without queries, so the same failure bound holds. In either
case, success means that all remaining weights are at most \(\beta\).

\paragraph{Switching to independent answers.}
Fix a view \(e\) at which the learner succeeds. We have
\begin{equation}
  \tau^e_{APB\mathsf H}
    =\tau^e_{AP}\otimes\tau^e_{B\mathsf H},\qquad
  \Tr(\Pi_x\tau^e_{B\mathsf H})\le\beta
       \quad(x\notin\Dom(h)).
  \label{eq:qb-ca-independent-answers-hyp}
\end{equation}
The register \(B\) is idle during \(\mathcal C\), whose queries in
\(\mathsf{SharedAnswers}\) and \(\mathsf{IndependentAnswers}\)
are answered by one procedure with parameter \(\nu\in\{1,2\}\).
Mode~1 supplies the answers for \(\mathsf{SharedAnswers}\), and
mode~2 supplies those for \(\mathsf{IndependentAnswers}\).
The difference is shown in \changed{color}.
\begin{algorithm}[t]
\caption{Answering classical queries after learning}
\label{alg:qb-ca-answer}
\KwInput{Mode \(\nu\), registers \(A,B,\mathsf H\), tables \(h,P\),
and the next classical query \(x\), designated private or public.}
\begin{enumerate}[leftmargin=2em,itemsep=0.3em]
  \item If \(x\in\Dom(h)\), set \(y=h(x)\) and proceed to
  Step~3. Otherwise, if \(x\in\Dom(P)\), set \(y=P(x)\)
  and proceed to Step~3.
  \item At an input outside both tables, prepare
  \(V=\ket{\widehat0}\), apply
  \[
    \changed{U_\nu:=
      \begin{cases}
        \operatorname{SWAP}_{\mathsf H_x,V},&\nu=1,\\
        I,&\nu=2,
      \end{cases}}
  \]
  and measure \(V\) in the computational basis to obtain \(y\).
  \item For a private query, store the answer in \(P\) if the input
  is outside \(\Dom(h)\).
  For a public query, append \((x,y)\) to \(E\), add it to \(h\), delete
  any entry of \(P\) at \(x\), and apply \(\mathcal R_x\).
  Return \(y\) to the querying procedure.
\end{enumerate}
\end{algorithm}
For \(\nu=1\), swapping the registers and measuring \(V\)
returns the original oracle value, which is retained in a table,
and leaves the oracle register reset, while for \(\nu=2\) the
measured value is independent and uniform. In both modes, each public
answer is followed by a reset of the underlying register.
Since existing table entries have identical answer and update rules in
both modes, it remains to bound the trace distance caused by a query at
a fresh input, retaining all registers used later.

Representing reprogramming by a swap with a fresh oracle register also
appears in Jaeger~\cite[Section~4.2]{Jaeger24}. Here we use the
learner's Fourier-weight guarantee to bound the effect of the swap
while retaining the algorithm registers and recorded answers, so the
bound also covers the subsequent computation with the reprogrammed oracle.

On a purification \(\ket\psi\) of the state just before a fresh
query at \(x\),
\[
  \bigl\|(\operatorname{SWAP}_{\mathsf H_x,V}-I)
       \ket\psi\ket{\widehat0}_V\bigr\|_2
  \le2\|\Pi_x\ket\psi\|_2.
\]
The swap fixes the component on which both registers contain
\(\ket{\widehat0}\), and its difference from the identity has norm
at most two on the orthogonal component.

Number the oracle calls of \(\mathcal C\) in order, appending
identity operations after termination if it makes fewer than \(n\)
calls, so that the extra positions make no queries and change no
records. For \(0\le j\le n\),
define the following experiment.
\begin{hybrid}{Hybrid \(\mathsf C_j\)}
  \item Start from the conditional state and tables after the common
  successful learner view \(e\).
  \item Run \(\mathcal C\), using \changed{mode~2 for its first
  \(j\) query positions} and mode~1 thereafter, with the same
  record and table updates throughout.
  \item Retain \(APB\mathsf H E\) and both final tables.
\end{hybrid}
Hybrids \(\mathsf C_{j-1}\) and \(\mathsf C_j\) have a
common prefix of independent answers, whose classical records we
condition on, including the input to the next query. These records
come from the \(AP\) factor and independent randomness, and the
only operations they specify on \(B\mathsf H\) are resets at
learned inputs. At a fresh input \(x\), none of those resets
acted on \(\mathsf H_x\), while resets at other coordinates leave
its marginal unchanged, so its weight remains at most \(\beta\).
The displayed bound for the swap therefore gives distance at most
\(2\sqrt\beta\) for this query, while queries at inputs already
in a table and identity operations contribute zero. Averaging over the
common prefixes and applying the same remaining operations gives
\[
  D(\rho^{\mathsf C_{j-1}},\rho^{\mathsf C_j})
    \le2\sqrt\beta.
\]
The triangle inequality now yields
\[
  D(\rho^{\mathsf C_0}_{APB\mathsf H E},
    \rho^{\mathsf C_n}_{APB\mathsf H E})\le2n\sqrt\beta.
\]
\paragraph{Returning to ordinary oracle executions.}
In \(\mathsf C_0\), every new answer is the shared oracle value,
retained in a table while its underlying register is reset, so
subsequent queries at that input receive the recorded value. In
\(\mathsf C_n\), fresh answers are independent and uniform, while
existing entries are retained. Thus \(\mathsf C_0\) and
\(\mathsf C_n\) implement the respective prefixes of
\(\mathsf{SharedAnswers}\) and \(\mathsf{IndependentAnswers}\)
through \(\mathcal C\). In both experiments,
Bob's final stage \(\mathcal T=\cB_2\) uses the recorded function
\(H[h'][P']\) and applies
the same channel to both joint states, so it cannot increase their trace
distance regardless of its query count. Tracing out the proof registers
then gives the claimed bound on the output states.

For the remaining independence claim, consider the ordinary experiment
with the original classical register \(H\). The learner preserves
the conditional product of \(AP\) and \(BH\), since its
observations come from \(BH\), while Alice's register and the empty
table \(P\) remain fixed. In \(\mathsf{IndependentAnswers}\), \(\mathcal C\)
then generates every new observation from \(AP\) and independent
randomness while \(BH\) is unchanged, so conditioning on the resulting
record preserves the product immediately before
\(\mathcal T\), as claimed.
This concludes the proof of \Cref{lem:standalone-learning-switching}.
\end{proof}

\paragraph{Applying the learner to the protocol.}
Apply \Cref{lem:standalone-learning-switching} to obtain \(L_1\)
after Bob's first message, then define \(L_2\) from the reference experiment
through Alice's reply, as specified in \Cref{def:qb-ca-two-learners}.
The switching guarantee applies with \(\mathcal C=\cA+L_2\),
which makes at most \(q_{\cA}+N_2\) classical queries, and
\(\mathcal T=\cB_2\). The reference experiment uses the order
\[
  \cB_1,\quad L_1,\quad \cA,\quad L_2,\quad \cB_2,
\]
with Bob and \(L_1\) querying \(H[h]\), while Alice and
\(L_2\) reuse table entries and sample independent answers elsewhere.
The next definition specifies these rules, allowing us to stop after
any of these stages to define a learner from the prefix already fixed.
In particular, \(L_1\) is constructed before \(L_2\), so choosing
\(L_2\) does not change the procedure that defines \(L_1\).

After \(L_1\), the experiment \(\mathsf{IndependentAnswers}\)
follows these reference rules through \(\cA+L_2\), and the
conditional-independence guarantee of
\Cref{lem:standalone-learning-switching} applies at that point.
Its continuation still uses \(H[h'][P']\), retaining Alice's private
answers, which \Cref{lem:qb-ca-removal} will let us remove from
Bob's oracle to reach the reference experiment.
The definition also permits longer sequences of stages, whose learner
construction is given separately in \Cref{sec:qb-ca-constant} using
\Cref{lem:general-learning-switching}.

\begin{definition}[Reference experiment for fixed stages and learners]
\label{def:qb-ca-reference-experiment}
Using the notation of \Cref{sec:learner-records} for
\(A,B,E=(M,J,F),h,P\), fix a finite ordered list
\((S_1,\ldots,S_t)\) of honest stages and learner calls, which we
call the \emph{schedule} of the experiment. 
Each honest stage uses only the corresponding party's registers and
the honest messages sent so far, with all communication and Alice's
oracle queries classical and Bob allowed quantum oracle access.
Each learner \(L_j\) runs after a specified honest message and has
a fixed query limit \(N_j\) and a deterministic rule for choosing
its next action,
\[
  L_j(e)\in\bigl(\cX\setminus\Dom(h(e))\bigr)
       \cup\{\mathsf{success},\mathsf{failure}\}.
\]
Also fix \(\mathsf{stop}(s,e)\in\{0,1\}\), determining when the
experiment stops from its position \(s\) in this schedule and view \(e\). 

Sample \(H\leftarrow\cH\), initialize
\[
  \rho_{AB}=\proj{0}_A\otimes\proj{0}_B,\qquad
  M=J=F=\langle\rangle,\qquad h=P=\varnothing,
\]
where \(\langle\rangle\) denotes an empty list. Execute the stages
and learner calls in this schedule until \(\mathsf{stop}(s,E)=1\)
or all stages and calls have finished, using these answers:
\begin{enumerate}[leftmargin=2em,itemsep=0.3em]
  \item \textbf{Alice and a learner following her message:}
  answer each classical query at input \(x\) with the value \(y\)
  given by 
  \[
    y=
    \begin{cases}
      h(x),&x\in\Dom(h),\\
      P(x),&x\in\Dom(P),\\
      U,&x\notin\Dom(h)\cup\Dom(P).
    \end{cases}
  \]
  In the last case, sample \(U\leftarrow\Unif(\cY)\) independently
  of \(H\) and prior computation.
  For Alice's queries, append \((x,y)\) to \(T_{\cA}\) in \(A\)
  and set \(P\leftarrow P\cup\{(x,y)\}\) if \(x\notin\Dom(h)\).
  \item \textbf{Bob and a learner following his message:}
  use \(H[h]\), ignoring \(P\), with Bob querying quantumly
  without updating the tables and the learner querying classically.
\end{enumerate}
Write \(V\mathbin{\|}v\) for appending an entry \(v\) to a list
\(V\). An honest message \(m\) gives \(M\leftarrow M\mathbin{\|}m\).
After learner \(L_j\) queries \(x\) and receives \(y\), set
\[
  J\leftarrow J\mathbin{\|}(L_j,x,y),\qquad
  h\leftarrow h\cup\{(x,y)\},\qquad
  P\leftarrow P|_{\cX\setminus\{x\}}.
\]
When \(L_j\) terminates with flag
\(f\in\{\mathsf{success},\mathsf{failure}\}\), set
\(F\leftarrow F\mathbin{\|}(L_j,f)\).
Thus \(h=h(E)\), its entries never change or disappear, and
\(\Dom(h)\cap\Dom(P)=\varnothing\).

During an honest stage, the other party only receives classical
messages, while both parties are idle during learner calls and neither
receives \(E,h,P\). After a learner fails, make no further learner
queries, but continue the honest stages with the same oracle-answer
rules until the stopping point. At termination, retain \(A,B,E,h,P\)
and, for the analysis, a classical register containing the original \(H\).
\end{definition}

In this experiment, a learner following Bob may receive a different
answer from the one Alice previously recorded in \(P\), in which
case the update places the learner's answer in \(h\) and removes
the private entry, while Alice's record \(T_{\cA}\) and residual
state still reflect the answer she originally received.

Every prefix of the experiment in \Cref{def:qb-ca-reference-experiment}
satisfies, at each view \(e\) of positive probability,
\begin{equation}
  \rho^e_{APBH}=\rho^e_{AP}\otimes\rho^e_{BH}.
  \label{eq:qb-ca-reference-product}
\end{equation}
Here \(H\) is the classical register containing the original oracle,
and the conditional states use \Cref{sec:state-notation}.
The identity holds initially because the work registers are zero,
the tables are empty, and the oracle is independent of them.
During an Alice stage and its learner, new observations come from
\(AP\) and independent randomness, while during a Bob stage and its
learner they come from \(BH\), so conditioning on those observations
leaves only fixed local updates to the other factor. Since a learner's
termination flag is determined by its current view and reveals no
further information, induction preserves the identity throughout.
In the two-message application, Alice has not run when \(L_1\)
starts, so \Cref{lem:standalone-learning-switching} applies directly.
Its query bound depends on \(q_{\cB,1}\), the number of queries
Bob has already made, while its distance bound depends on
\(q_{\cA}+N_2\), the combined queries of Alice and \(L_2\).
For later rounds, we verify the preparation conditions of
\Cref{lem:general-learning-switching} in \Cref{sec:qb-ca-constant}.

\subsubsection{Learning and Removing Alice's Private Answers}

The second learner runs after Alice's message and queries inputs that
have large conditional probability of belonging to her private table.
Once it succeeds, we can remove the remaining private answers from
the oracle used by Bob and the learner following his next message,
while retaining Alice's state and private table so that she can continue
her computation in later rounds.

The learning step uses the same heavy-query counting argument as
\Cref{lem:learn-success}. To bound the effect of removing private answers
from Bob's oracle, we apply \Cref{lem:reprogramming} to his stage,
then bound the probability that the following classical learner queries
a remaining private input. The argument retains Alice's quantum state
and private table for later rounds, while treating the learner separately
preserves the sharper error bound for its classical queries.

\begin{lemma}[Classical learning and removal of private answers]
\label{lem:qb-ca-removal}
Fix a view immediately after an Alice stage in the reference experiment
of \Cref{def:qb-ca-reference-experiment}, with Eve's record
\(E=(M,J,F)\). Let \(h_0\) be her learned table and \(P_0\)
Alice's remaining private table at that point, as in
\Cref{sec:learner-records}. Suppose
\(|\Dom(P_0)|\le s\). For \(\gamma,\eta\in(0,1)\),
there is a deterministic classical learner \(L\), defined from this
prefix, with query limit
\[
  N=\ceil{s/(\eta\gamma)}.
\]
Its queries are answered in the same way as Alice's, using existing
entries of \(h,P\) and fresh independent uniform values elsewhere,
and it adds every answer to \(h\).
Conditioned on any positive-probability starting view, its failure
probability is at most \(\eta\). If \(h\) is its final learned
table, set \(P=P_0|_{\cX\setminus\Dom(h)}\). For any successful
terminal view \(e\) of positive probability, the learner also guarantees
the trace-distance bound in \eqref{eq:qb-ca-removal-distance} between
\(\mathsf{UsePrivate}\) and \(\mathsf{IgnorePrivate}\), defined
below, with both experiments initialized from the joint state conditioned
on \(E=e\).

Consider a Bob stage followed by a single classical learner (omitting
that learner if Bob's stage is final), and call this portion of the
execution the \emph{Bob phase}.

During this phase, Alice only receives classical messages, and her
private table undergoes the specified updates. Suppose Bob makes at
most \(b\) oracle queries and the following learner makes at most
\(c\) classical queries, with \(c=0\) when it is absent. 
Bob may query quantumly, and his bound \(b\) includes any queries
he makes classically. 
Their algorithms, query-selection rules, and table updates are fixed. Define the following experiments, which differ only
in whether the oracle uses Alice's private table during this phase:
\begin{enumerate}[leftmargin=2em,itemsep=0.3em]
  \item \(\mathsf{UsePrivate}\): answer all queries in this phase using
  \(H[h][P]\), where the tables are updated after each learner answer.
  \item \(\mathsf{IgnorePrivate}\): answer them using \(H[h]\), with exactly
  the same table-update rules. Retain the undeleted part of \(P\)
  for Alice's subsequent stages.
\end{enumerate}
Let \(\rho^{\mathsf{UsePrivate}}\) and
\(\rho^{\mathsf{IgnorePrivate}}\) denote the respective joint states
at the end of this phase, retaining \(A,P,B,H,E\) and the learned
table. Then
\begin{equation}
  D\bigl(\rho^{\mathsf{UsePrivate}},
         \rho^{\mathsf{IgnorePrivate}}\bigr)
    \le 2b\sqrt\gamma+c\gamma.
  \label{eq:qb-ca-removal-distance}
\end{equation}
The same bound holds after any common continuation on these retained
registers, including one that again uses the remaining \(P\).
The experiment \(\mathsf{IgnorePrivate}\) preserves conditional independence of \(AP\)
and \(BH\) given the updated view. Here \(H\) denotes the
classical register containing the sampled original oracle.
\end{lemma}

In the two-message application, \(s=q_{\cA}\), \(b=q_{\cB,2}\),
and \(c=0\). In the constant-round application, the phase also
includes the next learner, and its queries account for the term
\(c\gamma\).

\begin{proof}\leavevmode
\paragraph{Learning likely private inputs.}
At the start of the learner, set \(S=\Dom(P_0)\), the set of
Alice's inputs that are still private at that point, which is determined
by the preceding execution and remains fixed throughout the learner's
queries. Although \(S\) is random over executions, it is not
resampled during learning, so only Eve's conditional distribution of
\(S\) changes as she obtains answers. At a current view \(e\), compute
\[
  p_e(x):=\Pr[x\in S\mid E=e]
       \qquad(x\notin\Dom(h))
\]
in the reference experiment of \Cref{def:qb-ca-reference-experiment},
including the learner's queries so far.
The learner succeeds if all these probabilities are at most
\(\gamma\), and otherwise queries the first input with larger
probability and adds its answer to \(h\), unless it has already
made \(N\) queries, in which case it fails. It also fails if the
current view has zero probability in the defining experiment.
Eve computes these probabilities from the experiment's description and
her view without needing to know the actual set \(S\). To answer
the learner's queries using a single function, complete
\(P_0\) to a function \(G\) using independent uniform values
elsewhere and answer its queries with \(G[h_0]\), which gives the
same distribution of answers as the procedure above. This completion
is part of the reference experiment in
\Cref{def:qb-ca-reference-experiment}, while Eve's queries in the
real attack are answered by \(H\).

Let \(N_{\mathrm{full}}\) be the number of queries the learner
would make without its query limit. Since its queries are distinct,
taking conditional expectations at each query gives
\[
  \gamma\,\mathbb E[N_{\mathrm{full}}]
   \le\mathbb E\!\left[
       \sum_{\text{queried }x}\boldsymbol1_{\{x\in S\}}\right]
   \le s.
\]
This is the heavy-query counting argument of \Cref{lem:heavy-bound},
from which Markov's inequality gives the failure bound conditioned
on the fixed starting view, with immediate success when \(s=0\).
For an input outside the final learned table, membership in
\(\Dom(P)\) is exactly membership in \(S\), so the stopping
condition bounds the probability that any such input remains private.
Since the learner's answers depend only on Alice's state, her private
table, and independent randomness, conditioning on them also preserves
\eqref{eq:qb-ca-reference-product}. Thus, for each successful terminal
view \(e\), the following properties hold in the reference experiment:
\begin{equation}
  \rho^e_{APBH}=\rho^e_{AP}\otimes\rho^e_{BH},\qquad
  \sup_x\Pr[x\in\Dom(P)\mid E=e]\le\gamma.
  \label{eq:qb-ca-remove-hyp}
\end{equation}
These are the intermediate properties we use to prove the learner's
stated guarantee for \(\mathsf{UsePrivate}\) and
\(\mathsf{IgnorePrivate}\).

\paragraph{Removing private answers from Bob's oracle.}
Fix a successful view \(e\), write \(h_{\mathrm{in}},P_{\mathrm{in}}\)
for the tables at the start of the Bob phase, and use
\eqref{eq:qb-ca-remove-hyp} to factor the conditional state into
\(AP_{\mathrm{in}}\) and \(BH\). We first change the oracle only
during Bob's stage, keeping the following learner's oracle rule fixed.

Apply \Cref{lem:reprogramming} with the following distinguisher.
Sample the classical \(H\) from the \(BH\) marginal conditioned
on \(e\), prepare \(B\) in its corresponding conditional state,
and set \(F_0=H[h_{\mathrm{in}}]\). The classical sampler \(C\)
outputs a table \(P_{\mathrm{in}}\) from its marginal conditioned
on \(e\), independently of \(BH\), so
\[
  F_1=H[h_{\mathrm{in}}][P_{\mathrm{in}}],\qquad
  \max_x\Pr[x\in\Dom(P_{\mathrm{in}})\mid e]\le\gamma.
\]
Run Bob's stage using the challenge oracle, retaining his output state
and classical message. After oracle access ends, the sampler's revealed
coins determine \(P_{\mathrm{in}}\), so prepare \(A\) in its
conditional state given that table and \(e\), and append Bob's
message to Alice's transcript and Eve's view. This delayed preparation
is valid because Alice does not compute during Bob's stage, and it
reproduces the complete joint state, including correlations between
\(A\) and \(P_{\mathrm{in}}\). The sampler itself is classical,
while the distinguisher prepares the residual quantum state after
learning its coins, exactly as allowed by \Cref{lem:reprogramming}.

Following this preparation by any binary measurement on the retained
registers gives a distinguishing advantage at most \(2b\sqrt\gamma\),
so the variational characterization of trace distance gives the same
bound for their joint states. The bound includes the original \(H\),
which the distinguisher sampled and retained, as well as both parties'
states and the private table. Running the following learner with the
same rule in both executions cannot increase this distance, even when
its oracle still uses the private table.

\paragraph{Removing private answers from the following learner's oracle.}
It remains to change that learner's oracle from \(H[h][P]\) to
\(H[h]\), now with Bob already using \(H[h_{\mathrm{in}}]\).
Analyze the learner that ignores \(P\), and let \(X_1,\ldots,X_k\)
be its query inputs, where \(k\le c\). Its entire query sequence
is generated from \(BH\) and the fixed initial view, so it is
independent of \(P_{\mathrm{in}}\) conditioned on \(e\).
Averaging the bound in \eqref{eq:qb-ca-remove-hyp} over each query
input and taking a union bound gives
\[
  \Pr[\text{some }X_j\in\Dom(P_{\mathrm{in}})\mid e]
    \le c\gamma.
\]
Couple the two learner executions using the same initial registers
and randomness. Until a query hits \(\Dom(P_{\mathrm{in}})\),
they receive the same answers, make the same subsequent choices, and
perform identical record and table updates, so their retained joint
states can differ only on an event of probability at most \(c\gamma\).
This also bounds their trace distance, including Alice's quantum state
conditioned on her private table.

Combining the two changes yields
\[
  D\bigl(\rho^{\mathsf{UsePrivate}},
         \rho^{\mathsf{IgnorePrivate}}\bigr)
    \le2b\sqrt\gamma+c\gamma.
\]
Applying a common continuation cannot increase trace distance, including
when Alice resumes using the undeleted part of \(P\).
In \(\mathsf{IgnorePrivate}\), every new observation in Eve's view
is generated from \(BH\), and fixing these observations leaves only
fixed local updates on Alice's side, so the conditional independence
of \(AP\) and \(BH\) is preserved as claimed.
This concludes the proof of \Cref{lem:qb-ca-removal}.
\end{proof}

\subsubsection{Recovering a Key from Conditional Independence}

The two learner lemmas will give an experiment whose output is close
to the real output and whose keys are independent conditioned on Eve's
view, allowing us to use the following lemma to obtain a classical
guess for Bob's key. We lose the distance bound twice: once for the change
in agreement probability and once when transferring Eve's success
probability back to the real execution.

\begin{lemma}[Key recovery from a conditional product distribution]
\label{lem:qb-ca-recovery}
Let \(\rho_{K_{\cA}K_{\cB}E}\) be the real joint distribution,
with valid agreement probability \(\alpha\). Suppose a finite
distribution \(\sigma_{K_{\cA}K_{\cB}E}\) can be computed from
the protocol description without queries to the actual oracle and
satisfies
\[
  D(\rho,\sigma)\le\Delta,\qquad
  \sigma=\sum_e p_e\proj e_E\otimes
       \sigma^e_{K_{\cA}}\otimes\sigma^e_{K_{\cB}}.
\]
Here \(p_e=\Pr_\sigma[E=e]\).
Suppose failure is determined by \(E\), and both keys in
\(\sigma\) are \(\bot\) on failure. There is a deterministic
map \(R\) from views to \(\cK_\lambda\cup\{\bot\}\),
computable without oracle queries and returning \(\bot\) on failure,
such that
\[
  \Pr_\rho[R(E)=K_{\cB}\ne\bot]\ge\alpha-2\Delta.
\]
\end{lemma}

\begin{proof}
For a nonfailure view \(e\) of positive probability in \(\sigma\),
choose, with the fixed tie-breaking order,
\[
  R(e)\in\operatorname*{arg\,max}_{k\in\cK_\lambda}
          \Pr_\sigma[K_{\cB}=k\mid E=e].
\]
Return \(\bot\) on failure or on a view of zero probability in
\(\sigma\). If \(\cK_\lambda\) is empty, the claim is immediate.
Otherwise, conditional independence implies
\[
  \Pr_\sigma[K_{\cA}=K_{\cB}\ne\bot\mid e]
   =\sum_{k\in\cK_\lambda}
       \Pr_\sigma[K_{\cA}=k\mid e]
       \Pr_\sigma[K_{\cB}=k\mid e]
   \le\Pr_\sigma[K_{\cB}=R(e)\ne\bot\mid e].
\]
Since both sides are zero on failure, averaging over Eve's view and using
the distance bound for the recovery and agreement events gives
\begin{align*}
  \Pr_\rho[R(E)=K_{\cB}\ne\bot]
  &\ge\Pr_\sigma[R(E)=K_{\cB}\ne\bot]-\Delta\\
  &\ge\Pr_\sigma[K_{\cA}=K_{\cB}\ne\bot]-\Delta
   \ge\alpha-2\Delta.
\end{align*}
To evaluate \(R(e)\), Eve computes Bob's conditional key
probabilities from the description of \(\sigma\) and her recorded
view \(e\), then chooses the first maximizing key, which requires
only classical computation and no further oracle queries.
This concludes the proof of \Cref{lem:qb-ca-recovery}.
\end{proof}

\subsubsection{Defining the Two-Message Attacker}

We now combine the two learners using the reference experiment in
\Cref{def:qb-ca-two-learners}, which specifies the conditional
distributions Eve uses to choose her queries. The experiment uses
the interleaved schedule \(\cB_1,L_1,\cA,L_2,\cB_2\), while
the real attacker evaluates the resulting rules after receiving the
complete transcript.

\begin{definition}[Reference experiment \(\mathsf D\)]
\label{def:qb-ca-two-learners}
Fix the parameters of \Cref{thm:qb-ca-two}. The following experiment
defines the two learners in order, stopping if either declares failure.
\begin{enumerate}[leftmargin=2em]
  \item \textbf{Bob sends \(M_1\), and Eve then runs \(L_1\).}
  Sample \(H\leftarrow\cH\), run \(\cB_1^H\), and give
  its message \(M_1\) to Eve. Run the learner of
  \Cref{lem:standalone-learning-switching} after this first message,
  with parameters \(\beta,\eta_1\), query limit
  \(N_1\), and answers from \(H\). At each current view,
  follow its fixed rule to choose the next query or declare success
  or failure.
  Denote its table by \(h_1\) and the resulting view by \(e_1\).

  \item \textbf{Alice sends \(M_2\) using an independent oracle.}
  Sample \(G\leftarrow\cH\) independently and run
  \(\cA^{G[h_1]}(1^\lambda,M_1)\). Let \(S\) be her query set outside
  \(\Dom(h_1)\), and let \(P_0=G|_S\) record those private
  answers. Give her message \(M_2\) to Eve.

  \item \textbf{Eve runs \(L_2\) before Bob resumes.}
  Use the learner of \Cref{lem:qb-ca-removal} for the conditional
  distribution just specified, with \(s=q_{\cA}\), threshold
  \(\gamma\), and limit \(N_2\). Its answers come from
  \(G[h_1]\). At a current view \(e\), it queries the first
  unrecorded input with \(\Pr_{\mathsf D}[x\in S\mid E=e]>\gamma\),
  if one exists. Let \(h_2\) be the new pairs it learns, so that
  Eve's final table is \(h_1\cup h_2\) and Alice's remaining
  private table is \(P=P_0|_{\cX\setminus\Dom(h_2)}\).
\end{enumerate}
Both learners test for success before checking their query limit, and
declare failure at a view of zero probability in this experiment.
\end{definition}

Each learner is defined from the part of the experiment preceding it
and its own queries so far, so its next action is determined by the
conditional distribution at its current view. In particular,
\(L_2\) computes the probability that an input belongs to \(S\)
without reading Alice's actual query set.

\paragraph{The real attack.}
Define \(\cE^H(1^\lambda,M_1,M_2)\) as follows. Initialize
\(M=(M_1)\) and \(J,F,h\) to be empty, and run the fixed rule
\(L_1\), answering every query with \(H\), then append \(M_2\)
to \(M\) and, if \(L_1\) succeeded, run \(L_2\) with the same
oracle. On either learner's failure, return \(\bot\), and otherwise
output \(K^*=R(E)\), where \(R\) is the recovery map of
\Cref{lem:qb-ca-recovery} for the output distribution of the final
hybrid \(\mathsf G_2\), defined below, selecting a nonabort key that
maximizes Bob's conditional key probability in that distribution with the lemma's
specified output on views of zero probability.
All these computations take place after the honest execution, using
only the public transcript, the protocol description, and Eve's own
classical query answers, and evaluating \(R\) requires no further
oracle queries. The reference experiment \(\mathsf D\) of
\Cref{def:qb-ca-two-learners} specifies the query-selection rules,
while the actual answers always come from \(H\).

\subsubsection{Proof of the Two-Message Bound}

\begin{proof}[Proof of \Cref{thm:qb-ca-two}]
We use the learner rules of \Cref{def:qb-ca-two-learners} throughout.

\paragraph{Fixed query rules.}
The reference experiment of \Cref{def:qb-ca-reference-experiment}
specifies the learners' conditional distributions by answering Alice
and a learner following her message from \(h,P\), with fresh
independent answers elsewhere, while answering Bob and a learner
following his message from \(H[h]\). The schedules and parameters are given in
\Cref{def:qb-ca-two-learners,def:qb-ca-round-decoupled}.
The rules are constructed in order, with each learner defined from the
experiment up to its current query and all earlier learner rules
already fixed. Eve computes the required conditional states and
probabilities by classically simulating this finite experiment and
summing over oracle choices and outcomes consistent with her view,
retaining classical descriptions of the simulated quantum states
without querying the actual oracle.

In the real attack, Eve applies these fixed rules to her actual view
and obtains every answer from \(H\), while in the hybrids we keep
the rules fixed and change only the oracle supplying the answers.
Each learner has an accuracy parameter, \(\beta\) or \(\gamma\),
and a hard query limit \(N\), and first tests its stopping condition
at every view of positive probability in its defining experiment.
If the condition holds, it succeeds, and otherwise it makes its
prescribed next query unless it has already made \(N\) queries,
in which case it fails. It also fails at views of zero probability
in the defining experiment, which contribute no failure probability there.

After a failure, Eve makes no further queries, skips later learners,
appends any remaining messages to \(M\), and outputs \(\bot\).
In the interleaved executions, the honest parties receive no failure
signal and continue their protocol, so only \(M\) continues to grow.

\paragraph{Running the learners between messages.}
To analyze \(\mathsf{KR}_{\Pi^{\mathsf{qb}},\cE}(1^\lambda)\), we
use the order
\[
  \cB_1,\quad L_1,\quad \cA,\quad L_2,\quad \cB_2.
\]
To see why this gives the same distribution as running Eve after she
receives the complete transcript, fix \(H\) and an honest execution.
Eve's queries only read \(H\), and the parties receive no information
from her, so running a learner earlier does not affect the messages or keys.
In either order, each learner uses the same message prefix and previous
query answers, so induction over its queries gives the same inputs,
answers, and stopping decisions, and both executions end with the same
complete message record, final view, and key guess. The same reasoning applies
to the constant-round schedule in \Cref{def:qb-ca-round-decoupled}.

\paragraph{Continuing with the reprogrammed oracle.}
Every hybrid outputs \((K_{\cA},K_{\cB},E)\).
After changing an honest stage and its following learner, the hybrids
let the remaining stages run with all parties and learners querying
\(H[h][P]\), as defined in \Cref{def:reprogrammed-oracle}.
They resume from their current registers and tables with the same
learner rules and failure handling, storing new Alice answers in
\(P\) and moving learner answers into \(h\), as in
\Cref{sec:learner-records}. Recording a fresh answer from \(H\),
or moving an answer from \(P\) to \(h\), leaves this function
unchanged, while an independently sampled answer, once recorded, is
used by all subsequent queries at its input. For two messages, the
remaining computation is Bob's last stage, while for constant rounds
we repeat this construction recursively at each changed stage.
Each hybrid below
specifies where its oracle is reprogrammed and which computation follows.

We first apply \Cref{lem:standalone-learning-switching} to
\(\cA+L_2\) together with Bob's continuation, and then apply
\Cref{lem:qb-ca-removal} to Bob's final stage, using the experiments
of \Cref{def:oracle-switching-experiments,lem:qb-ca-removal} to
identify each oracle reprogramming step.
Changes from the preceding hybrid are shown in \changed{color}.
Replacing an output key by \(\bot\) is an operation of the proof
experiment, since the real Eve only controls her own guess.

\begin{hybrid}{Hybrid \(\mathsf G_0\) (real execution with interleaved learners)}
  \item Sample \(H\leftarrow\cH\). Run \(\cB_1^H\),
  obtaining \(M_1,S_{\cB}\), and append \(M_1\) to Eve's view.
  Run \(L_1^H\), obtaining table \(h_1\) and recording whether
  it succeeded or failed.
  \item Run \((M_2,K_{\cA})\leftarrow\cA^H(1^\lambda,M_1)\)
  and append \(M_2\) to Eve's view. If \(L_1\) succeeded, run
  \(L_2^H\), obtaining table \(h_2\) and recording whether
  it succeeded or failed.
  \item Run \(K_{\cB}\leftarrow\cB_2^H(1^\lambda,S_{\cB},M_2)\).
  Output the honest keys and Eve's complete view, including on failure.
\end{hybrid}
By the equivalence of the two execution orders, the output of \(\mathsf G_0\) has
the same distribution as the honest keys and Eve's final record in
\(\mathsf{KR}_{\Pi^{\mathsf{qb}},\cE}(1^\lambda)\).
The honest marginal is \(\Real_{\Pi^{\mathsf{qb}}}(1^\lambda)\),
so \(\Pr_{\mathsf G_0}[K_{\cA}=K_{\cB}\ne\bot]=\alpha\).

\begin{hybrid}{Hybrid \(\mathsf G'_0\) (first-learner failure)}
  \item Execute \(\mathsf G_0\), retaining all its messages and
  learner records.
  \item \changed{If \(L_1\) failed, replace both recorded keys by
  \(\bot\).} Output \((K_{\cA},K_{\cB},E)\).
\end{hybrid}
The common prefix \(\cB_1,L_1\) follows the rules of the reference
experiment in \Cref{def:qb-ca-two-learners}. By the failure guarantee
of \Cref{lem:standalone-learning-switching},
\[
  D(\mathsf G_0,\mathsf G'_0)
   \le\Pr[L_1\text{ fails}]\le\eta_1.
\]
Denote the successful terminal view of \(L_1\) by \(e_1\).
Since Alice has not yet run \(\cA\), her private work registers
are still in the all-zero state, and the message \(M_1\) she
received is fixed by \(e_1\), so her registers are independent
of \(BH\) conditioned on that view.

\begin{hybrid}{Hybrid \(\mathsf G_1\) (independent answers)}
  \item Execute \(\cB_1^H\) followed by \(L_1^H\) as before.
  On \(L_1\)-failure,
  keep the entire remaining \(\mathsf G'_0\) execution unchanged.
  \item On success, \changed{sample an independent uniform \(G\)
  and run \(\cA^{G[h_1]}\)}, obtaining \(M_2,K_{\cA}\) and
  her query table \(T_{\cA}\) outside \(\Dom(h_1)\). 
  Append \(M_2\) to Eve's view.
  \item Run the same \(L_2\) rule with answers from
  \changed{\(G[h_1]\)}, obtaining \(h_2\) and recording whether
  it succeeded or failed.
  Let \(T=T_{\cA}\cup h_2\), including inputs queried only by
  the learner.
  \item Run \(\cB_2\) with oracle \changed{\(H[h_1][T]\)}
  on \(S_{\cB},M_2\). Output the resulting keys and \(E\),
  retaining the keys even if \(L_2\) failed.
\end{hybrid}
Conditioned on a successful \(e_1\), the hybrids \(\mathsf G'_0\)
and \(\mathsf G_1\) run \(\mathsf{SharedAnswers}\) and
\(\mathsf{IndependentAnswers}\), respectively, from
\Cref{def:oracle-switching-experiments}, with
\(\mathcal C=\cA+L_2\) and continuation \(\mathcal T=\cB_2\).
The procedure \(\mathcal C\) makes at most \(q_{\cA}+N_2\)
classical queries, and using independent answers at
fresh inputs is equivalent to querying \(G[h_1]\). Bob then continues
with the reprogrammed oracle, which retains every answer from this
procedure as required by the lemma. The fixed rule \(L_2\) is the same
in both experiments,
even though its probabilities are defined using \(\mathsf D\).
Applying \Cref{lem:standalone-learning-switching} and averaging over
the common distribution of successful \(e_1\) gives
\[
  D(\mathsf G'_0,\mathsf G_1)
   \le2(q_{\cA}+N_2)\sqrt\beta.
\]
The bound includes executions in which \(L_2\) fails, requires
no assumption about Alice's query probabilities, and is independent
of Bob's final query count because the lemma allows any number of
queries in the common continuation.

Let \(S=\Dom(T_{\cA})\) and let \(P\) be the remaining private
table after \(L_2\). On every branch where \(L_1\) succeeded,
\begin{equation}
  \begin{gathered}
    P=G|_{S\setminus\Dom(h_2)},\qquad T=h_2\cup P,\qquad
    \Dom(h_2)\cap\Dom(P)=\varnothing,\\
    H[h_1][T]=H[h_1][h_2][P].
  \end{gathered}
  \label{eq:qb-ca-two-full-table}
\end{equation}
These identities hold even when \(L_2\) fails. The execution through \(L_2\) is exactly
the prefix specified by \(\mathsf D\) in \Cref{def:qb-ca-two-learners}.

\begin{hybrid}{Hybrid \(\mathsf G'_1\) (second-learner failure)}
  \item Execute \(\mathsf G_1\), retaining its full message and
  query records and the keys \(\bot,\bot\) on \(L_1\)-failure.
  \item \changed{If \(L_1\) succeeded and \(L_2\) failed,
  replace both recorded keys by \(\bot\).}
  Output \((K_{\cA},K_{\cB},E)\).
\end{hybrid}
Conditioned on each successful \(e_1\) and each reachable Alice
message, the learner \(L_2\) now has the distribution specified
by the reference experiment in \Cref{def:qb-ca-two-learners}.
Its initial set \(S\) has size at most
\(q_{\cA}\). By \Cref{lem:qb-ca-removal},
\[
  D(\mathsf G_1,\mathsf G'_1)
   \le\Pr[L_1\text{ succeeds},L_2\text{ fails}]\le\eta_2.
\]
For each view \(e\) on which both learners succeeded, the
conditional-independence guarantee of
\Cref{lem:standalone-learning-switching} gives the product of \(AP\)
and \(BH\) immediately before Bob resumes. Moreover,
\Cref{lem:qb-ca-removal} now gives its stated bound on removing
private answers from his oracle, which we apply in the next hybrid.

\begin{hybrid}{Hybrid \(\mathsf G_2\) (removing private answers)}
  \item Execute the prefix through \(L_2\) as in \(\mathsf G'_1\).
  Keep every failed branch's entire continuation unchanged, including
  its transcript and its recorded keys \(\bot,\bot\).
  \item If both learners succeeded, run Bob's last stage with
  \changed{\(H[h_1][h_2]\)} in place of \(H[h_1][h_2][P]\).
  Retain Alice's state and \(P\). Output the resulting keys and \(E\).
\end{hybrid}
For each view at which both learners succeeded, Bob's last stage in
\(\mathsf G'_1\) and \(\mathsf G_2\) follows
\(\mathsf{UsePrivate}\) and \(\mathsf{IgnorePrivate}\),
respectively, from \Cref{lem:qb-ca-removal}, with
\(h=h_1\cup h_2\), \(b=q_{\cB,2}\), and \(c=0\).
Both learned tables remain in the oracle, so averaging the lemma's
bound over successful views gives
\[
  D(\mathsf G'_1,\mathsf G_2)\le2q_{\cB,2}\sqrt\gamma.
\]

Let \(\rho\) and \(\sigma\) be the output distributions of
\(\mathsf G_0\) and \(\mathsf G_2\). Adding the four bounds on
the distance between successive hybrids gives
\[
  D(\rho,\sigma)
  \le\eta_1+2(q_{\cA}+N_2)\sqrt\beta
       +\eta_2+2q_{\cB,2}\sqrt\gamma=\Delta_2.
\]
The distance bound includes Eve's view \(E\) and is obtained by
averaging each conditional bound over the common distribution of views,
so no error term is divided by the probability of a particular view.
We applied each failure bound only after the corresponding learner had
the distribution specified by \(\mathsf D\).

In \(\mathsf G_2\), Bob's last stage acts on \(BH\) using only
the learned tables fixed by \(e\), so the two keys are independent
conditioned on every successful view, while on failure both recorded
keys are \(\bot\), so \(\sigma\) meets the hypotheses of
\Cref{lem:qb-ca-recovery}. Since \(\rho\) is also the joint
distribution of the honest keys and Eve's final record in the
key-recovery experiment, her output \(K^*=R(E)\) satisfies
\[
  \Adv^{\mathsf{KR}}_{\Pi^{\mathsf{qb}},\cE}(\lambda)
    =\Pr_\rho[R(E)=K_{\cB}\ne\bot]
    \ge\alpha-2\Delta_2.
\]
Only \(L_1,L_2\) query the actual oracle, for at most \(N_1+N_2\)
classical queries, while computing their rules and the recovery map
requires only classical computation and no further oracle access,
so \(\cE\) is an attacker of the type specified in
\Cref{def:qb-ca-two}.
This concludes the proof of \Cref{thm:qb-ca-two}.
\end{proof}

\begin{corollary}[Polynomial-query two-message attack]
\label{cor:qb-ca-two}
For \(0<\xi<1\), there is an attacker as in
\Cref{thm:qb-ca-two} with
\(\Adv^{\mathsf{KR}}_{\Pi^{\mathsf{qb}},\cE}(\lambda)
\ge\alpha-\xi\) and query complexity
\[
  O\!\left(
    \frac{q_{\cB,1}(q_{\cA}+1)^2Q_2^4}{\xi^9}
    +\frac{q_{\cA}Q_2^2}{\xi^3}+1\right),
  \qquad Q_2:=\max\{1,q_{\cB,2}\}.
\]
Consequently, a family with polynomial honest query bounds and
\(\alpha_{\Pi^{\mathsf{qb}}}(\lambda)\ge1/\poly(\lambda)\)
is not passive key-recovery secure against classical attackers, and
hence is not secure against quantum attackers either. In particular,
this applies to every correct family with polynomial honest query bounds.
\end{corollary}

\begin{proof}
For any \(0<\xi<1\), choose, in this order,
\[
  \eta_1=\eta_2=\xi/8,\qquad
  \gamma=\left(\frac{\xi}{16Q_2}\right)^2,\qquad
  N_2=\ceil{\frac{2048q_{\cA}Q_2^2}{\xi^3}},
\]
and then set
\[
  n_*:=\max\{1,q_{\cA}+N_2\},\qquad
  \beta=\left(\frac{\xi}{16n_*}\right)^2,\qquad
  N_1=\ceil{\frac{2048q_{\cB,1}n_*^2}{\xi^3}}.
\]
The two terms in \(\Delta_2\) that account for changes to oracle
answers are each at most
\(\xi/8\), so \(2\Delta_2\le\xi\). Substituting the bound on
\(N_2\) into \(n_*\) proves the query estimate. For fixed
\(\beta,\eta_1\), the first learner's query bound involves only
\(q_{\cB,1}\), while Bob's final query count enters when choosing
thresholds to meet a specified total error.
For the security conclusion, take \(\xi=\alpha/3\), giving an
attacker with advantage at least \(2\alpha/3\) and query complexity
polynomial in the honest query bounds and \(1/\alpha\), which
violates the security condition in \Cref{def:qb-ca-two} whenever
\(\alpha\ge1/\poly(\lambda)\). Correctness implies
\(\alpha\ge1-\epsilon_c(\lambda)\) for negligible
\(\epsilon_c\), so it is a special case.
This concludes the proof of \Cref{cor:qb-ca-two}.
\end{proof}

\subsection{Constant-Round Protocols}
\label{sec:qb-ca-constant}

We extend the attack by processing the complete transcript in order
and repeating the two learning steps for each Alice stage, placing
the learner of \Cref{lem:general-learning-switching} before that
stage and the learner of \Cref{lem:qb-ca-removal} after it in the
interleaved analysis.
Alice's private table \(P\) is
retained across rounds, so the learner after her \(i\)-th stage
must account for all inputs still private from her first \(i\)
stages.

The second oracle reprogramming step now includes Bob's next stage and the learner
following his message, both of which must use \(H[h]\) in place
of \(H[h][P]\) for this part of the execution to agree with the
reference experiment in \Cref{def:qb-ca-round-decoupled}.
A classical learner query can itself encounter a
private input, so \Cref{lem:qb-ca-removal} contributes an error
of \(c\gamma\) for these queries, in addition to
\(2b\sqrt\gamma\) for Bob's quantum queries. The number of
queries made by the next learner therefore affects the error bound
for the current stage, a dependence we account for when choosing
the thresholds.

\begin{definition}[Alternating protocol with quantum Bob]
\label{def:qb-ca-rounds}
An alternating protocol \(\Pi^{\mathsf{qb}}\) consists of oracle
quantum algorithms with stage order
\[
  \cB_1,\cA_1,\cB_2,\ldots,\cA_r,\cB_{r+1},
  \qquad r\ge1.
\]
Each stage acts on its party's retained register and the honest messages
sent so far, retains an updated private register, and sends at most one
classical message, with a fixed empty message representing a stage
that sends none. All \(\cA_i\) have classical oracle access, and
all \(\cB_i\) may have quantum oracle access. As in
\Cref{def:qb-ca-two}, the parties start with all-zero private work
registers and no shared state or correlated randomness, with local
state preparation and generation of private randomness included in
their stages. The honest execution uses one uniform \(H\leftarrow\cH\)
throughout.

After these stages, the parties apply local, oracle-free output maps
to their retained registers and the honest-message transcript, obtaining
\(K_{\cA},K_{\cB}\in\cK_\lambda\cup\{\bot\}\).
A key produced earlier can be retained until this point, and any
final oracle queries are included in the corresponding party's last
stage. Empty initial or final stages allow either party to speak first
or last, so \(r\) counts Alice's stages and the schedule contains
\(2r+1\) stages in all, some possibly empty.
Denote this honest experiment by
\(\Real_{\Pi^{\mathsf{qb}}}(1^\lambda)\).

Let \(a_i\) bound the queries in \(\cA_i\), and let \(b_i\)
bound the queries in \(\cB_i\). Fix an integer
\[
  Q\ge\max\left\{1,
       \sum_{i=1}^r a_i+\sum_{i=1}^{r+1} b_i\right\}.
\]
Define
\[
  \alpha:=\alpha_{\Pi^{\mathsf{qb}}}(\lambda)
    :=\Pr_{\Real_{\Pi^{\mathsf{qb}}}(1^\lambda)}
        [K_{\cA}=K_{\cB}\ne\bot].
\]
The completeness and security notions of \Cref{def:qb-ca-two}
extend to this protocol by giving the complete classical transcript
\(M_{\mathrm{full}}\) to the attacker after the honest experiment. Thus
\(\mathsf{KR}_{\Pi^{\mathsf{qb}},\cE}(1^\lambda)\) runs
\(K^*\leftarrow\cE^H(1^\lambda,M_{\mathrm{full}})\) and outputs \(1\) if and
only if \(K^*=K_{\cB}\ne\bot\). We use the same notation
\(\Adv^{\mathsf{KR}}_{\Pi^{\mathsf{qb}},\cE}(\lambda)\)
for its success probability. In particular, correctness means
\(\alpha\ge1-\epsilon_c(\lambda)\) for negligible
\(\epsilon_c\), and security requires negligible key-recovery
advantage for every polynomially query-bounded attacker of the
specified type.
\end{definition}

\begin{theorem}[Constant-round attack with quantum Bob]
\label{thm:qb-ca-rounds}
Fix an integer \(r\ge1\) independently of the security parameter
\(\lambda\), and let \(\Pi^{\mathsf{qb}}\) be a protocol in
\Cref{def:qb-ca-rounds} with \(r\) Alice stages. For parameters
\(\eta,\beta_i,\gamma_i\in(0,1)\), set
\begin{equation}
  \begin{gathered}
    N_i^-:=\ceil{\frac{Q}{\eta\beta_i}},\qquad
    N_i^+:=\ceil{\frac{Q}{\eta\gamma_i}}
       \quad(1\le i\le r),\\
    N_{r+1}^-:=0.
  \end{gathered}
  \label{eq:qb-ca-round-caps}
\end{equation}
There is a computationally unbounded passive classical attacker \(\cE\) making
at most \(\sum_{i=1}^r(N_i^-+N_i^+)\) classical oracle queries
and satisfying
\begin{equation}
  \begin{gathered}
    \Adv^{\mathsf{KR}}_{\Pi^{\mathsf{qb}},\cE}(\lambda)
      \ge\alpha-2\Delta_r,\\
    \Delta_r:=2r\eta+
      \sum_{i=1}^r\left(
        2(a_i+N_i^+)\sqrt{\beta_i}
        +2b_{i+1}\sqrt{\gamma_i}
        +N_{i+1}^-\gamma_i\right).
  \end{gathered}
  \label{eq:qb-ca-round-loss}
\end{equation}
The attacker runs after receiving the complete transcript, makes all
its oracle queries through the specified learners, and then computes
its key guess from their recorded view without further queries. For a
protocol with completeness error \(\epsilon_c(\lambda)\), the
advantage is at least \(1-\epsilon_c(\lambda)-2\Delta_r\).
\end{theorem}

\subsubsection{Learning after Earlier Rounds}

Before a later Alice stage, both parties may retain information from
earlier rounds and the tables \(h,P\) may already be nonempty, so we
need the following extension of \Cref{lem:standalone-learning-switching}.
The preparation procedure \(\mathcal G\) represents the prefix of the
reference experiment through the current Bob stage, including the
learners already defined. The assumptions below ensure that this prefix
has the two properties used in the earlier proof: a bound on the
remaining Fourier information about \(H\), and independence of
\(AP\) from \(BH\) conditioned on Eve's view.

\begin{lemma}[Classical learning after earlier rounds]
\label{lem:general-learning-switching} 
Fix a finite preparation procedure \(\mathcal G\) of the following
form, with access to a uniformly sampled function \(H:\cX\to\cY\),
where \(\cX\) is a finite ordered set and
\(\cY=\{0,1\}^{\ell}\). Its description, including any learners
it runs, is fixed independently of \(H\).
\begin{enumerate}[leftmargin=2em,itemsep=0.3em]
  \item \textbf{Initialization and oracle access.}
  The procedure starts with zero work registers and an empty table
  \(h\), and accesses \(H\) only through \(H[h]\), which uses
  recorded answers on \(\Dom(h)\) and \(H\) elsewhere.
  The bound \(d\) applies to the total number of quantum queries to
  \(H[h]\) and classical queries to \(H[h]\) whose query--answer
  pairs are not immediately recorded in \(E\) and \(h\).
  The procedure may make additional classical queries to \(H[h]\)
  that do not count toward \(d\), provided that their query--answer
  pairs are immediately recorded in \(E\) and \(h\). 
  \item \textbf{Recorded answers.}
  Every additional classical query to \(H[h]\) is immediately
  recorded, together with its answer, in a classical record \(E\)
  and in \(h\). Queries answered by local tables or local randomness
  without accessing \(H\), such as Alice's private queries in the
  reference experiment, are not subject to this requirement.  
  The procedure may also insert an arbitrary pair \((x,y)\) into
  \(h\), provided that \(x\notin\Dom(h)\) and the insertion
  is recorded in \(E\), but never changes or deletes an existing
  entry of \(h\), which is determined by \(E\). 
  \item \textbf{Output.}
  The procedure outputs quantum registers \(A,B\), the record \(E\),
  and classical tables \(h,P\) with disjoint domains, such that
  \(A,P\) are jointly independent of \(B,H\) conditioned on
  every output record \(e_0\) of positive probability:
  \[
    \rho^{\mathcal G,e_0}_{APBH}
      =\rho^{\mathcal G,e_0}_{AP}\otimes\rho^{\mathcal G,e_0}_{BH}.
  \]
  Here \(H\) denotes a classical register storing the originally
  sampled function. Subsequent procedures have the specified oracle
  access to this same \(H\), with the recorded reprogrammings, but
  do not receive the register containing its full truth table. 
\end{enumerate}

\textbf{Existence of a learner.}
Then, for every \(\beta,\eta\in(0,1)\), there exists a deterministic
classical learner \(L\) that runs after \(\mathcal G\), may use
unbounded computation, and has the following guarantees:
\begin{enumerate}[leftmargin=2em,itemsep=0.3em]
  \item \textbf{Query bound and success probability.}
  It makes at most
  \[
    N=\ceil{\frac{d}{\eta\beta}}
  \]
  queries to \(H[h]\), always at inputs not yet recorded in \(h\),
  and fails with probability at most \(\eta\) conditioned on every
  view \(E=e_0\) of positive probability immediately after
  \(\mathcal G\), before \(L\) begins. 
  \item \textbf{Classical query rule.}
  Its next query or decision to declare \(\mathsf{success}\) or
  \(\mathsf{failure}\) is computable from
  \(\mathcal G,\beta,\eta,E\) alone, without further oracle queries
  or access to \(A,B\). At a record of zero probability in this
  experiment, it declares failure.
  \item \textbf{Record updates.}
  It appends each query--answer pair to \(E\), stores it in \(h\),
  and deletes any entry at that input from \(P\), without acting
  on \(A\) or \(B\), and records its termination flag in \(E\).
  \item \textbf{Oracle switching after success.}
  Run \(\mathcal G\) followed by \(L\), and condition on any
  successful terminal view \(E=e\) of positive probability. For
  every \(\mathcal C,\mathcal T\) as in
  \Cref{def:oracle-switching-experiments}, initialize both experiments
  with the resulting conditional joint state of \(A,P,B,H\), the
  updated learned table \(h=h(e)\), and the record \(E=e\).  

  Their final states satisfy
  \[
    D\bigl(\rho_{\mathrm{out}}^{\mathsf{SharedAnswers}},
           \rho_{\mathrm{out}}^{\mathsf{IndependentAnswers}}\bigr)
      \le 2n\sqrt\beta.
  \]
  Here \(n\) bounds the classical queries made by \(\mathcal C\),
  and the bound holds for any finite number of queries by \(\mathcal T\).
  \item \textbf{Conditional independence after success.}
  Under the same conditions, \(AP\) and \(BH\) remain independent
  in \(\mathsf{IndependentAnswers}\), conditioned on every resulting
  record of positive probability after \(\mathcal C\), before
  \(\mathcal T\) begins. 
\end{enumerate}
\end{lemma}

\begin{proof}[Proof of \Cref{lem:general-learning-switching}]
We use the Fourier representation, reset channel \(\mathcal R_x\),
and learner construction from the proof of
\Cref{lem:standalone-learning-switching}. We first verify that the
preparation conditions supply the same Fourier-weight and independence
properties, now with \(d\) in place of \(q_{\cB,1}\) and with
possibly nonempty tables.

Represent \(\mathcal G\) by resetting \(\mathsf H_x\) after
each entry added to \(h\). Initially, every oracle coordinate is in
state \(\ket{\widehat0}\), independently of the all-zero work
registers. Each quantum query to \(H[h]\) adds at most one nonzero
Fourier label. An additional classical query at an unrecorded input
reads one oracle coordinate and immediately records its answer in
\(h\), allowing that coordinate to be reset without increasing the
number of nonzero labels at the remaining coordinates, even after
conditioning on the answer. Since local computation and the addition
of locally generated table entries do not increase that number either,
the state after
\(\mathcal G\) is supported on Fourier basis states with at most
\(d\) nonzero labels, even conditioned on any positive-probability
record \(e_0\). 

Fix such a starting record \(e_0\), and let
\(\sigma^{e_0}_{B\mathsf H}\) be the marginal of this
representation on \(B\mathsf H\). We construct the auxiliary state
\[
  \tau^{e_0}_{APB\mathsf H}
    :=\rho^{\mathcal G,e_0}_{AP}\otimes
      \sigma^{e_0}_{B\mathsf H}.
\]
This state is constructed for the proof, rather than being an abbreviation
for the output of \(\mathcal G\): its oracle register \(\mathsf H\)
is quantum, whereas the original experiment retains a classical
register \(H\). 
The conditional product assumption in item~3 for \(\mathcal G\)
allows us, in the ordinary experiment, to discard \(AP\) and
independently prepare \(\rho^{\mathcal G,e_0}_{AP}\) without changing
its joint state with \(BH\). 
To see why the auxiliary state represents the subsequent computation,
imagine measuring the unrecorded oracle coordinates in the computational
basis before they are used. This produces a classical function with
its corresponding conditional state of \(B\), as in the ordinary
experiment. Such measurements can instead be deferred until the oracle
registers are discarded, because queries read their computational-basis
values without changing them, and any coordinate that is reset is never
used again. Consequently, coherence between different function values
cannot affect the retained output registers, while queries at inputs
in \(\Dom(h)\) simply use the fixed table values. The auxiliary
state therefore gives the same learner observations and output states,
while taking a marginal and preparing the independent \(AP\) factor
preserve the bound of \(d\) nonzero Fourier labels. 

Starting from this auxiliary state, define the conditional states
\(\tau^e_{APB\mathsf H}\) and their marginals by
\eqref{eq:qb-ca-auxiliary-conditioning}. The learner uses the same
weights \(w_e(x)=\Tr(\Pi_x\tau^e_{B\mathsf H})\) and
queries the first unrecorded input with weight greater than \(\beta\),
now simulating \(\mathcal G\) in place of \(\cB_1\) and stopping
after at most \(N=\ceil{d/(\eta\beta)}\) queries. It records each
answer in \(E,h\), deletes the entry at that input from \(P\),
and applies \(\mathcal R_x\) in the proof representation.
These operations preserve the conditional product of \(AP\) and
\(B\mathsf H\), since the observation comes from \(B\mathsf H\)
and fixes only a local deletion in \(P\).

The initial total Fourier weight is at most \(d\) for every
starting view \(e_0\), so the same application of Austrin
et al.~\cite[Lemma~4.12]{AustrinEtAl22}, now with this value of
\(d\), gives
\(\mathbb E[N_{\mathrm{full}}\mid E_{\mathrm{start}}=e_0]
\le d/\beta\). Markov's inequality then bounds the probability of
\(N_{\mathrm{full}}>N\) by \(\eta\), with immediate success
when \(d=0\). All decisions are computed from
\(\mathcal G,\beta,\eta,E\), as required.

After success, the auxiliary state satisfies the product and weight
bounds in \eqref{eq:qb-ca-independent-answers-hyp}. The query-by-query
argument in the proof of \Cref{lem:standalone-learning-switching}
therefore applies to any \(\mathcal C\) in
\Cref{def:oracle-switching-experiments}: existing entries of \(h,P\)
are reused in both experiments, and each of its at most \(n\) fresh
queries contributes at most \(2\sqrt\beta\) to the distance.
That argument bounds the joint state including the oracle coordinates
used for subsequent queries, so applying any common continuation
\(\mathcal T\) with \(H[h'][P']\) and tracing out the proof
registers gives the stated \(2n\sqrt\beta\) bound.

Finally, in the ordinary experiment, the learner preserves the
conditional product of \(AP\) and \(BH\), because its observations
come from \(BH\) and determine only local deletions in \(P\).
During \(\mathsf{IndependentAnswers}\), all new observations of
\(\mathcal C\) come from \(AP\) and independent randomness,
while \(BH\) is unchanged, which proves the required independence
before \(\mathcal T\). This concludes the proof of
\Cref{lem:general-learning-switching}.
\end{proof}

In the reference experiment of \Cref{def:qb-ca-reference-experiment},
a prefix ending after a Bob stage meets these preparation conditions.
The conditional product is \eqref{eq:qb-ca-reference-product}, and
we take \(d\) to bound Bob's total queries, including his classical queries.
Every other access to \(H\) is a classical learner query whose
answer is immediately recorded in \(h\), while Alice and the
learners following her use the local answer rules without querying
\(H\). The table updates therefore have exactly the form required
by the lemma.

\subsubsection{The Reference Experiment for Alternating Stages}

For each \(i\), the superscript in \(L_i^-\) means
\emph{before Alice's stage} \(\cA_i\), and the superscript in
\(L_i^+\) means \emph{after that stage} in the execution schedule
used in the proof.
Thus the interleaved order is
\[
  \cB_1,\ L_1^-,\ \cA_1,\ L_1^+,\ \cB_2,\
  L_2^-,\ \cA_2,\ L_2^+,\ \ldots,\
  L_r^-,\ \cA_r,\ L_r^+,\ \cB_{r+1}.
\]
There is no learner after \(\cB_{r+1}\), which explains the convention
\(N_{r+1}^-=0\).

\begin{definition}[Reference experiment \(\mathsf D_r\)]
\label{def:qb-ca-round-decoupled}
Initialize the parties as in \Cref{def:qb-ca-rounds}, sample
\(H\leftarrow\cH\), and set \(h=P=\varnothing\).
Run the displayed interleaving of honest stages and learners,
updating Eve's view after each message, query answer, and learner
termination and answering oracle queries as follows:
\begin{enumerate}[leftmargin=2em]
  \item \textbf{During \(\cA_i\) and \(L_i^+\):}
  return existing entries of \(h\) or \(P\). At any other
  input sample an independent uniform answer. Store each new Alice
  answer in \(P\). Each learner answer is added to \(h\),
  deleting any private entry at that input.
  \item \textbf{During \(\cB_{i+1}\) and \(L_{i+1}^-\):}
  answer from \(H[h]\), ignoring \(P\). Bob may query
  quantumly. Each learner answer is added to \(h\), deleting any
  private entry at that input. This rule also applies to the initial
  \(\cB_1,L_1^-\).
\end{enumerate}
For each \(i=1,\ldots,r\), define the learners as follows:
\begin{itemize}[leftmargin=2em,label=\(\bullet\)]
  \item \(L_i^-\), after \(\cB_i\) and before \(\cA_i\),
  is the learner of \Cref{lem:general-learning-switching} for this
  prefix of the experiment, with \(d=Q\), parameters \(\beta_i,\eta\), and
  limit \(N_i^-\). The prefix through \(\cB_i\) has
  already been defined when this rule is constructed.
  \item \(L_i^+\), after \(\cA_i\) and before \(\cB_{i+1}\),
  is the classical learner of \Cref{lem:qb-ca-removal}, with
  \(s=Q\), threshold \(\gamma_i\), and limit \(N_i^+\).
  Its initial private table is the current \(P\). Equivalently,
  set \(S_i=\Dom(P)\) after \(\cA_i\) and query unrecorded
  inputs with \(\Pr_{\mathsf D_r}[x\in S_i\mid E=e]>\gamma_i\).
  Write \(P_i\) for the table remaining at termination.
\end{itemize}
Each learner tests for success before checking its query limit and
fails at a view of zero probability in this experiment. After the first
failure, there are no further learner queries, but the honest stages
continue with queries answered as above, and the experiment records
both final keys as \(\bot\).
\end{definition}

As for \(\mathsf D\) in \Cref{def:qb-ca-two-learners}, each rule
is defined from the experiment
specified up to that point.

\paragraph{The real attack.}
The attacker \(\cE^H(1^\lambda,M_{\mathrm{full}})\) receives the
complete transcript \(M_{\mathrm{full}}\) and initializes
\(M,J,F,h\) to be empty. It processes the messages of
\(M_{\mathrm{full}}\) in order: for each \(i\), it appends the
message from \(\cB_i\) to \(M\) and runs \(L_i^-\), then appends
the message from \(\cA_i\) to \(M\) and runs \(L_i^+\). Each learner uses
only the current working view, and every query is answered by the
actual \(H\). After a failure, it skips all remaining learners,
records the remaining messages, and returns \(\bot\), while if all
learners succeed, it appends the final message from \(\cB_{r+1}\)
and outputs \(K^*=R(E)\), where \(R\) is the map of
\Cref{lem:qb-ca-recovery} for the final hybrid
\(\mathsf J_{r,4}\), defined in the proof below, and evaluating it
requires no further oracle queries.

In \(\mathsf D_r\) from \Cref{def:qb-ca-round-decoupled}, the set
\(S_i\) includes all Alice inputs still private after her
first \(i\) stages, so \(|S_i|\le\sum_{j=1}^i a_j\le Q\).
It is fixed during \(L_i^+\), although \(P\) loses entries as
answers are added to \(h\). At an input outside the current \(\Dom(h)\),
membership in \(S_i\) is equivalent to membership in \(\Dom(P)\).

\begin{lemma}[Conditional independence and learner guarantees]
\label{lem:qb-ca-round-invariant}
At the end of \(L_i^-\) or \(L_i^+\), every positive-probability
view \(e\) in \(\mathsf D_r\) from \Cref{def:qb-ca-round-decoupled}
on which all learners completed so far have succeeded satisfies
\begin{equation}
  \rho^{\mathsf D_r,e}_{APBH}
   =\rho^{\mathsf D_r,e}_{AP}\otimes\rho^{\mathsf D_r,e}_{BH}.
  \label{eq:qb-ca-round-product}
\end{equation}
After \(L_i^-\), consider \(\mathsf{SharedAnswers}\) and
\(\mathsf{IndependentAnswers}\) from
\Cref{def:oracle-switching-experiments} with
\(\mathcal C=\cA_i+L_i^+\), followed by any continuation with
the reprogrammed oracle. Their output states have distance at most
\[
  2(a_i+N_i^+)\sqrt{\beta_i}.
\]
After \(L_i^+\), consider \(\mathsf{UsePrivate}\) and
\(\mathsf{IgnorePrivate}\) from \Cref{lem:qb-ca-removal} for
\(\cB_{i+1}\) and its following learner \(L_{i+1}^-\), with
no following learner when \(i=r\). Starting from the conditional
joint state at \(e\), their output states satisfy
\[
  D\bigl(\rho^{\mathsf{UsePrivate}},
         \rho^{\mathsf{IgnorePrivate}}\bigr)
    \le2b_{i+1}\sqrt{\gamma_i}+N_{i+1}^-\gamma_i.
\]
The same bound holds after any common continuation on the retained
registers. 
Each learner fails with probability at most \(\eta\), conditioned
on any reachable successful view at which it starts in \(\mathsf D_r\).
\end{lemma}

\begin{proof}
The experiment \(\mathsf D_r\) in \Cref{def:qb-ca-round-decoupled}
uses the answer sources and table updates of
\Cref{def:qb-ca-reference-experiment}.
The product identity holds initially because the parties' work registers
are all zero, \(P\) is empty, and \(H\) is sampled independently.
The first Bob stage and \(L_1^-\) act only on \(BH\), leaving
Alice's initially fixed registers independent of that factor.
Across \(\cA_i+L_i^+\), the conditional-independence guarantee
of \Cref{lem:general-learning-switching} preserves this product.
Across \(\cB_{i+1}\) and its following learner, the
\(\mathsf{IgnorePrivate}\) guarantee of \Cref{lem:qb-ca-removal}
preserves it, giving the induction at both kinds of learner endpoints.

For the learner before Alice's stage, Bob has made at most
\(\sum_{j=1}^{i}b_j\le Q\) quantum queries in total, so applying
\Cref{lem:general-learning-switching} with \(d=Q\),
\(\beta=\beta_i\), and failure parameter \(\eta\) gives its
query limit, failure bound, and the displayed trace-distance bound
for the next Alice stage and learner, which together make at most
\(a_i+N_i^+\) classical queries even if the learner later fails.

For the learner after Alice's stage, the table \(P\) contains only
inputs queried during Alice's first \(i\) stages, so
\(|\Dom(P)|\le\sum_{j=1}^i a_j\le Q\).
Applying \Cref{lem:qb-ca-removal} with \(s=Q\),
\(\gamma=\gamma_i\), and failure parameter \(\eta\), with queries
answered exactly as in \(\mathsf D_r\) during \(L_i^+\),
therefore gives the failure bound and, with \(b=b_{i+1}\) and
\(c=N_{i+1}^-\), the stated distance bound for the next Bob phase.
Both lemma applications
are conditioned on the learner's starting view in the reference
experiment of \Cref{def:qb-ca-round-decoupled}, and their failure
guarantees apply when queries are answered as specified there.
This concludes the proof of \Cref{lem:qb-ca-round-invariant}.
\end{proof}

\subsubsection{Proof of the Constant-Round Attack}
We first prove the attack bound in \Cref{thm:qb-ca-rounds}, keeping
its dependence on the number of Alice stages \(r\) explicit, and then
choose parameters in the next subsection to obtain a polynomial-query
attack for constant \(r\). 
\begin{proof}[Proof of \Cref{thm:qb-ca-rounds}]
Use the fixed rules of \Cref{def:qb-ca-round-decoupled} in every
hybrid. As in the two-message proof, the equivalence of the two execution orders lets
us analyze the attacker given the complete transcript in
\(\mathsf{KR}_{\Pi^{\mathsf{qb}},\cE}(1^\lambda)\) by placing
its learners between the honest stages. All hybrids output
\((K_{\cA},K_{\cB},E)\).
At each step, on executions in which the learners so far have succeeded,
the stages already changed answer queries as in \(\mathsf D_r\),
while in the remaining stages all parties and learners query the
reprogrammed oracle \(H[h][P]\).

When we apply a learner's failure bound, we introduce a hybrid that
sets both recorded output keys to \(\bot\) if that learner fails,
keeping the rest of that execution unchanged in all later hybrids,
including its later messages and Eve's recorded view. These replacements
of honest keys are for the analysis, since the real Eve only sets her
own guess to \(\bot\).

\begin{hybrid}{Hybrid \(\mathsf J_0\) (real execution with interleaved learners)}
  \item Sample a uniform \(H\), and run the honest stages with the
  learners interleaved in the order of \Cref{def:qb-ca-round-decoupled}.
  Answer every
  query with \(H\), and maintain the records \(E,h,P,T_{\cA}\).
  \item After the first learner failure, make no further learner
  queries, but complete the honest execution and retain its keys.
  Output \((K_{\cA},K_{\cB},E)\).
\end{hybrid}
All recorded values agree with \(H\), so using \(H[h][P]\)
throughout would give the same execution. The output distribution
is also that of the honest keys and Eve's final record in the
key-recovery experiment. In particular, the honest marginal is
\(\Real_{\Pi^{\mathsf{qb}}}(1^\lambda)\), and the valid
agreement probability in \(\mathsf J_0\) is \(\alpha\).

\begin{hybrid}{Hybrid \(\mathsf J_{1,0}\) (initial learner failure)}
  \item Run \(\mathsf J_0\).
  \item \changed{If \(L_1^-\) failed, replace both recorded keys
  by \(\bot\).} Retain the entire transcript and output as above.
\end{hybrid}
Since \(P\) is still empty, the real execution and
\(\mathsf D_r\) from \Cref{def:qb-ca-round-decoupled} have
identical prefixes through \(\cB_1,L_1^-\), so
\Cref{lem:qb-ca-round-invariant} gives
\[
  D(\mathsf J_0,\mathsf J_{1,0})\le\eta.
\]
On success, the prefix through \(L_1^-\) is exactly that of
\(\mathsf D_r\).

\paragraph{The induction on stages.}
For each \(i=1,\ldots,r\), define the following four hybrids in
order. At the start, \(\mathsf J_{i,0}\) has the prefix of
\(\mathsf D_r\) from \Cref{def:qb-ca-round-decoupled}
through \(L_i^-\) on every branch with no failure so far.
More precisely, immediately after \(L_i^-\) and for every successful view \(e\),
the corresponding subnormalized states satisfy
\begin{equation}
  \Pr_{\mathsf J_{i,0}}[E=e]\,
       \rho^{\mathsf J_{i,0},e}_{APBH}
    =\Pr_{\mathsf D_r}[E=e]\,
       \rho^{\mathsf D_r,e}_{APBH}.
  \label{eq:qb-ca-round-prefix-invariant}
\end{equation}
This equality records both the probability of the view and the
conditional state of the registers and tables, with \(h\) determined by \(e\).
Here \(H\) is the classical register containing the original oracle,
and the equality follows from the identical operations defining the
prefixes. The switching lemma bounds the trace distance between the
complete output states after the specified continuation, with the
separate register containing the original oracle excluded from those outputs.

\begin{hybrid}{Hybrid \(\mathsf J_{i,1}\) (independent answers)}
  \item Run \(\mathsf J_{i,0}\) through \(L_i^-\), keeping
  previously failed branches unchanged.
  \item On success, execute \(\cA_i\) and \(L_i^+\) with
  \changed{existing table entries and independent uniform answers
  at fresh inputs}. Append the Alice message and learner records to
  \(E\), storing Alice's private answers in \(P\) and recording
  the learner's answers in \(h\).
  \item Continue the remaining stages with the shared oracle
  \(H[h][P]\). At this step retain the resulting keys even if
  \(L_i^+\) failed.
\end{hybrid}
Applying \Cref{lem:qb-ca-round-invariant} at each successful view gives
\[
  D(\mathsf J_{i,0},\mathsf J_{i,1})
   \le2(a_i+N_i^+)\sqrt{\beta_i}.
\]
Conditioned on each successful view, \(\mathsf J_{i,0}\) and
\(\mathsf J_{i,1}\) run \(\mathsf{SharedAnswers}\) and
\(\mathsf{IndependentAnswers}\), respectively, with
\(\mathcal C=\cA_i+L_i^+\) and the existing private table.
The bound covers all remaining stages, during which the oracle retains
every answer obtained by Alice and the learner. On previously successful
branches, the prefix
through \(L_i^+\) now has its distribution in \(\mathsf D_r\),
including its probabilities and conditional states.

\begin{hybrid}{Hybrid \(\mathsf J_{i,2}\) (failure after Alice's stage)}
  \item Run \(\mathsf J_{i,1}\), keeping all its records.
  \item \changed{If all earlier learners succeeded and \(L_i^+\)
  failed, replace both recorded keys by \(\bot\).}
\end{hybrid}
This learner now has the distribution specified by the reference
experiment in \Cref{def:qb-ca-round-decoupled}, so its conditional
failure probability is at most \(\eta\), giving
\[
  D(\mathsf J_{i,1},\mathsf J_{i,2})\le\eta.
\]
On a remaining successful view \(e\), the prefix through
\(L_i^+\) now has exactly its reference distribution, so the second
trace-distance guarantee of \Cref{lem:qb-ca-round-invariant} applies
to the next Bob phase.

\begin{hybrid}{Hybrid \(\mathsf J_{i,3}\) (removing private answers)}
  \item Run \(\mathsf J_{i,2}\) through \(L_i^+\), keeping
  all failed branches unchanged.
  \item On success, answer queries during \(\cB_{i+1}\) and,
  if \(i<r\), \(L_{i+1}^-\), with \changed{\(H[h]\)} in
  place of \(H[h][P]\). Use the same learner rule. Record each
  learner answer in \(h\) and delete any private entry at that
  input, retaining the remaining \(P\).
  \item Continue the remaining stages with the shared oracle
  \(H[h][P]\). For now retain the resulting keys if
  \(L_{i+1}^-\) failed.
\end{hybrid}
The changed phase includes both Bob's quantum queries and the next
learner's classical queries, and follows \(\mathsf{UsePrivate}\)
in \(\mathsf J_{i,2}\) and \(\mathsf{IgnorePrivate}\) in
\(\mathsf J_{i,3}\), conditioned on each successful view.
The second guarantee of \Cref{lem:qb-ca-round-invariant}, which
already includes the following learner's queries, gives
\[
  D(\mathsf J_{i,2},\mathsf J_{i,3})
   \le2b_{i+1}\sqrt{\gamma_i}+N_{i+1}^-\gamma_i.
\]
A learner query at an input still in \(P\) can distinguish the two
oracles, so its contribution must be included. We retain Alice's table
and state when deriving this bound, which also covers all remaining
stages with the reprogrammed oracle. The changed phase now
follows \(\mathsf D_r\).

\begin{hybrid}{Hybrid \(\mathsf J_{i,4}\) (failure after Bob's stage)}
  \item Run \(\mathsf J_{i,3}\), keeping all its records.
  \item If \(i<r\), \changed{replace both recorded keys by
  \(\bot\) when all preceding learners succeeded and
  \(L_{i+1}^-\) failed}. For \(i=r\), make no change.
\end{hybrid}
For \(i<r\), the prefix through this learner now agrees with
\(\mathsf D_r\) from \Cref{def:qb-ca-round-decoupled}, so
\[
  D(\mathsf J_{i,3},\mathsf J_{i,4})\le\eta.
\]
On success, its prefix through \(L_{i+1}^-\) coincides with that
of \(\mathsf D_r\).
For \(i<r\), setting \(\mathsf J_{i+1,0}=\mathsf J_{i,4}\) establishes
\eqref{eq:qb-ca-round-prefix-invariant} for the next stage.
For \(i=r\), there is no further learner and no such failure cost.

We average each conditional distance bound over the common distribution
of views at which the preceding learners succeeded, whose total
probability is at most one. There are \(r\) learners of each type,
and we apply each failure bound only after the corresponding prefix
agrees with the reference experiment of \Cref{def:qb-ca-round-decoupled}.
Summing all the bounds gives
\[
  D(\rho_{K_{\cA}K_{\cB}E},\sigma_{K_{\cA}K_{\cB}E})
  \le2r\eta+\sum_{i=1}^r\left(
      2(a_i+N_i^+)\sqrt{\beta_i}
      +2b_{i+1}\sqrt{\gamma_i}+N_{i+1}^-\gamma_i\right)
  =\Delta_r,
\]
where \(\rho\) is the output of \(\mathsf J_0\) and \(\sigma\)
is the output of \(\mathsf J_{r,4}\).

On successful views the final hybrid follows \(\mathsf D_r\) from
\Cref{def:qb-ca-round-decoupled}, with Bob using \(H[h]\) after
\(L_r^+\), so his final computation preserves conditional
independence from Alice, as do any final classical message and the
parties' local oracle-free output maps. On failed views, later messages
may differ from those in \(\mathsf D_r\) because each failed branch
was kept from an earlier hybrid, but both recorded keys are \(\bot\),
so their conditional product is preserved.

The finite final experiment and its probabilities are computable from
the protocol description. Applying \Cref{lem:qb-ca-recovery} and
identifying \(\rho\) with the joint distribution of the honest keys
and Eve's final record in the key-recovery experiment gives
\[
  \Adv^{\mathsf{KR}}_{\Pi^{\mathsf{qb}},\cE}(\lambda)
    =\Pr_\rho[R(E)=K_{\cB}\ne\bot]
    \ge\alpha-2\Delta_r.
\]
The attacker uses only classical computation and makes oracle queries
only in the learners, with total bound \(\sum_i(N_i^-+N_i^+)\).
This concludes the proof of \Cref{thm:qb-ca-rounds}.
\end{proof}

\subsubsection{Choosing Parameters for a Polynomial-Query Attack}

We now choose the parameters in \Cref{thm:qb-ca-rounds} so that the
attack loses at most \(\xi\) in success probability and, for constant
\(r\), its query bound is polynomial in the honest query bound \(Q\)
and \(1/\xi\).

\begin{corollary}[Constant-round impossibility]
\label{cor:qb-ca-constant}
Fix \(r\ge1\) independently of the security parameter \(\lambda\).
For every \(0<\xi<1\), a protocol \(\Pi^{\mathsf{qb}}\) in
\Cref{def:qb-ca-rounds} with \(r\) Alice stages has a passive classical attacker \(\cE\)
with \(\Adv^{\mathsf{KR}}_{\Pi^{\mathsf{qb}},\cE}(\lambda)
\ge\alpha-\xi\).
Writing
\[
  \delta:=\frac{\xi}{8r},\qquad
  s:=\frac{\delta^2}{8(Q+1)},
\]
one sufficient bound on its number of classical queries is
\begin{equation}
  2r+\frac{2rQ}{\delta}\,
       s^{-(10\cdot2^{r-1}-4)}
  =\bigl(\poly(Q,r,1/\xi)\bigr)^{O(2^r)}.
  \label{eq:qb-ca-constant-cost}
\end{equation}
Consequently, for every fixed \(r\), a family with polynomial honest
query bounds and inverse-polynomial valid agreement is not passive
key-recovery secure against classical attackers, and hence is not
secure against quantum attackers either. In particular, this rules out
passive key-recovery security for every correct family with polynomial
honest query bounds in this setting.
\end{corollary}

\begin{proof}
Set \(\eta=\delta\), and choose the thresholds by reverse induction on the round index:
\[
  \gamma_r=s^2,\qquad
  \beta_i=s^2\gamma_i^2\quad(1\le i\le r),\qquad
  \gamma_i=s\beta_{i+1}\quad(1\le i<r).
\]
Induction gives
\[
  \gamma_i=s^{5\cdot2^{r-i}-3},\qquad
  \beta_i=s^{10\cdot2^{r-i}-4}.
\]
All these thresholds lie in \((0,1)\), with
\(\gamma_i\le s^2\) and \(\sqrt{\beta_i}=s\gamma_i\).
Using \(\ceil u\le u+1\), \(a_i,b_{i+1}\le Q\), and
\(\delta<1\), we obtain
\begin{align*}
  2(a_i+N_i^+)\sqrt{\beta_i}
  &\le 2(Q+1)s\gamma_i+\frac{2Qs}{\delta}
   \le\frac{\delta^2}{4}+\frac\delta4
   \le\frac\delta2,\\
  2b_{i+1}\sqrt{\gamma_i}
  &\le 2Qs\le\frac{\delta^2}{4}\le\frac\delta4.
\end{align*}
For \(i<r\),
\[
  N_{i+1}^-\gamma_i
  \le s\beta_{i+1}+\frac{Qs}{\delta}
  \le\frac{\delta(\delta+Q)}{8(Q+1)}
  \le\frac\delta8.
\]
For \(i=r\), this contribution is zero, so
\(\Delta_r\le3r\delta\) and
\(2\Delta_r\le6r\delta=3\xi/4\le\xi\).
The smallest threshold is
\(\beta_1=s^{10\cdot2^{r-1}-4}\), so
\[
  \sum_{i=1}^r(N_i^-+N_i^+)
  \le2r+\frac{2rQ}{\delta\beta_1},
\]
proving \eqref{eq:qb-ca-constant-cost} for every \(0<\xi<1\).
Choosing \(\xi=\alpha/3\) gives key-recovery advantage at least
\(2\alpha/3\), and for fixed \(r\) the query bound is polynomial
in \(Q\) and \(1/\alpha\), violating the security condition in
\Cref{def:qb-ca-rounds} whenever \(\alpha\ge1/\poly(\lambda)\).
The correctness case follows from
\(\alpha\ge1-\epsilon_c(\lambda)\) with negligible
\(\epsilon_c\).
This concludes the proof of \Cref{cor:qb-ca-constant}.
\end{proof}

\paragraph{Remark on growing round counts.}
The query and advantage bounds in
\Cref{thm:qb-ca-rounds,cor:qb-ca-constant} also hold when
\(r=r(\lambda)\) grows with the security parameter, since the same
proofs apply for each \(\lambda\). However, the exponent in
\eqref{eq:qb-ca-constant-cost} grows with \(r\), so these bounds do
not give a polynomial-query guarantee for general protocols with
growing round counts.

\section{Applications}\label{sec:corollaries}

We derive consequences for quantum public-key encryption, distinguishing
games, and constant-round oblivious state preparation.

\subsection{Quantum Public-Key Encryption}
\label{sec:applications-qpke}
The two attacks yield complementary restrictions on QPKE, depending on
whether key generation or encryption has classical oracle access.

\begin{theorem}[Attack on QPKE with classical-query key generation]
\label{thm:qpke}
Let \(\Sigma\) be a QPKE scheme with classical public and secret
keys and classical-query key generation, as in \Cref{def:qpke}.  

For every \(0<\delta<1\), there is
a computationally unbounded quantum IND-CPA attacker \(\cD_\delta\) making
at most
\[
  \ceil{q_{\Gen}/\delta^2}+q_{\Dec}
\]
random-oracle queries and satisfying
\[
  \Adv^{\IND\text{-}\mathrm{CPA}}_{\Sigma,\cD_\delta}(\lambda)
  \ge
  \alpha_\Sigma-\frac12-\delta
  -2(2q_{\Enc}+q_{\Dec})\sqrt\delta.
\]
In particular, if \(\alpha_\Sigma(\lambda)>0\), let
\[
  Q_\Sigma:=\max\{1,2q_{\Enc}+q_{\Dec}\},
  \qquad
  \delta:=\left(\frac{\alpha_\Sigma}{8Q_\Sigma}\right)^2.
\]
Then
\[
  \Adv^{\IND\text{-}\mathrm{CPA}}_{\Sigma,\cD_\delta}(\lambda)
  \ge \frac{47}{64}\alpha_\Sigma-\frac12,
\]
and \(\cD_\delta\) makes at most
\[
  \ceil{4096q_{\Gen}Q_\Sigma^4/\alpha_\Sigma^4}+q_{\Dec}
\]
random-oracle queries. Consequently, a scheme with negligible correctness
error and honest query complexity bounded by \(\poly(\lambda)\) is not
query-bounded IND-CPA secure: for every polynomial \(p\), an attacker making
at most \(\poly(\lambda)\) queries achieves advantage at least
\(1/2-1/p(\lambda)-\negl(\lambda)\), for every message-length function \(\mu\)
bounded above by a polynomial in \(\lambda\).
Encryption, decryption, and the ciphertext may all be quantum.
If key generation and decryption are classical algorithms and the
ciphertext is classical, the attacker can also be classical, with the
same bounds. 
\end{theorem}

\begin{proof}
Define a key-agreement protocol \(\Pi_\Sigma\) from \(\Sigma\) as follows.
Alice's first stage runs
\[
  (\pk,\sk)\leftarrow\Gen^H(1^\lambda)
\]
and outputs \((T_{\Gen},M_1,S_{\cA})=(T_{\Gen},\pk,\sk)\).  On input
\(\pk\), Bob fixes the two distinct messages
\[
  m_0:=0^{\mu(\lambda)},
  \qquad
  m_1:=0^{\mu(\lambda)-1}1,
\]
samples \(b\leftarrow\bits\), runs
\(C\leftarrow\Enc^H(1^\lambda,\pk,m_b)\), sends \(C\), and sets
\(K_{\cB}=b\). Alice's final stage computes
\(\widehat m\leftarrow\Dec^H(1^\lambda,\sk,C)\) and outputs
\[
  K_{\cA}:=
  \begin{cases}
    0 & \text{if }\widehat m=m_0,\\
    1 & \text{if }\widehat m=m_1,\\
    \bot & \text{otherwise.}
  \end{cases}
\]
For \(b\in\bits\), let
\begin{align*}
  p_b:=\Pr\bigl[&\widehat m=m_b:\ \\
    &H\leftarrow\cH_\lambda,\
    (\pk,\sk)\leftarrow\Gen^H(1^\lambda),\\
    &C_b\leftarrow\Enc^H(1^\lambda,\pk,m_b),\
    \widehat m\leftarrow\Dec^H(1^\lambda,\sk,C_b)\bigr].
\end{align*}
Therefore
\[
  q_{\cA,1}=q_{\Gen},\qquad
  q_{\cB}=q_{\Enc},\qquad
  q_{\cA,2}=q_{\Dec},
  \qquad
  \alpha_{\Pi_\Sigma}=\frac{p_0+p_1}{2}\ge\alpha_\Sigma,
\]
where the inequality follows from the definition of the minimum
correct-decryption probability in \Cref{def:qpke}.

This is the standard reduction from QPKE to key agreement used in prior
QROM impossibility work~\cite{BouazizEtAl24,LiEtAl24}. Apply
\Cref{thm:main} to obtain the
intercepting attacker
\(\cE_\delta=(\cE_{1},\cE_{2})\).  The IND-CPA
attacker \(\cD_\delta\) has its first stage run
\[
  Z_{\cD}\leftarrow\cE_{1}^H(1^\lambda,\pk)
\]
and returns \((Z_{\cD},m_0,m_1)\). Its second stage, upon
receiving the challenge ciphertext register \(C\), runs
\(K^*\leftarrow\cE_{2}^H(1^\lambda,Z_{\cD},C)\) and outputs \(K^*\)
when \(K^*\in\bits\), or an independent uniform bit otherwise.  Thus its
success probability satisfies
\[
  \Pr[\IND_{\Sigma,\cD_\delta}(1^\lambda)=1]
  \ge
  \Pr[\IntKR_{\Pi_\Sigma,\cE_\delta}(1^\lambda)=1].
\]
The claimed advantage and query bound follow from \Cref{thm:main}, since
\begin{align*}
  \Adv^{\IND\text{-}\mathrm{CPA}}_{\Sigma,\cD_\delta}(\lambda)
  &\ge
  \Pr[\IND_{\Sigma,\cD_\delta}(1^\lambda)=1]-\frac12\\
  &\ge
  \Pr[\IntKR_{\Pi_\Sigma,\cE_\delta}(1^\lambda)=1]-\frac12.
\end{align*}

For the stated parameter choice,
\[
  2Q_\Sigma\sqrt\delta=\frac{\alpha_\Sigma}{4},
  \qquad
  \delta\le\frac{\alpha_\Sigma}{64},
\]
which gives the factor \(47/64\); the query bound follows by substituting
\(1/\delta^2=4096Q_\Sigma^4/\alpha_\Sigma^4\).  A quantum challenge
ciphertext is transferred to the attacker's second stage, so the reduction
does not copy it. If \(\Gen\) and \(\Dec\) are classical algorithms
and the ciphertext is classical, the conditional Alice view is a
classical distribution, so Eve can sample it and run decryption using
only classical computation and classical oracle queries. The same
attack and bounds therefore hold for a classical adversary, as stated.
\end{proof}

\begin{theorem}[Attack on QPKE with classical-query encryption]
\label{thm:qpke-classical-encryption}
Let \(\Sigma\) be a QPKE scheme with classical public and secret
keys, classical-query encryption, and classical ciphertexts, as in
\Cref{def:qpke-classical-encryption}.  
For every
\(0<\xi<1\), there is a computationally unbounded classical IND-CPA
attacker \(\cD_\xi\) making
\[
  O\!\left(
    \frac{q_{\Gen}(q_{\Enc}+1)^2Q_{\Dec}^4}{\xi^9}
    +\frac{q_{\Enc}Q_{\Dec}^2}{\xi^3}+1\right),
  \qquad Q_{\Dec}:=\max\{1,q_{\Dec}\},
\]
classical random-oracle queries and satisfying
\[
  \Adv^{\IND\text{-}\mathrm{CPA}}_{\Sigma,\cD_\xi}(\lambda)
  \ge 2\alpha_\Sigma-\frac32-\xi.
\]
In particular, if the correctness error \(\epsilon_c\) is negligible
and the honest query bounds are polynomial, the scheme is not
query-bounded IND-CPA secure. For every polynomial \(p\), a classical
attacker making polynomially many queries has advantage at least
\(1/2-1/p(\lambda)-2\epsilon_c(\lambda)\).
\end{theorem}

\begin{proof}
Fix \(m_0:=0^{\mu(\lambda)}\) and
\(m_1:=0^{\mu(\lambda)-1}1\). We construct a protocol of the form
\(\cB_1\to\cA\to\cB_2\) in \Cref{def:qb-ca-two}.
Bob first runs \(\Gen^H\), sends \(\pk\), and retains \(\sk\).
Alice chooses a uniform bit \(b\), sends
\(C\leftarrow\Enc^H(1^\lambda,\pk,m_b)\), and outputs
\(K_{\cA}=b\). Bob then computes
\(\widehat m\leftarrow\Dec^H(1^\lambda,\sk,C)\), outputting
\(K_{\cB}=j\) if \(\widehat m=m_j\) for \(j\in\bits\),
and \(\bot\) otherwise. Both messages are classical, Alice's queries
are classical, and Bob may query in superposition in both stages.
The valid agreement probability \(\alpha\) of this protocol satisfies
\(\alpha=\Pr[K_{\cB}=b]\ge\alpha_\Sigma\), with query bounds
\[
  q_{\cB,1}=q_{\Gen},\qquad q_{\cA}=q_{\Enc},\qquad
  q_{\cB,2}=q_{\Dec}.
\]
By \Cref{cor:qb-ca-two}, there is a passive classical attacker
\(\cE_\xi\) with the stated query bound and
\(\Pr[K^*=K_{\cB}\ne\bot]\ge\alpha-\xi\).

The IND-CPA attacker selects \(m_0,m_1\) and stores \(\pk\) in
its first stage. On receiving the classical challenge \(C\), it runs
\(\cE_\xi^H\) on the transcript \((\pk,C)\), returning its
guess when that guess is a bit and a uniform bit otherwise. It never
needs \(\sk\) or access to decryption. The honest decryption above is
used only in the analysis, and does not change the distribution of
the transcript available to \(\cE_\xi\). Since
\(\Pr[K_{\cB}\ne b]=1-\alpha\),
\begin{align*}
  \Pr[\IND_{\Sigma,\cD_\xi}(1^\lambda)=1]
  &\ge \Pr[K^*=K_{\cB}=b]\\
  &\ge \Pr[K^*=K_{\cB}\ne\bot]-\Pr[K_{\cB}\ne b]\\
  &\ge 2\alpha-1-\xi.
\end{align*}
Subtracting \(1/2\) and using \(\alpha\ge\alpha_\Sigma\)
proves the advantage bound. Finally, choose
\(\xi=\min\{1/2,1/p(\lambda)\}\) and use
\(\alpha_\Sigma\ge1-\epsilon_c\).
\end{proof}

\subsection{Distinguishing Games}
\label{sec:applications-lifting}
The same learning and reprogramming argument used in \Cref{thm:main}
also converts quantum distinguishers into classical ones when both
challenge distributions are sampled using classical oracle queries.
The following statement makes the query and advantage losses explicit.

\begin{theorem}[Lifting for two classical-query samplers]
\label{thm:distinguishing-lifting}
Let \(S_0,S_1\) be samplers that each make at most \(q_S\)
classical queries to \(H\leftarrow\cH_\lambda\) and output a
classical string in the same finite sample space. Their local computation
may be quantum. For an algorithm \(\cD\) whose input is the sampled
string, define
\[
  \Adv^{\mathsf{dist}}_{S_0,S_1,\cD}(\lambda)
  :=\left|
    \Pr[\cD^H(1^\lambda,S_1^H(1^\lambda))=1]
    -\Pr[\cD^H(1^\lambda,S_0^H(1^\lambda))=1]
  \right|.
\]
Each probability samples a uniform \(H\) and uses independent private
randomness for the sampler and the distinguisher.
If \(\cD\) makes at most \(q_D\) quantum queries, then for every
\(0<\delta<1\) there is a computationally unbounded classical
distinguisher \(\cC_\delta\) making at most
\(\ceil{q_S/\delta^2}\) classical queries such that
\[
  \Adv^{\mathsf{dist}}_{S_0,S_1,\cC_\delta}(\lambda)
  \ge \Adv^{\mathsf{dist}}_{S_0,S_1,\cD}(\lambda)
       -2\delta-4q_D\sqrt\delta.
\]
Consequently, for every \(0<\varepsilon<1\), the additive loss can
be made at most \(\varepsilon\) with a classical query bound of
\[
  O\!\left(\frac{q_S\max\{1,q_D\}^4}{\varepsilon^4}+1\right).
\]
\end{theorem}

\begin{proof}
Consider the sampler that privately chooses a uniform bit \(B\), runs
\(S_B^H\), and outputs its classical sample \(Z\). Let \(T\)
be its ordered query--answer transcript. Apply \(\Learn_\delta\)
from \Cref{alg:learn} with first message \(Z\), query bound \(q_S\),
and probabilities taken over this mixture of the two samplers. Thus the
learner does not know \(B\). By \Cref{lem:learn-success}, it uses at
most \(\ceil{q_S/\delta^2}\) queries and fails with probability at
most \(\delta\). On success, every input outside its returned table
\(h\) has conditional probability at most \(\delta\) of occurring
in \(T\).

On input \(z\), the classical distinguisher \(\cC_\delta\) runs
this learner using the actual oracle. If learning fails, it outputs a
uniform bit. Otherwise, it samples a full function
\(F\leftarrow\Unif(\cH(h))\), where
\(\cH(h)=\{F:h\preceq F\}\), and classically simulates
\(\cD^F(1^\lambda,z)\) to sample its output bit. Once \(h\) is
fixed, generating \(F\) and simulating the quantum circuit use only
unbounded classical computation and no further queries to the actual
oracle.

To bound the change in the output distribution, consider running
\(\cD^H\) on the real sample with the learner inserted before \(\cD\).
Because the learner's queries only read the fixed oracle and \(\cD\)
is not given its answers, this insertion leaves \(\cD\)'s output
distribution unchanged. Fix a successful reachable
view \((z,h)\). By the deterministic-output property of the learner,
conditioning on this view is equivalent to conditioning on
\(Z=z\) and \(h\preceq H\). Under this conditioning, the joint
law of \((B,T,H)\) can be generated by sampling \((B,T)\) from its
conditional distribution, independently sampling
\(F\leftarrow\Unif(\cH(h))\), and setting \(H=F[T]\).
Indeed, for each fixed \((B,T,Z)\), the sampler's execution depends
only on the answers in \(T\), so all remaining oracle values are
uniform subject to \(h\). This remains true if the sampler's local
computation is quantum, since its oracle access and output are classical.

The two experiments now differ only in whether \(\cD\) uses
\(F[T]\) or \(F\). Apply \Cref{lem:reprogramming} to the partial
function \(T\) restricted to \(\Dom(T)\setminus\Dom(h)\).
Its sampler draws the pair \((B,T)\) conditioned on \((z,h)\),
independently of \(F\). After \(\cD\) finishes, the revealed
sampling coins determine \(B\), so the final test may check whether
\(\cD\)'s output equals \(B\). Since each unlearned input occurs
in \(T\) with conditional probability at most \(\delta\), the
reprogramming lemma bounds the change in this test's acceptance
probability by \(2q_D\sqrt\delta\). Averaging over successful views
and adding at most \(\delta\) for the failure probability gives
\[
  \left|
    \Pr[\cD^H(1^\lambda,Z)=B]
    -\Pr[\cC_\delta^H(1^\lambda,Z)=B]
  \right|
  \le\delta+2q_D\sqrt\delta.
\]
For either algorithm \(\mathcal A\), its distinguishing advantage is
\(2|\Pr[\mathcal A^H(1^\lambda,Z)=B]-1/2|\). The reverse
triangle inequality therefore gives the claimed advantage bound.
Finally, set
\(\delta=(\varepsilon/(8\max\{1,q_D\}))^2\). Then
\(2\delta+4q_D\sqrt\delta\le\varepsilon\), and substituting
in the learner's query bound proves the last claim.
\end{proof}

\subsection{Oblivious State Preparation}
\label{sec:applications-osp}
We combine the constant-round attack of \Cref{cor:qb-ca-constant}
with the established OSP-to-key-agreement reduction, rather than with
the distinguishing-game lifting theorem.

\begin{theorem}[Impossibility of constant-round OSP]
\label{thm:osp-constant}
For every fixed number of rounds, no oblivious state preparation protocol
in the QROM with polynomial honest query complexity and negligible
correctness error is secure against computationally unbounded quantum
receivers making polynomially many oracle queries. The sender is
classical and has classical oracle access, all communication is classical,
and the honest receiver may make quantum oracle queries in every round.
\end{theorem}

\begin{proof}

Use the OSP-to-key-agreement reduction of Bartusek and
Khurana~\cite[Theorem~7.5 and proof of Corollary~7.7]{BartusekKhurana24}.
It gives key agreement with a classical party and classical communication,
using polynomially many OSP instances in parallel within a constant
number of phases, followed by a final classical message. Hence
constant-round OSP yields constant-round key agreement with polynomial
honest query complexity, while negligible OSP error remains negligible
through the construction.

The security reductions use the adversary as a subroutine and add only
polynomially many oracle queries, so they also apply to computationally
unbounded adversaries in our query model. Thus OSP security as stated
would imply passive key-recovery security for the resulting protocol
against polynomial-query classical eavesdroppers.
However, \Cref{cor:qb-ca-constant} with \(\xi=1/4\) gives such an
eavesdropper with recovery probability at least
\(3/4-\negl(\lambda)\), contradicting that security. The use of
perfect correctness in the earlier impossibility was needed for its
key-agreement attack, which we have replaced by one tolerating
negligible correctness error.
This concludes the proof of \Cref{thm:osp-constant}.
\end{proof}

\bibliographystyle{alpha}
\bibliography{quantum}

\end{document}